\documentclass[12pt]{article}
\usepackage[latin9]{luainputenc}
\usepackage{geometry}
\usepackage{array}
\usepackage{float}
\usepackage{ifsym}
\usepackage{multirow}
\usepackage{amsmath}
\usepackage{amsthm}
\usepackage{amssymb}
\usepackage{stmaryrd}
\usepackage{graphicx}
\usepackage{setspace}
\usepackage[authoryear]{natbib}
\usepackage[unicode=true,
 bookmarks=false,
 breaklinks=false,pdfborder={0 0 1},backref=section,colorlinks=false]
 {hyperref}
\hypersetup{
 colorlinks,citecolor=blue,pdftex}

\makeatletter

\providecommand{\tabularnewline}{\\}

\usepackage{amsfonts}
\usepackage{amsthm}
\usepackage{lscape}
\usepackage{appendix}
\usepackage{dcolumn}
\usepackage{rotating}
\usepackage{comment}
\usepackage{color}
\usepackage{amsmath}

\usepackage{pgf}
\usepackage{tikz}\usepackage{caption}
\usepackage{subcaption}
\usepackage{tkz-tab}
\usepackage{tikz-3dplot}
\usetikzlibrary{calc}
\usetikzlibrary{arrows, automata}

\newcolumntype{d}[1]{D{.}{.}{#1}}
\newcolumntype{t}[1]{D{,}{,}{#1}}
\newcolumntype{i}[1]{D{.}{}{#1}}
\newtheorem{theorem}{Theorem}[section]

\newtheorem{corollary}{Corollary}[section]

\newtheorem{definition}{Definition}[section]
\newtheorem{example}{Example}
\newtheorem{lemma}{Lemma}[section]

\newtheorem{proposition}{Proposition}[section]
\newtheorem{remark}{Remark}[section]

\theoremstyle{plain} 

\newtheorem*{asSEL}{Assumption SEL}
\newtheorem*{asEX}{Assumption EX}

\newtheorem*{asR}{Assumption R}
\newtheorem*{asU}{Assumption U}
\newtheorem*{asU2}{Assumption U$^{*}$}
\newtheorem*{asU3}{Assumption U$^{0}$}

\newtheorem*{asM}{Assumption M}
\newtheorem*{asC}{Assumption C}

\newtheorem*{asMTS}{Assumption MTS}

\newtheorem*{asMAT}{Assumption MAT}

\newtheorem*{asEC}{Assumption EC}

\newtheorem*{asI}{Assumptions I}

\numberwithin{equation}{section}

\makeatother

\begin{document}

\title{A Computational Approach to Identification of Treatment Effects for
Policy Evaluation\thanks{The authors are grateful to Jason Abrevaya, Brendan Kline, Xun Tang,
Alex Torgovitsky, Ed Vytlacil, Haiqing Xu, and participants in the
Royal Economic Society 2021 Annual Conference, the 2021 North American
Summer Meeting, the 2021 Asian Meeting, the 2020 North American Winter
Meeting of the Econometric Society, the 2020 Texas Econometrics Camp,
and the workshop at UT Austin for helpful comments and discussions.
We also thank Elie Tamer, the Associate Editor, and the three anonymous
referees for valuable suggestions.}}

\author{Sukjin Han\\
 School of Economics\\
 University of Bristol\\
 \href{mailto:sukjin.han\%40gmail.com}{sukjin.han@gmail.com}\and
Shenshen Yang\\
Ma Yinchu School of Economics \\
 Tianjin University \\
\href{mailto: shenshenyang@tju.edu.cn}{ shenshenyang@tju.edu.cn}}

\date{This Draft: July 31, 2023}
\maketitle
\begin{abstract}
For counterfactual policy evaluation, it is important to ensure that
treatment parameters are relevant to policies in question. This is
especially challenging under unobserved heterogeneity, as is well
featured in the definition of the local average treatment effect (LATE).
Being intrinsically local, the LATE is known to lack external validity
in counterfactual environments. This paper investigates the possibility
of extrapolating local treatment effects to different counterfactual
settings when instrumental variables are only binary. We propose a
novel framework to systematically calculate sharp nonparametric bounds
on various policy-relevant treatment parameters that are defined as
weighted averages of the marginal treatment effect (MTE). Our framework
is flexible enough to fully incorporate statistical independence (rather
than mean independence) of instruments and a large menu of identifying
assumptions beyond the shape restrictions on the MTE that have been
considered in prior studies. We apply our method to understand the
effects of medical insurance policies on the use of medical services.

\vspace{0.1in}

\noindent \textit{JEL Numbers:} C14, C32, C33, C36

\noindent \textit{Keywords:} Heterogeneous treatment effects, local
average treatment effects, marginal treatment effects, extrapolation,
partial identification, linear programming.
\end{abstract}

\section{Introduction\label{sec:Introduction}}

For counterfactual policy evaluation, it is important to ensure that
treatment parameters are relevant to the policies in question. This
is especially challenging in the presence of unobserved heterogeneity.
This challenge is well featured in the definition of the local average
treatment effect (LATE). The LATE has been one of the most popular
treatment parameters used by empirical researchers since it was introduced
by \citet{imbens1994identification}. It induces a straightforward
linear estimation method that requires only a binary instrumental
variable (IV), and yet, allows for unrestricted treatment heterogeneity.
The unfortunate feature of the LATE is that, as the name suggests,
the parameter is intrinsically local, recovering the average treatment
effect (ATE) for a specific subgroup of population called compliers.
This feature leads to two major challenges in making the LATE a reliable
parameter for counterfactual policy evaluation. First, the subpopulation
for which the effect is measured (e.g., via randomized experiments)
may not be the population of policy interest. Second, the definition
of the subpopulation depends on the IV chosen, rendering the parameter
even more difficult to extrapolate to targeted environments.

Dealing with the lack of external validity of the LATE has been an
important theme in the literature. One approach in theoretical work
(\citet{angrist2010extrapolate,bertanha2019external}) and empirical
research (\citet{dehejia2019local,muralidharan2019disrupting}) has
been to show the similarity between complier and non-complier groups
based on observables. This approach, however, cannot attend to possible
unobservable discrepancies between these groups. \citet{heckman2005structural}
unify well-known treatment parameters by expressing them as weighted
averages of what they define as the marginal treatment effect (MTE).
This MTE framework has a great potential for extrapolation because
a class of treatment parameters that are policy-relevant can also
be generated as weighted averages of the MTE.\footnote{See \citet{heckman2010building} for elaboration of this point.}
The only obstacle is that the MTE is identified via a method called
local IV (\citet{heckman1999local}), which requires the continuous
variation of the IV that is sometime large depending on the target
parameter. This in turn reflects the intrinsic difficulty of extrapolation
when available exogenous variation is only discrete. Acknowledging
this nature of the challenge, previous studies in the literature have
proposed imposing shape restrictions on the MTE, which is a function
of the treatment-selection unobservable, while allowing for binary
instruments in the framework of \citet{heckman2005structural}. \citet{brinch2017beyond}
introduce shape restrictions (e.g., linearity) on the MTE functions
in an attempt to identify the LATE extrapolated to different subpopulations
or to test for its external validity. In interesting recent work,
\citet{mogstad2018using} propose a general partial identification
framework where bounds on various policy-relevant treatment parameters
can be obtained from a set of ``IV-like estimands'' that are directly
identified from the data and routinely obtained in empirical work.
\citet{kowalski2021reconciling} applies an approach similar to these
studies to extrapolate the results from one health insurance experiment
to an external setting.

This paper continues this pursuit and investigates the possibility
of extrapolating local treatment parameters to different policy settings
in the MTE framework when IVs are possibly only binary. We propose
a computational approach to calculate sharp nonparametric bounds on
various extrapolated treatment parameters for discrete and continuous
outcomes. We use IVs that satisfy the statistical independence assumption
conditional on covariates. The parameters are defined as weighted
averages of the MTE. Examples include the ATE, the treatment effect
on the treated, the LATE for subgroups induced by new policies, and
the policy-relevant treatment effect (PRTE). We also show how to place
in this procedure restrictions from a large menu of identifying assumptions
beyond the shape restrictions considered in earlier work. 

In this paper, we make four main contributions. First, we propose
a novel framework for systematically calculating bounds on policy-relevant
treatment parameters. We introduce the distribution of the latent
state of the outcome-generating process conditional on the treatment-selection
unobservable. This latent conditional distribution is the key ingredient
for our analysis, as both the target parameter and the distribution
of the observables can be written as linear functionals of it. Because
the latent distribution is a fundamental quantity in the data-generating
process, it is convenient to impose identifying assumptions. Having
the latent distribution as a decision variable, we can formulate infinite-dimensional
linear programming (LP) that produces bounds on a targeted treatment
parameter. Our approach is reminiscent of \citet{balke1997bounds}
and can be viewed as its generalization to the MTE framework. \citet{balke1997bounds}
characterize bounds on the ATE using a binary outcome, treatment and
instrument by introducing a LP approach with the latent response vector
as the decision variable. The main distinction of our approach is
that the latent distribution is conditioned on the selection unobservable,
which makes the program infinite-dimensional, but is important for
our extrapolation purpose. We also allow for both discrete and continuous
$Y$. To make it feasible to solve the resulting infinite-dimensional
program, we use a sieve-like approximation of the program and produce
a finite-dimensional LP. We also develop a method to rescale the LP
to resolve computational issues that arise with a large sieve dimension.

The use of approximation to construct an LP is similar to \citet{mogstad2018using}'s
approach. However, the approach we take differs from theirs in the
following way. While they use the MTE function (more precisely, each
term in the MTE) as the main ingredient to relate IV-like estimands
to target parameters, we use the latent distribution as our main building
block to relate the full distribution of the data to target parameters.
The main consequence of this difference is that we can exhaust the
identifying power of statistical independence of IVs, while their
approach can exploit mean independence. The two approaches are complementary.
For example, when IVs are generated from randomized experiments, one
can comfortably assume full independence, in which case our approach
can be applied to enjoy the tighter bounds than those under mean independence.
This can be useful when the external validity of experimental results
is in question, making the extrapolation of the LATE desirable. 

Second, we introduce identifying assumptions that have not been used
in the context of the MTE framework or the LATE extrapolation. They
include assumptions that there exist exogenous variables other than
IVs. One of the main messages we hope to deliver in this paper is
that, given the challenge of extrapolation, additional exogenous variation
can be useful to conduct informative policy evaluation. We propose
two types of exogenous variables that have been used in the literature
in the context of identifying the ATE: \citet{SV11}, \citet{mourifie2015sharp},
\citet{han2017identification}, \citet{vuong2017counterfactual},
and \citet{han2019estimation} use the first type (entering the outcome
and selection equations), and \citet{VY07}, \citet{liu2020two},
and \citet{balat2020multiple} use the second type (only entering
the outcome equation). For example, the existence of the second type
can be plausible when the agent has imperfect foresight when making
the treatment selection decision. We utilize these variables in the
context of the MTE framework. Moreover, while the existing papers
on the ATE make use of these variables in combination with rank similarity,
rank invariance, or additive separability, we show that they independently
have identifying power for treatment parameters, including the ATE.
In general, it may not be always easy to find such exogenous variables.
But when the researcher does find it, it can be a more reliable source
of identification than assumptions on counterfactual quantities, as
the identifying power comes from the data rather than the researcher's
prior.

We also propose identifying assumptions that restrict treatment effect
heterogeneity. In particular, we propose a range of uniformity assumptions
that relate to rank similarity or rank invariance (\citet{chernozhukov2005iv})
and monotone treatment response in \citet{MP00}, including a novel
identifying assumption, called \textit{rank dominance}. The direction
of endogeneity can also be incorporated in this MTE framework. This
assumption is sometimes imposed in empirical work to characterize
selection bias and has been shown to have identifying power for the
ATE (\citet{MP00}).

Third, we show that our approach yields straightforward proof of the
sharpness of the resulting bounds, no matter whether the outcome is
discrete or continuous and whether additional identifying assumptions
are imposed or not. This feature stems from the use of the latent
conditional distribution in the linear programming and the convexity
of the feasible set in the program. When the MTE itself is the target
parameter, we distinguish between the notions of pointwise and uniform
sharpness and argue why uniform sharpness is often difficult to achieve.

Fourth, as an application, we study the effects of insurance on medical
service utilization by considering various counterfactual policies
related to insurance coverage. The LATE for compliers and the bounds
on the LATE for always-takers and never-takers reveal that possessing
private insurance tend to have the largest effect on medical visits
for never-takers, i.e., those who face higher insurance cost. This
provides a policy implication that lowering the cost of private insurance
is important, because the high cost might hinder people with most
need from receiving adequate medical services.

The linear programming approach to partial identification of treatment
effects was pioneered by \citet{balke1997bounds} and recently gained
attention in the literature; see, e.g., \citet{Chi10}, \citet{mogstad2018using},
\citet{machado2018instrumental}, \citet{kamat2017identification},
\citet{gunsilius2019bounds}, \citet{han2019optimal}, \citet{russell2021sharp},
\citet{han2023quantile}.\footnote{There are also studies that use linear programming for partial identification
of parameters that are not necessarily treatment effects; see e.g.,
\citet{honore2006bounds}, \citet{honore2006bounds2}, \citet{freyberger2015identification},
\citet{torgovitsky2019partial}, and \citet{gu2022partial}.} As these papers suggest, there are many settings, including ours,
where analytical derivation of bounds is cumbersome or nearly impossible
due to the complexity of the problems. Also, the computational approach
can streamline the sensitivity analysis of a researcher without needing
to analytically derive bounds and prove their sharpness whenever changing
the set of identifying assumptions.

As concurrent work to ours, \citet{marx2020sharp} also considers
partial identification of policy-relevant treatment parameters in
the MTE framework. In his paper, sharp analytical bounds are derived
for treatment parameters for the subset of compliers, and the identifying
power of rank similarity and covariates is explored for general treatment
parameters. The current paper is similar to his in that we also fully
exhaust statistical independence (rather than mean independence) and
produce sharp bounds. However, our approach differs in a few important
ways. First, we provide a computational framework that enables the
systematic calculation of bounds. Second, with the computational approach,
we produce bounds for a range of treatment parameters under various
identifying assumptions that have not been previously explored in
this context. Also, the computational approach makes it convenient
to conduct sensitivity analyses with different sets of assumptions.

This paper proceeds as follows. The next section introduces the main
observables, maintained assumptions, and parameters of interest. Section
\ref{sec:Distribution-of-State} defines the latent conditional probability
and formulates the infinite-dimensional LP, and Section \ref{sec:Sieve-Approximation-and}
introduces sieve approximation to the program. Section \ref{sec:Identifying-Power-of}
establishes the connection between this paper and \citet{mogstad2018using}.
Section \ref{sec:Additional-Assumptions} introduces additional identifying
assumptions and shows how they can easily be incorporated in the LP.
So far, the analysis is given with discrete $Y$, which is extended
to the case with continuous $Y$ in Section \ref{sec:Extension:-Continuous}.
Section \ref{sec:Simulation} provides numerical illustrations, and
Section \ref{sec:Empirical-Application} contains an empirical application.
In the Appendix, Section \ref{sec:Examples-of-the} lists other examples
of target parameters. Section \ref{sec:Discussions} discusses (i)
rescaling of the LP, (ii) the pointwise and uniform sharpness for
the MTE bounds, (iii) the extension with continuous covariates, and
(iv) estimation and inference. All proofs are contained in Section
\ref{sec:Proofs}. Additional numerical results can be found in Section
\ref{sec:Additional-Numerical-Exercise}.

\section{Assumptions and Target Parameters\label{sec:Preliminaries:-Observables,-Assu}}

Assume that we observe a discrete or continuous outcome $Y\in\mathcal{Y}\subseteq\mathbb{R}$,
binary treatment $D\in\{0,1\}$, and possibly discrete instrument
$Z\in\mathcal{Z}\subseteq\mathbb{R}$. The leading case is binary
$Z$, which is common especially in randomized experiments. We may
additionally observe an exogenous variable $W\in\mathcal{W}\subseteq\mathbb{R}$
and possibly endogenous covariates $X\in\mathcal{X}\subseteq\mathbb{R}^{d_{X}}$.

Let $Y(d)$ be the counterfactual outcome given $d$ and $Y(d,w)$
be the extended counterfactual outcome given $(d,w)$, which are consistent
with the observed outcome: $Y=\sum_{d\in\{0,1\}}1\{D=d\}Y(d)=\sum_{d\in\{0,1\},w\in\mathcal{W}}1\{D=d,W=w\}Y(d,w)$.

\begin{asEX}For given $(d,w)\in\{0,1\}\times\mathcal{W}$, $(Y(d,w),U)\perp(Z,W)|X$.\end{asEX}

When there is no $W$ this assumption and all below are understood
as $W$ being degenerate. Assumption EX imposes the exclusion restriction
and conditional statistical independence for $Z$ (and $W$). One
of the contributions of this paper is to propose a framework that
can make use of full independence instead of mean independence (\citet{mogstad2018using})
and show the identifying power of the former relative to the latter.
This feature arises regardless of the existence of $W$. We formally
discuss the identifying power of full independence relative to mean
independence in Section \ref{sec:Identifying-Power-of}. Also, note
that Assumption EX imposes marginal independence rather than joint
independence of $\{Y(d,w),U\}_{(d,w)\in\{0,1\}\times\mathcal{W}}\perp(Z,W)|X$.

We consider two different scenarios related to $W$: (a) $W$ directly
affects $Y$ but not $D$ and (b) $W$ directly affects both $Y$
and $D$. Accordingly, we maintain the following assumptions.

\begin{asSEL}(a) $D=1\{U\le P(Z,X)\}$ where $P(Z,X)\equiv\Pr[D=1|Z,X]$;
\\
(b) $D=1\{U\le P(Z,X,W)\}$ where $P(Z,X,W)\equiv\Pr[D=1|Z,X,W]$.\end{asSEL}

We introduce $W$ as an additional exogenous variable researchers
may be equipped with in addition to the instrument $Z$. Given the
challenge of extrapolation with minimal variation in $Z$, it would
be important to search for additional exogenous variables. In the
case of (a), such variables can be motivated by exogenous shocks that
agents cannot fully anticipate at the time of making treatment choices.
For example, let $Y$ be the earning and $D$ be the college attendance.
In this example, $W$ can be a randomized job training program that
directly affects $Y$ but whose lottery outcome cannot be foreseen
when making the college decision. As another example, when $Y$ is
the health outcome and $D$ is getting an insurance, $W$ can be random
health or policy shocks that cannot be fully anticipated when making
the insurance decision. In the context of generalized Roy models,
(a) is consistent with agents' limited information when comparing
the potential outcomes of the two treatment states. We show the identifying
power of $W$ even with its minimal variation. This is the first paper
that formally uses this type of variable for identification in the
MTE framework.\footnote{Relatedly, \citet{eisenhauer2015generalized} allows variables of
type (a) in the context of generalized Roy models. However, they consider
these variables only as a feature of agent's limited information but
\emph{not} as a source of identification. Specifically, their identification
of the MTE and cost parameters only relies on the exogenous variables
that are known to the agent at the time of selection into treatment
(i.e., using their notation, they rely on $Z$ to identify the MTE
and $X_{I}$ to identify the cost parameters); see pp. 430-431 of
their paper.} Note that the requirement of reverse exclusion of $W$ in (a) can
be tested from the data by inspecting whether $\Pr[D=1|Z,X,W]=\Pr[D=1|Z,X]$.

Assumption SEL imposes a selection model for $D$, which is important
in motivating and interpreting marginal treatment effects later. This
assumption is also equivalent to \citet{imbens1994identification}'s
monotonicity assumption (\citet{vytlacil2002independence}). We introduce
the standard normalization that $U\sim Unif[0,1]$ conditional on
$X=x$.\footnote{Note that for any index function $g(z,x)$ and an unobservable $\varepsilon$
with any distribution, the selection model satisfies $D=1\{\varepsilon\le g(Z,X)\}=1\{F_{\varepsilon|X}(\varepsilon|X)\le F_{\varepsilon|X}(g(Z,X)|X)\}=1\{U\le P(Z,X)\}$,
since $P(z,x)=\Pr[\varepsilon\le g(z,x)|X=x]=\Pr[U\le F_{\varepsilon|X}(g(z,x)|x)|X=x]=F_{\varepsilon|X}(g(z,x)|x)$
and $F_{\varepsilon|X}(\varepsilon|X)=U$ is uniformly distributed
conditional on $X$.} 

In Assumption SEL, Case (a) is where $W$ is a reversely excluded
exogenous variable, which we call \textit{reverse IV}. This type of
exogenous variables was considered by \citet{VY07}, \citet{balat2020multiple},
and \citet{liu2020two}. In Case (b), we show that a reverse IV is
not necessary, and $W$ can be present in the selection equation.
This type of exogenous variables was considered by \citet{SV11},
\citet{mourifie2015sharp}, \citet{han2017identification}, \citet{vuong2017counterfactual},
and \citet{han2019estimation}. In both scenarios, however, we show
that we can use $W$ for identification without necessarily invoking
rank similarity, rank invariance, or additive separability in contrast
to the above studies. We show this is possible due to the computational
approach we take. Below, we combine the existence of $W$ with assumptions
that are related to rank similarity. Another distinct feature of our
approach in comparison to the prior studies is that we consider a
broad class of the generalized LATEs as our target parameter, including
the ATE considered in those studies. For notational simplicity, we
focus on Case (a) henceforth; it is straightforward to draw analogous
results for Case (b).

We aim to establish sharp bounds on various treatment parameters.
Following \citet{heckman2005structural}, we express treatment parameters
as integral equations of the MTE. The MTE is defined in our setting
as
\begin{align*}
E[Y(1)-Y(0)|U=u,X=x],
\end{align*}
where $Y(d)=Y(d,W)$. Similar to \citet{mogstad2018using}, it is
convenient to introduce the marginal treatment response (MTR) function
\begin{align*}
m_{d}(u,w,x) & \equiv E[Y(d,w)|U=u,X=x]
\end{align*}
where $W$ does not appear as a conditioning variable due to Assumption
EX. Now, we define the target parameter $\tau$ to be the difference
of the weighted averages of the MTRs:
\begin{align}
\tau & =E[\tau_{1}(Z,W,X)-\tau_{0}(Z,W,X)],\label{eq:target_para}
\end{align}
where
\begin{align}
\tau_{d}(z,w,x) & =\int m_{d}(u,w,x)\omega_{d}(u,z,x)du\label{eq:tau_zx}
\end{align}
by using $F_{U|X}(u|x)=u$, and $\omega_{d}(u,z,x)$ is a known weight
specific to the parameter of interest. This definition agrees with
the insight of \citet{heckman2005structural}. The target parameter
includes a wide range of policy-relevant treatment parameters. We
list a few examples of the target parameter here; other examples can
be found in Table \ref{tab:GLATE} in the Appendix.

\begin{example}With a Dirac delta function for a given value $u$
as the weight, the MTE itself can be a target parameter.
\begin{align*}
\tau_{MTE} & =E[m_{1}(u,W,X)-m_{0}(u,W,X)]
\end{align*}

\end{example}

\begin{example}\label{ex:ATE}The ATE can be a target parameter with
$\omega_{d}(u,z,x)=1$ for any $(u,z,x)$.

\[
\tau_{ATE}=E\left[\int_{0}^{1}m_{1}(u,W,X)du-\int_{0}^{1}m_{0}(u,W,X)du\right]
\]

\end{example}

\begin{example}\label{ex:LATE}The generalized LATE is also a target
parameter. Suppose we are interested in the LATE for individuals lying
in $[\underline{u},\overline{u}]$. We assign the weight $\omega_{d}(u,z,x)=\frac{1(u\in[\underline{u},\overline{u}])}{\overline{u}-\underline{u}}$
for any $(u,z,x)$ , where the numerator excludes people outside this
range and the denominator gives a weight to people within $[\underline{u},\overline{u}]$
according to their fraction in the whole population. The generalized
LATE is expressed as:

\[
\tau_{GLATE}=E\left[\int_{0}^{1}m_{1}(u,W,X)\frac{1(u\in[\underline{u},\overline{u}])}{\overline{u}-\underline{u}}du-\int_{0}^{1}m_{0}(u,W,X)\frac{1(u\in[\underline{u},\overline{u}])}{\overline{u}-\underline{u}}du\right]
\]

\end{example}

\begin{example}\label{ex:PRTE}The policy relevant treatment effect
(PRTE) is a target parameter that is particularly useful for policy
evaluation. It is defined as the welfare difference between two different
policies. Let $Z$ and $Z'$ be two instrument variables under two
policies and $P(Z,X)$ and $P'(Z',X)$ be propensity scores under
the two policies.

\begin{align*}
\tau_{PRTE} & =E\Bigg[\int_{0}^{1}m_{1}(u,W,X)\frac{\Pr\left[u\leq P'(Z',X)\right]-\Pr\left[u\leq P(Z,X)\right]}{E\left[P'(Z',X)\right]-E\left[P(Z,X)\right]}du\\
 & -\int_{0}^{1}m_{0}(u,W,X)\frac{\text{\ensuremath{\Pr}}\left[u\leq P'(Z',X)\right]-\Pr\left[u\leq P(Z,X)\right]}{E\left[P'(Z',X)\right]-E\left[P(Z,X)\right]}du\Bigg]
\end{align*}
\end{example} 

To define a broader class of parameters beyond these examples, the
weights $\omega_{0}$ and $\omega_{1}$ can be set asymmetrically.
All the parameters we consider in this paper can be defined conditional
on $X$ and $W$, although we omit them for succinctness.

Typically, a binary instrument is not sufficient in producing informative
bounds on the target parameters. This is because a binary instrument
has no extrapolative power for general non-compliers, e.g., always-takers
and never-takers, but only identifies the effect for compliers. Prior
studies have tried to overcome this challenge by imposing shape restrictions
on the MTE (\citet{cornelissen2016late}, \citet{brinch2017beyond},
\citet{mogstad2018using}, \citet{kowalski2021reconciling}), although
these restrictions are not always empirically justified. Evidently,
it would be useful to provide empirical researchers with a larger
variety of assumptions so that it is easier to find justifiable assumptions
that suit their specific examples.

The existence of additional exogenous variables embodied in Assumptions
SEL and EX may be appealing as it can be warranted by data with less
arbitrariness. We accompany Assumptions SEL and EX with an assumption
that $W$ and $Z$ are relevant variables, which make the role of
these variables more explicit.

\begin{asR}For given $x\in\mathcal{X}$, (i) $\Pr[Y(d,w)\neq Y(d,w')|X=x]>0$
for some $d$ and $w\neq w'$; (ii) either (a) $P(z,x)\neq P(z',x)$
for $z\neq z'$ and $0<P(z,x)<1$ for all $z$ or (b) $P(z,x,w)\neq P(z',x,w)$
for $z\neq z'$ and $0<P(z,x,w)<1$ for all $(z,w)$.

\end{asR}

Assumption R(i) is a relevance condition for $W$ in determining $Y$.
R(ii) is the standard relevance assumption for the instrument and
the positivity assumption. We later show that under Assumptions SEL,
EX and R, the variation of $W$ (in addition to $Z$) is a useful
source for extrapolation and narrowing the bounds on target parameters.

\section{Distribution of Latent State and Infinite-Dimensional Linear Program\label{sec:Distribution-of-State}}

Our goal is to provide a systematic framework to calculate bounds
on the target parameters, which is easy to incorporate various identifying
assumptions. To this end and as a crucial first step of our analysis,
we define a state variable that determines a specific mapping of
\begin{align}
(d,w) & \mapsto y\label{eq:map}
\end{align}
for discrete $y\in\mathcal{Y}=\{y_{1},...,y_{L}\}$. We discuss the
extension with continuously distributed $Y$ in Section \ref{sec:Extension:-Continuous}.
Given that $d$ is binary and assuming $w$ is also binary, there
are $L^{4}$ possible states or maps from $(d,w)$ onto $y$. Using
the extended counterfactual outcome $Y(d,w)$, define a latent vector
$\epsilon$ as
\begin{align*}
\epsilon & \equiv(Y(0,0),Y(0,1),Y(1,0),Y(1,1))
\end{align*}
and its realized value as $e\equiv(y(0,0),y(0,1),y(1,0),y(1,1))\in\mathcal{E}=\mathcal{Y}^{4}$.
Then, each value of $e$ represents each possible state. Table \ref{tab:16maps}
lists all 16 maps in a leading case of binary $y$ with $\mathcal{Y}=\{0,1\}$.
\begin{table}
\begin{centering}
{\scriptsize{}}%
\begin{tabular}{|c|c|c|c|}
\hline 
{\scriptsize{}\#} & {\scriptsize{}$d$} & {\scriptsize{}$w$} & {\scriptsize{}$Y(d,w)$}\tabularnewline
\hline 
\hline 
\multirow{4}{*}{{\scriptsize{}1}} & {\scriptsize{}0} & {\scriptsize{}0} & {\scriptsize{}0}\tabularnewline
\cline{2-4} 
 & {\scriptsize{}0} & {\scriptsize{}1} & {\scriptsize{}0}\tabularnewline
\cline{2-4} 
 & {\scriptsize{}1} & {\scriptsize{}0} & {\scriptsize{}0}\tabularnewline
\cline{2-4} 
 & {\scriptsize{}1} & {\scriptsize{}1} & {\scriptsize{}0}\tabularnewline
\hline 
\multirow{4}{*}{{\scriptsize{}2}} & {\scriptsize{}0} & {\scriptsize{}0} & {\scriptsize{}1}\tabularnewline
\cline{2-4} 
 & {\scriptsize{}0} & {\scriptsize{}1} & {\scriptsize{}0}\tabularnewline
\cline{2-4} 
 & {\scriptsize{}1} & {\scriptsize{}0} & {\scriptsize{}0}\tabularnewline
\cline{2-4} 
 & {\scriptsize{}1} & {\scriptsize{}1} & {\scriptsize{}0}\tabularnewline
\hline 
\multirow{4}{*}{{\scriptsize{}3}} & {\scriptsize{}0} & {\scriptsize{}0} & {\scriptsize{}0}\tabularnewline
\cline{2-4} 
 & {\scriptsize{}0} & {\scriptsize{}1} & {\scriptsize{}1}\tabularnewline
\cline{2-4} 
 & {\scriptsize{}1} & {\scriptsize{}0} & {\scriptsize{}0}\tabularnewline
\cline{2-4} 
 & {\scriptsize{}1} & {\scriptsize{}1} & {\scriptsize{}0}\tabularnewline
\hline 
\multirow{4}{*}{{\scriptsize{}4}} & {\scriptsize{}0} & {\scriptsize{}0} & {\scriptsize{}1}\tabularnewline
\cline{2-4} 
 & {\scriptsize{}0} & {\scriptsize{}1} & {\scriptsize{}1}\tabularnewline
\cline{2-4} 
 & {\scriptsize{}1} & {\scriptsize{}0} & {\scriptsize{}0}\tabularnewline
\cline{2-4} 
 & {\scriptsize{}1} & {\scriptsize{}1} & {\scriptsize{}0}\tabularnewline
\hline 
\end{tabular}{\scriptsize{} }%
\begin{tabular}{|c|c|c|c|}
\hline 
{\scriptsize{}\#} & {\scriptsize{}$d$} & {\scriptsize{}$w$} & {\scriptsize{}$Y(d,w)$}\tabularnewline
\hline 
\hline 
\multirow{4}{*}{{\scriptsize{}5}} & {\scriptsize{}0} & {\scriptsize{}0} & {\scriptsize{}0}\tabularnewline
\cline{2-4} 
 & {\scriptsize{}0} & {\scriptsize{}1} & {\scriptsize{}0}\tabularnewline
\cline{2-4} 
 & {\scriptsize{}1} & {\scriptsize{}0} & {\scriptsize{}1}\tabularnewline
\cline{2-4} 
 & {\scriptsize{}1} & {\scriptsize{}1} & {\scriptsize{}0}\tabularnewline
\hline 
\multirow{4}{*}{{\scriptsize{}6}} & {\scriptsize{}0} & {\scriptsize{}0} & {\scriptsize{}1}\tabularnewline
\cline{2-4} 
 & {\scriptsize{}0} & {\scriptsize{}1} & {\scriptsize{}0}\tabularnewline
\cline{2-4} 
 & {\scriptsize{}1} & {\scriptsize{}0} & {\scriptsize{}1}\tabularnewline
\cline{2-4} 
 & {\scriptsize{}1} & {\scriptsize{}1} & {\scriptsize{}0}\tabularnewline
\hline 
\multirow{4}{*}{{\scriptsize{}7}} & {\scriptsize{}0} & {\scriptsize{}0} & {\scriptsize{}0}\tabularnewline
\cline{2-4} 
 & {\scriptsize{}0} & {\scriptsize{}1} & {\scriptsize{}1}\tabularnewline
\cline{2-4} 
 & {\scriptsize{}1} & {\scriptsize{}0} & {\scriptsize{}1}\tabularnewline
\cline{2-4} 
 & {\scriptsize{}1} & {\scriptsize{}1} & {\scriptsize{}0}\tabularnewline
\hline 
\multirow{4}{*}{{\scriptsize{}8}} & {\scriptsize{}0} & {\scriptsize{}0} & {\scriptsize{}1}\tabularnewline
\cline{2-4} 
 & {\scriptsize{}0} & {\scriptsize{}1} & {\scriptsize{}1}\tabularnewline
\cline{2-4} 
 & {\scriptsize{}1} & {\scriptsize{}0} & {\scriptsize{}1}\tabularnewline
\cline{2-4} 
 & {\scriptsize{}1} & {\scriptsize{}1} & {\scriptsize{}0}\tabularnewline
\hline 
\end{tabular}{\scriptsize{} }%
\begin{tabular}{|c|c|c|c|}
\hline 
{\scriptsize{}\#} & {\scriptsize{}$d$} & {\scriptsize{}$w$} & {\scriptsize{}$Y(d,w)$}\tabularnewline
\hline 
\hline 
\multirow{4}{*}{{\scriptsize{}9}} & {\scriptsize{}0} & {\scriptsize{}0} & {\scriptsize{}0}\tabularnewline
\cline{2-4} 
 & {\scriptsize{}0} & {\scriptsize{}1} & {\scriptsize{}0}\tabularnewline
\cline{2-4} 
 & {\scriptsize{}1} & {\scriptsize{}0} & {\scriptsize{}0}\tabularnewline
\cline{2-4} 
 & {\scriptsize{}1} & {\scriptsize{}1} & {\scriptsize{}1}\tabularnewline
\hline 
\multirow{4}{*}{{\scriptsize{}10}} & {\scriptsize{}0} & {\scriptsize{}0} & {\scriptsize{}1}\tabularnewline
\cline{2-4} 
 & {\scriptsize{}0} & {\scriptsize{}1} & {\scriptsize{}0}\tabularnewline
\cline{2-4} 
 & {\scriptsize{}1} & {\scriptsize{}0} & {\scriptsize{}0}\tabularnewline
\cline{2-4} 
 & {\scriptsize{}1} & {\scriptsize{}1} & {\scriptsize{}1}\tabularnewline
\hline 
\multirow{4}{*}{{\scriptsize{}11}} & {\scriptsize{}0} & {\scriptsize{}0} & {\scriptsize{}0}\tabularnewline
\cline{2-4} 
 & {\scriptsize{}0} & {\scriptsize{}1} & {\scriptsize{}1}\tabularnewline
\cline{2-4} 
 & {\scriptsize{}1} & {\scriptsize{}0} & {\scriptsize{}0}\tabularnewline
\cline{2-4} 
 & {\scriptsize{}1} & {\scriptsize{}1} & {\scriptsize{}1}\tabularnewline
\hline 
\multirow{4}{*}{{\scriptsize{}12}} & {\scriptsize{}0} & {\scriptsize{}0} & {\scriptsize{}1}\tabularnewline
\cline{2-4} 
 & {\scriptsize{}0} & {\scriptsize{}1} & {\scriptsize{}1}\tabularnewline
\cline{2-4} 
 & {\scriptsize{}1} & {\scriptsize{}0} & {\scriptsize{}0}\tabularnewline
\cline{2-4} 
 & {\scriptsize{}1} & {\scriptsize{}1} & {\scriptsize{}1}\tabularnewline
\hline 
\end{tabular}{\scriptsize{} }%
\begin{tabular}{|c|c|c|c|}
\hline 
{\scriptsize{}\#} & {\scriptsize{}$d$} & {\scriptsize{}$w$} & {\scriptsize{}$Y(d,w)$}\tabularnewline
\hline 
\hline 
\multirow{4}{*}{{\scriptsize{}13}} & {\scriptsize{}0} & {\scriptsize{}0} & {\scriptsize{}0}\tabularnewline
\cline{2-4} 
 & {\scriptsize{}0} & {\scriptsize{}1} & {\scriptsize{}0}\tabularnewline
\cline{2-4} 
 & {\scriptsize{}1} & {\scriptsize{}0} & {\scriptsize{}1}\tabularnewline
\cline{2-4} 
 & {\scriptsize{}1} & {\scriptsize{}1} & {\scriptsize{}1}\tabularnewline
\hline 
\multirow{4}{*}{{\scriptsize{}14}} & {\scriptsize{}0} & {\scriptsize{}0} & {\scriptsize{}1}\tabularnewline
\cline{2-4} 
 & {\scriptsize{}0} & {\scriptsize{}1} & {\scriptsize{}0}\tabularnewline
\cline{2-4} 
 & {\scriptsize{}1} & {\scriptsize{}0} & {\scriptsize{}1}\tabularnewline
\cline{2-4} 
 & {\scriptsize{}1} & {\scriptsize{}1} & {\scriptsize{}1}\tabularnewline
\hline 
\multirow{4}{*}{{\scriptsize{}15}} & {\scriptsize{}0} & {\scriptsize{}0} & {\scriptsize{}0}\tabularnewline
\cline{2-4} 
 & {\scriptsize{}0} & {\scriptsize{}1} & {\scriptsize{}1}\tabularnewline
\cline{2-4} 
 & {\scriptsize{}1} & {\scriptsize{}0} & {\scriptsize{}1}\tabularnewline
\cline{2-4} 
 & {\scriptsize{}1} & {\scriptsize{}1} & {\scriptsize{}1}\tabularnewline
\hline 
\multirow{4}{*}{{\scriptsize{}16}} & {\scriptsize{}0} & {\scriptsize{}0} & {\scriptsize{}1}\tabularnewline
\cline{2-4} 
 & {\scriptsize{}0} & {\scriptsize{}1} & {\scriptsize{}1}\tabularnewline
\cline{2-4} 
 & {\scriptsize{}1} & {\scriptsize{}0} & {\scriptsize{}1}\tabularnewline
\cline{2-4} 
 & {\scriptsize{}1} & {\scriptsize{}1} & {\scriptsize{}1}\tabularnewline
\hline 
\end{tabular}{\scriptsize\par}
\par\end{centering}
\caption{All Possible Maps from $(d,w)$ to $y\in\{0,1\}$}
\label{tab:16maps}
\end{table}

Now, as a key component of our LP, we define the probability mass
function of $\epsilon$ conditional on $(U,X)$: for $e\in\mathcal{E}$,
\begin{align}
q(e|u,x) & \equiv\Pr[\epsilon=e|U=u,X=x]\label{eq:q}
\end{align}
with $\sum_{e\in\mathcal{E}}q(e|u,x)=1$ for any $(u,x)$. The quantity
$q(e|u,x)$ captures the joint distribution of $(\epsilon,U)$ and
thus reflects endogenous treatment selection. It is shown below that
this latent conditional probability is a building block for various
treatment parameters and thus serves as the decision variable in the
LP. The introduction of $q(e|u,x)$ distinguishes our approach from
those in \citet{balke1997bounds} and \citet{mogstad2018using}. Since
the probability is conditional on continuously distributed $U$, the
simple finite-dimensional linear programming approach of \citet{balke1997bounds}
is no longer applicable. Instead, we use an approximation method similar
to \citet{mogstad2018using}. However, \citet{mogstad2018using} uses
the MTR function as a building block for treatment parameters and
introduces the ``IV-like'' estimands as a means of funneling the
information from the data. Unlike in \citet{mogstad2018using}, $q(e|u,x)$
can be directly related to the distribution of data. This allows us
to (i) fully exhaust the full independence assumption, (ii) facilitate
proving sharpness and (iii) incorporating a large menu of additional
identifying assumptions.

In the remaining section and the next section, we focus on binary
$Y$ for simplicity; the extension to general discrete $Y$ is straightforward
and is discussed in Section \ref{sec:Identifying-Power-of}; continuous
$Y$ is considered in Section \ref{sec:Extension:-Continuous}. Also,
we will mostly focus on binary $W$ and discrete $X$ for expositional
simplicity. Section \ref{subsec:Linear-Programming-with} in the Appendix
extends the framework to incorporate continuously distributed $X$;
it is also straightforward to extend to allow for general discrete
variable $W$. By \eqref{eq:q}, note that
\begin{align*}
\Pr[Y(d,w)=1|U=u,X=x] & =\Pr[\epsilon\in\{e\in\mathcal{E}:y(d,w)=1\}|U=u,X=x]\\
 & =\sum_{e\in\mathcal{E}:y(d,w)=1}q(e|u,x).
\end{align*}
Therefore, the MTR can be expressed as
\begin{align}
m_{d}(u,w,x) & =\sum_{e:y(d,w)=1}q(e|u,x).\label{eq:MTR}
\end{align}
Combining \eqref{eq:tau_zx} and \eqref{eq:MTR}, we have $\tau_{d}(z,w,x)=\sum_{e:y(d,w)=1}\int q(e|u,x)\omega_{d}(u,z,x)du$,
and thus the target parameter $\tau=E[\tau_{1}(Z,W,X)]-E[\tau_{0}(Z,W,X)]$
in \eqref{eq:target_para} can be written as
\begin{align}
\tau & =E\left[\sum_{e:y(1,W)=1}\int q(e|u,X)\omega_{1}(u,Z,X)du-\sum_{e:y(0,W)=1}\int q(e|u,X)\omega_{0}(u,Z,X)du\right]\label{eq:target}
\end{align}
for some $q$ that satisfies the properties of probability.

The goal of this paper is to (at least partially) infer the target
parameter $\tau$ based on the data, i.e., the distribution of $(Y,D,Z,W,X)$.
The key insight is that there are observationally equivalent $q(e|u,x)$'s
that are consistent with the data, which in turn produces observationally
equivalent $\tau$'s that define the identified set.

Let $p(y,d|z,w,x)\equiv\Pr[Y=y,D=d|Z=z,W=w,X=x]$ be the observed
conditional probability. This data distribution imposes restrictions
on $q(e|u,x)$. For instance, for $D=1$,
\begin{align*}
p(y,1|z,w,x) & =\Pr[Y(1,w)=y,U\le P(z,x)|Z=z,W=w,X=x]\\
 & =\Pr[Y(1,w)=y,U\le P(z,x)|X=x]
\end{align*}
by Assumption EX, but
\begin{align}
\Pr[Y(1,w)=y,U\le P(z,x)|X=x] & =\int_{0}^{P(z,x)}\Pr[Y(1,w)=y|U=u,X=x]du\nonumber \\
 & =\sum_{e:y(1,w)=y}\int_{0}^{P(z,x)}q(e|u,x)du,\label{eq:data}
\end{align}
where the second equality is by $\Pr[Y(d,w)=y|U=u,X=x]=\sum_{e:y(d,w)=y}q(e|u,x)$.

To define the identified set for $\tau$, we introduce some simplifying
notation. Let $q(u)\equiv\{q(e|u,x)\}_{e\in\mathcal{\mathcal{E}},x\in\mathcal{X}}$
and 
\[
\mathcal{Q}\equiv\{q(\cdot):\sum_{e\in\mathcal{E}}q(e|u,x)=1\text{ and }q(e|u,x)\ge0\text{ }\forall(e,u,x)\}
\]
be the class of $q(u)$, and let
\[
p\equiv\{p(1,d|z,w,x)\}_{(d,z,w,x)\in\{0,1\}\times\mathcal{Z}\times\mathcal{W}\times\mathcal{X}}.
\]
Also, let $R_{\omega}:\mathcal{Q}\rightarrow\mathbb{R}$ and $R_{0}:\mathcal{Q}\rightarrow\mathbb{R}^{d_{p}}$
(with $d_{p}$ being the dimension of $p$) denote the linear operators
of $q(\cdot)$ that satisfy
\begin{align*}
R_{\omega}q & \equiv E\left[\sum_{e:y(1,W)=1}\int q(e|u,X)\omega_{1}(u,Z,X)du-\sum_{e:y(0,W)=1}\int q(e|u,X)\omega_{0}(u,Z,X)du\right],\\
R_{0}q & \equiv\left\{ \sum_{e:y(d,w)=1}\int_{\mathcal{U}_{z,x}^{d}}q(e|u,x)du\right\} _{(d,z,w,x)\in\{0,1\}\times\mathcal{Z}\times\mathcal{W}\times\mathcal{X}},
\end{align*}
where $\mathcal{U}_{z,x}^{d}$ denotes the intervals $\mathcal{U}_{z,x}^{1}\equiv[0,P(z,x)]$
and $\mathcal{U}_{z,x}^{0}\equiv(P(z,x),1]$. Then, we can characterize
the baseline identified set for $\tau$ where we only impose modeling
primitives. Later, we show how to characterize the identified set
with additional assumptions introduced in Section \ref{sec:Additional-Assumptions}.

\begin{definition}Suppose Assumptions EX and SEL(a) hold. The identified
set of $\tau$ is defined as
\begin{align*}
\mathcal{T}^{*} & \equiv\{\tau\in\mathbb{R}:\tau=R_{\omega}q\text{ for some }q\in\mathcal{Q}\text{ such that }R_{0}q=p\}.
\end{align*}

\end{definition}

In what follows, we formulate the infinite-dimensional LP ($\infty$-LP)
that characterizes $\mathcal{T}^{*}$. This program conceptualizes
sharp bounds on $\tau$ from the data and the maintained assumptions
(Assumptions SEL and EX). The upper and lower bounds on $\tau$ are
defined as
\begin{align}
\overline{\tau}= & \sup_{q\in\mathcal{Q}}R_{\omega}q,\tag{{\text{\ensuremath{\infty}-LP1}}}\label{eq:upper}\\
\underline{\tau}= & \inf_{q\in\mathcal{Q}}R_{\omega}q,\tag{{\text{\ensuremath{\infty}-LP2}}}\label{eq:lower}
\end{align}
subject to
\begin{align}
R_{0}q & =p.\tag{{\text{\ensuremath{\infty}-LP3}}}\label{eq:constr}
\end{align}
Observe that the set of constraints \eqref{eq:constr} does not include
\begin{align}
\sum_{e:y(d,w)=0}\int_{\mathcal{U}_{z,x}^{d}}q(e|u,x)du & =p(0,d|z,w,x)\qquad\forall(d,z,w,x).\label{eq:redun_constraint}
\end{align}
This is because we know a priori that they are redundant in the sense
that they do not further restrict the \textit{feasible set}, namely,
the set of $q(e|u,x)$'s that satisfy all the constraints ($q\in\mathcal{Q}$
and \eqref{eq:constr}).

\begin{lemma}\label{lem:redun}In the linear program \eqref{eq:upper}--\eqref{eq:constr},
the feasible set defined by $q\in\mathcal{Q}$ and \eqref{eq:constr}
is identical to the feasible set defined by $q\in\mathcal{Q}$, \eqref{eq:constr},
and \eqref{eq:redun_constraint}.\end{lemma}

\begin{theorem}\label{thm:sharp_bds}Under Assumptions SEL and EX,
suppose $\mathcal{T}^{*}$ is non-empty. Then, the bounds $[\underline{\tau},\overline{\tau}]$
in \eqref{eq:upper}--\eqref{eq:constr} are sharp for the target
parameter $\tau$, that is, $cl(\mathcal{T}^{*})=[\underline{\tau},\overline{\tau}]$,
where $cl(\cdot)$ is the closure of a set.\end{theorem}

The result of this theorem is immediate due to the convexity of the
feasible set $\{q:q\in\mathcal{Q}\}\cap\{q:R_{0}q=p\}$ in the LP
and the linearity of $R_{\omega}q$ in $q$, which implies that $[\underline{\tau},\overline{\tau}]$
is convex.

\section{Sieve Approximation and Finite-Dimensional Linear Programming\label{sec:Sieve-Approximation-and}}

Although conceptually useful, the LP \eqref{eq:upper}--\eqref{eq:constr}
is not feasible in practice because $\mathcal{Q}$ is an infinite-dimensional
space. In this section, we approximate \eqref{eq:upper}--\eqref{eq:constr}
with a finite-dimensional LP via a sieve approximation of the conditional
probability $q(e|u,x)$. We use Bernstein polynomials as the sieve
basis. Bernstein polynomials are useful in imposing restrictions on
the original function (\citet{joy2000bernstein}; \citet{chen2011sensitivity};
\citet{chen2017approximation}) and therefore have been introduced
in the context of linear programming (\citet{mogstad2018using}; \citet{mogstad2021causal};
\citet{masten2021salvaging}). Bernstein approximation based on those
polynomials possesses the property of being ``shape-preserving,''
which effectively prevent undesired distortions from the original
function's shape properties (\citet{goodman1989shape}; \citet{carnicer1993shape};
\citet{goodman1999convexity}).

Consider the following sieve approximation of $q(e|u,x)$ using Bernstein
polynomials of order $K$
\begin{align*}
q(e|u,x) & \approx\sum_{k=1}^{K}\theta_{k}^{e,x}b_{k}(u),
\end{align*}
where $b_{k}(u)\equiv b_{k,K}(u)\equiv\tbinom{K}{k}u^{k}(1-u)^{K-k}$
is a univariate Bernstein basis, $\theta_{k}^{e,x}\equiv\theta_{k,K}^{e,x}\equiv q(e|k/K,x)$
is its coefficient, and $K$ is finite. By using this approximation,
we implicitly impose smoothness in $q(e|\cdot,x)$. Note that discrete
$x$ can index $\theta$, because $q(e|u,x)$ is a saturated function
of $x$. By the definition of the Bernstein coefficient, for any $(e,x)$,
it satisfies $q(e|u,x)\ge0$ for all $u$ if and only if $\theta_{k}^{e,x}\ge0$
for all $k$. Also, $\sum_{e\in\mathcal{E}}q(e|u,x)=1$ for all $(u,x)$
is approximately equivalent to $\sum_{e\in\mathcal{E}}\theta_{k}^{e,x}=1$
for all $(k,x)$. To see this, first, $\sum_{e\in\mathcal{E}}q(e|u,x)=1$
for all $(u,x)$ implies $\sum_{e\in\mathcal{E}}\theta_{k}^{e,x}=\sum_{e\in\mathcal{E}}q(e|k/K,x)=1$
for all $(k,x)$. Conversely, when $\sum_{e\in\mathcal{E}}\theta_{k}^{e,x}=1$
for all $(k,x)$,
\begin{align*}
\sum_{e\in\mathcal{E}}q(e|u,x) & \approx\sum_{e\in\mathcal{E}}\sum_{k=1}^{K}\theta_{k}^{e,x}b_{k}(u)=\sum_{k=1}^{K}b_{k}(u)=1
\end{align*}
by the binomial theorem (\citet{coolidge1949story}). Motivated by
this approximation, we formally define the following sieve space for
$\mathcal{Q}$:
\begin{align}
\mathcal{Q}_{K} & \equiv\left\{ \Big\{\sum_{k=1}^{K}\theta_{k}^{e,x}b_{k}(u)\Big\}_{e\in\mathcal{E}}:\sum_{e\in\mathcal{E}}\theta_{k}^{e,x}=1\text{ and }\theta_{k}^{e,x}\ge0\text{ }\forall(e,k,x)\right\} \subseteq\mathcal{Q}.\label{eq:sieve}
\end{align}
Let $\mathcal{K}\equiv\{1,...,K\}$ and $p(z,w,x)\equiv\Pr[Z=z,W=w,X=x]$.
For $q\in\mathcal{Q}_{K}$, by \eqref{eq:target} and \eqref{eq:sieve},
the target parameter $\tau=E[\tau_{1}(Z,W,X)]-E[\tau_{0}(Z,W,X)]$
can be expressed with
\begin{align}
E[\tau_{d}(Z,W,X)] & =\sum_{(w,x)\in\mathcal{W\times\mathcal{X}}}\sum_{e:y(d,w)=1}\sum_{k\in\mathcal{K}}\theta_{k}^{e,x}\gamma_{k}^{d}(w,x),\label{eq:derive_target}
\end{align}
where $\gamma_{k}^{d}(w,x)\equiv\sum_{z\in\{0,1\}}p(z,w,x)\int b_{k}(u)\omega_{d}(u,z,x)du$.
Also, for $q\in\mathcal{Q}_{K}$ and $D=1$, by \eqref{eq:data},
we have 
\begin{align}
p(y,1|z,w,x) & =\sum_{e:y(1,w)=y}\sum_{k\in\mathcal{K}}\theta_{k}^{e,x}\delta_{k}^{1}(z,x),\label{eq:derive_constr}
\end{align}
where $\delta_{k}^{d}(z,x)\equiv\int_{\mathcal{U}_{z,x}^{d}}b_{k}(u)du$.

From \eqref{eq:derive_target} and \eqref{eq:derive_constr}, we can
expect that a finite-dimensional LP can be obtained with respect to
$\theta_{k}^{e,x}$. Let $\theta\equiv\{\theta_{k}^{e,x}\}_{(e,k,x)\in\mathcal{E}\times\mathcal{K}\times\mathcal{X}}$
and let
\begin{align*}
\Theta_{K} & \equiv\left\{ \theta:\sum_{e\in\mathcal{E}}\theta_{k}^{e,x}=1\text{ and }\theta_{k}^{e,x}\ge0\text{ }\forall(e,k,x)\in\mathcal{E}\times\mathcal{K}\times\mathcal{X}\right\} .
\end{align*}
Then, we can formulate the following finite-dimensional LP that corresponds
to the $\infty$-LP in \eqref{eq:upper}--\eqref{eq:constr}:
\begin{align}
\overline{\tau}_{K} & =\max_{\theta\in\Theta_{K}}\sum_{(k,w,x)\in\mathcal{K\times\mathcal{W}\times\mathcal{X}}}\Big\{\sum_{e:y(1,w)=1}\theta_{k}^{e,x}\gamma_{k}^{1}(w,x)-\sum_{e:y(0,w)=1}\theta_{k}^{e,x}\gamma_{k}^{0}(w,x)\Big\}\tag{{\text{LP1}}}\label{eq:upper2}\\
\underline{\tau}_{K} & =\min_{\theta\in\Theta_{K}}\sum_{(k,w,x)\in\mathcal{K\times\mathcal{W}\times\mathcal{X}}}\Big\{\sum_{e:y(1,w)=1}\theta_{k}^{e,x}\gamma_{k}^{1}(w,x)-\sum_{e:y(0,w)=1}\theta_{k}^{e,x}\gamma_{k}^{0}(w,x)\Big\}\tag{{\text{LP2}}}\label{eq:lower2}
\end{align}
subject to
\begin{align}
\sum_{e:y(d,w)=1}\sum_{k\in\mathcal{K}}\theta_{k}^{e,x}\delta_{k}^{d}(z,x) & =p(1,d|z,w,x)\qquad\forall(d,z,w,x)\in\{0,1\}\times\mathcal{Z}\times\mathcal{W}\times\mathcal{X}.\tag{{\text{LP3}}}\label{eq:constr2}
\end{align}
One of the advantages of LP is that it is computationally very easy
to solve using standard algorithms, such as the simplex algorithm.
Assuming binary $W$ and $X$ and setting $K=50$, we have $\dim(\theta)=1632$,
and it takes less than 3 seconds to calculate $\overline{\tau}_{K}$
and $\underline{\tau}_{K}$ with moderate computing power. The increase
in the support of $\mathcal{W}$ (and thus the number of maps \eqref{eq:map})
only linearly increases the computation time.

The important remaining question is how to choose $K$ in practice.
We discuss this issue in Section \ref{sec:Simulation}. Also, when
$K$ is too large, coefficients on $\theta_{k}^{e,x}$ in \eqref{eq:constr2}
tend to take values with incomparable orders of magnitude, which may
cause common optimization algorithms to arbitrarily drop small values.
To address this, we propose a rescaling method in Section \ref{subsec:Rescaling-of-Linear}.
Finally, extending Proposition 4 in \citet{mogstad2018using}, we
may exactly calculate $\overline{\tau}$ and $\underline{\tau}$ (i.e.,
$\overline{\tau}=\overline{\tau}_{K}$ and $\underline{\tau}=\underline{\tau}_{K}$)
under the assumptions that (i) the weight function $\omega_{d}(u,z,x)$
is piece-wise constant in $u$ and (ii) the constant spline that provides
the best mean squared error approximation of $q(e|u,x)$ satisfies
all the maintained assumptions (possibly including the identifying
assumptions introduced later) that $q(e|u,x)$ itself satisfies; see
\citet{mogstad2018using} for details.

\section{The Use of Statistical Independence\label{sec:Identifying-Power-of}}

The framework proposed in this paper allows us to systematically calculate
sharp bounds of treatment parameters using the (conditional) statistical
independence of $Z$ and $W$ (Assumption EX). This is possible because
the full distribution from the data enters the LP. This is in contrast
to the approach in \citet{mogstad2018using} who utilize the first
moment information of the data and exploit mean independence of $Z$
in their bound analysis. The full independence assumption is common
in the nonparametric treatment effect literature and can be justified
by, for example, IVs generated by randomized experiments. This section
formally show the identifying power of full independence and, more
importantly, how our framework allows us to make use of it to produce
sharp bounds. To this end, we compare our approach with \citet{mogstad2018using}'s
and show when the identified sets coincide and when they do not. We
first focus on a simple case where $Z$ is the only exogenous variable.
Then, we discuss how  $W$ can add further identifying power.

To help the reader, we restate the maintained assumption in \citet{mogstad2018using}
in terms of our notation:

\begin{asI}[Mogstad et al. (2019)]I.1. $U\perp Z|X$; I.2. $E[Y(d)|Z,X,U]=E[Y(d)|X,U]$
and $E[Y(d)^{2}]<\infty$ for $d\in\{0,1\}$; I.3. $U$ is continuously
distributed conditional on $X$.\end{asI}

Assumption I.2 imposes mean independence of the instrument $Z$. On
the other hand Assumption EX (after suppressing $W$ for clear comparison)
implies $Y(d)\perp Z|X,U$, namely, statistical independence of $Z$.
Therefore, we can expect that the latter has stronger identifying
power. Below, we show that the LP constructed with $q$ being the
choice variable instead of the MTR function $m_{d}$ enables us to
exploit this extra identifying power.

In \citet{mogstad2018using} the IV-like estimands are the channel
through which the data enters to restrict the set of $m_{d}$. They
show (their Proposition 3) that if the IV-like estimands are carefully
chosen, then the resulting set of $m_{d}$ is equivalent to the set
of $m_{d}$ that are consistent with
\begin{align}
 & E[Y|D=0,Z,X]=E[Y_{0}|U>p(Z,X),Z,X]=\frac{1}{1-P(Z,X)}\int_{p(Z,X)}^{1}m_{0}(u,X)du,\label{eq: infoset0}\\
 & E[Y|D=1,Z,X]=E[Y_{1}|U\leq p(Z,X),Z,X]=\frac{1}{P(Z,X)}\int_{0}^{p(Z,X)}m_{1}(u,X)du.\label{eq: infoset1}
\end{align}
Therefore, assuming such IV-like estimands are chosen, we can define
their identified set as 
\[
\mathcal{M}_{id}=\Big\{ m=(m_{0},m_{1}),m_{0},m_{1}\in L^{2}:m_{0},m_{1}\text{ satisfies equation (\ref{eq: infoset0}) and (\ref{eq: infoset1}) a.s.}\Big\}.
\]
Note that this set is the result of funneling the information from
data via the conditional means. 

When $Y$ is binary, the mean and full independence assumptions are
equivalent. Therefore we expect to find no difference in the resulting
identified sets of treatment parameters between the two methods. We
show this by comparing the identified set of the MTR functions $\mathcal{M}_{id}$
used in \citet{mogstad2018using} and the set of MTR functions derived
from the feasible set in the LP proposed in the current paper (denoted
as $\mathcal{M}_{f}$ below). More importantly, as soon as $Y$ departs
from binarity, we show that $\mathcal{M}_{f}$ is a proper subset
$\mathcal{M}_{id}$.

For comparison, define the feasible set $\mathcal{Q}_{f}$ of our
proposed LP as
\[
\mathcal{Q}_{f}=\Big\{ q\in L^{2}([0,1]):q\in\mathcal{Q}\text{ and satisfies equation }(\ref{eq:constr})\Big\},
\]
where $w$ is suppressed from \eqref{eq:constr}. To establish the
connection with $\mathcal{M}_{id}$, we construct the set of MTR functions
based on $\mathcal{Q}_{f}$. For this purpose, we drop $W$ from the
setting and redefine $\epsilon\equiv(Y(0),Y(1))$ and $e\equiv(y(0),y(1))$:

\[
\mathcal{M}_{f}=\Big\{ m=(m_{0},m_{1}):m_{d}(u,x)=\sum_{y\in\mathcal{Y}}y\sum_{e:y(d)=y}q(e|u,x),d=\{0,1\},q\in\mathcal{Q}_{f}\Big\}.
\]
Then the following holds:

\begin{theorem}\label{thm:equivalence_with_mogstad}Suppose $\mathcal{Y}=\{0,1\}$.
Under Assumptions SEL and EX, $\mathcal{M}_{f}=\mathcal{M}_{id}$.
\end{theorem}

Theorem \ref{thm:equivalence_with_mogstad} demonstrates that, when
$Y$ is binary, the constraint used in our LP approach captures the
same information as the constraint defined by the IV-like estimands
that are appropriately selected. Consistent with Proposition 3 in
\citet{mogstad2018using}, $\mathcal{M}_{f}$ and $\mathcal{M}_{id}$
are sharp in this case.

When $Y$ is non-binary and we continue to impose statistical independence,
the above equivalence breaks down. That is, unlike $\mathcal{M}_{id}$,
$\mathcal{Q}_{f}$ exploits the full data distribution and the statistical
independence assumption. We formally show this with $\mathcal{Y}=\{y_{1},...,y_{L}\}$.
We accordingly redefine $p$ and $R_{0}$ in \eqref{eq:constr} (i.e.,
$R_{0}q=p$) in the definition of $\mathcal{Q}_{f}$ as
\begin{align*}
p & \equiv\{p(y,d|z,x)\}_{(y,d,z,x)\in\mathcal{Y}\times\{0,1\}\times\mathcal{Z}\times\mathcal{X}},\\
R_{0}q & \equiv\left\{ \sum_{e:y(d)=y}\int_{\mathcal{U}_{z,x}^{d}}q(e|u,x)du\right\} _{(y,d,z,x)\in\mathcal{Y}\times\{0,1\}\times\mathcal{Z}\times\mathcal{X}},
\end{align*}
without excluding redundant constraints and suppressing $w$ for simplicity.
We introduce a technical assumption.

\begin{asEC}(i) $P(z,x)$ is upper semicontinuous on compact $\mathcal{Z}\times\mathcal{X}$;
(ii) $\mathcal{Q}_{f}\subset L^{2}([0,1])$ is equicontinuous.\end{asEC}

\begin{theorem}\label{thm:better_than_mogstad}Suppose $\mathcal{Y}=\{y_{1},...,y_{L}\}$
with arbitrary $L\ge3$. Under Assumptions SEL, EX and CT, $\mathcal{M}_{f}\subsetneq\mathcal{M}_{id}$.
\end{theorem}

When $Y$ takes more than two values,  $\mathcal{Q}_{f}$ continues
to exploit the full independence assumption and as shown in Theorem
\ref{thm:sharp_bds}, $\mathcal{M}_{f}$ is sharp. Note $\mathcal{M}_{id}$
is a strict outer set of $\mathcal{M}_{f}$ because the former only
exploits the first moment implication of full independence. This intuition
also suggests that the case of continuous $Y$ would deliver a similar
result as in Theorem \ref{thm:better_than_mogstad}. Theorems \ref{thm:equivalence_with_mogstad}
and \ref{thm:better_than_mogstad} corroborate the simulation results
in Section \ref{subsec:Identifying-Power-of}. The simulation results
for continuous $Y$ can be found in Section \ref{subsec:More-results-on}.

Finally, consider the case where $W$ is also present and satisfies
Assumptions EX and SEL. After proper modification of their assumptions
(i.e., $U\perp(Z,W)|X$ and $E[Y(d)|Z,W,X,U]=E[Y(d)|X,U]$), \citet{mogstad2018using}'s
framework may be able to use the exogenous variation of $W$. This
can be done by using the moments $E[Y|D=d,Z,W,X]$ ($d=0,1$) as inputs
in \eqref{eq: infoset0} and \eqref{eq: infoset1}.\footnote{Alternatively, the variation of $W$ can be reflected via IV-like
estimands as inputs. In this case, the identification of $\tau_{LATE}(w,w')$
is useful. We show its identifiability in Lemma \ref{lem:direction}
below.} However, by the same argument as the one above for $Z$, the full
independence assumption with respect to $W$ will have a stronger
identifying power than mean independence. When $Y$ is beyond binary,
our framework can exploit this. The two frameworks also differ in
how other additional identifying assumptions can be incorporated in
the procedures. This point appears as Remark \ref{rem:U_in_Mogstad}
in the next section.

\begin{remark}[Manski's Bounds]When the ATE is the target parameter,
the findings of this section has implications on the comparison of
\citet{manski1990nonparametric}'s bounds and our bounds (with no
additional assumptions). \citet{manski1990nonparametric} assumes
the mean independence of exogenous variables. It is shown by \citet{heckman2001instrumental}
that Manski's bounds coincide with their bounds, which should also
be equivalent to \citet{mogstad2018using}'s bounds (with no additional
assumptions).\footnote{Although \citet{heckman2001instrumental} impose full independence,
what they actually require in the derivation of their bounds is the
mean independence as in \citet{mogstad2018using}.} Therefore, Theorem \ref{thm:better_than_mogstad} implies that, under
the full independence assumption, our procedure produces the ATE bounds
that are narrower than Manski's bounds.\end{remark}

\section{Incorporating Additional Identifying Assumptions\label{sec:Additional-Assumptions}}

\subsection{Additional Identifying Assumptions\label{subsec:Additional-Identifying-Assumptio}}

In addition to Assumptions SEL, EX and R, researchers may be willing
to restrict the degree of treatment heterogeneity, the direction of
endogeneity, or the shape of the MTR functions. Although not necessary,
such restrictions play significant roles in yielding informative bounds,
especially given the challenge of extrapolating the LATE with minimal
variation of the exogenous variables. Among the proposed assumptions,
researchers want to use those that they deem plausible in given applications.
Also, we provide only a few examples of assumptions here; we believe
our proposed framework may open a venue for other assumptions we have
not explored.

\subsubsection{Restrictions on Treatment Heterogeneity\label{subsec:Restrictions-on-Treatment Heterogeneity}}

We present a range of restrictions on treatment heterogeneity in the
order of stringency starting from the strongest.

\begin{asU2}For every $w,w'\in\mathcal{W}$ and $x\in\mathcal{X}$,
either $\Pr[Y(1,w)\ge Y(0,w')|X=x]=1$ or $\Pr[Y(1,w)\le Y(0,w')|X=x]=1$.\end{asU2}

The following assumption is weaker than Assumption U$^{*}$.

\begin{asU}For every $w\in\mathcal{W}$ and $x\in\mathcal{X}$, either
$\Pr[Y(1,w)\ge Y(0,w)|X=x]=1$ or $\Pr[Y(1,w)\le Y(0,w)|X=x]=1$.\end{asU}

When $W$ is not available at all, this assumption can be understood
with $\text{\ensuremath{\mathcal{W}}}$ being degenerate. In Assumption
U$^{*}$, $w$ and $w'$ may be the same or different, i.e., the uniformity
is for all combinations of $(w,w')\in\{(0,0),(1,1),(1,0),(0,1)\}$.
Therefore, Assumption U$^{*}$ implies Assumption U. Assumptions U
and U$^{*}$ posit that individuals present uniformity in the sense
that the treatment either weakly increases the outcome for all individuals
or decreases it for all individuals. Intuitively, Assumption U$^{*}$
is stronger because the uniformity remains to hold even under an outcome
shift via $w\neq w'$; this may hold when the treatment effect is
strong across individuals. Assumptions U and U$^{*}$ share insights
with the monotone treatment response assumption that is introduced
to bound the ATE in \citet{manski1997monotone} and \citet{MP00}.
Assumptions U and U$^{*}$ are also related to rank invariance considered
in the literature (e.g., \citet{chernozhukov2005iv}, \citet{marx2020sharp}).
To see this, consider binary $Y$ and a structural model $Y=1[s(D,W)\ge V_{D}]$
with $V_{D}\equiv DV_{1}+(1-D)V_{0}$ (suppressing $X$). When rank
invariance ($V_{1}=V_{0}\equiv V$) holds, Assumption U$^{*}$ holds
(and thus Assumption U) because either $P[s(1,w)<V,s(0,w')\ge V]=0$
or $P[s(1,w)\ge V,s(0,w')<V]=0$. On the other hand, Assumption U$^{*}$
can still hold even if $V_{1}\neq V_{0}$ when, for example, the distribution
of $(V_{1},V_{0})$ is concentrated around $45^{\circ}$ line. Therefore
the converse is not true.\footnote{On the other hand, Assumption U$^{*}$ and rank similarity ($F_{V_{1}|U}=F_{V_{0}|U}$)
are not nested.} It is important to note that Assumptions U and U$^{*}$ still allow
treatment heterogeneity in terms of $X$ (and similarly of $W$).
For instance, Assumption U allows that $Y(1,w)\ge Y(0,w)$ a.s. for
$X=x$ but $Y(1,w)\le Y(0,w)$ a.s. for $X=x'$.\footnote{A related idea of conditional rank preservation appears in \citet{han2021identification}
and also used in \citet{marx2020sharp}.} Assumption U is also used, for example, in \citet{chiburis2010semiparametric}
although his focus is bounds on the ATE with binary $Y$.

In fact, it is possible to further weaken Assumption U. To motivate
this with binary $Y$, we state Assumption U as follows (suppressing
$X$): for every $w$, either $P[Y(1,w)=0,Y(0,w)=1]=0$ or $P[Y(1,w)=1,Y(0,w)=0]=0$.
In other words, all individuals respond weakly monotonically to the
treatment in the sense that there is no (strictly) negative-treatment-response
type or positive-treatment-response type. By iterated expectation,
Assumption U holds if and only if either $P[Y(1,w)=0,Y(0,w)=1|U]=0$
a.s. or $P[Y(1,w)=1,Y(0,w)=0|U]=0$ a.s. Then this assumption can
be relaxed by allowing for the existence of both types in the population
and instead assuming the dominance of one of the types over the other:
for example, $P[Y(1,w)=1,Y(0,w)=0|U]>P[Y(1,w)=0,Y(0,w)=1|U]$ a.s.
For general $Y$, such an assumption is written as follows (ignoring
a measure zero set of $U$):

\begin{asU3}For every $w\in\mathcal{W}$ and $x\in\mathcal{X}$,
either $P[Y(1,w)\ge Y(0,w)|U=u,X=x]\ge P[Y(1,w)\le Y(0,w)|U=u,X=x]$
for all $u$ or $P[Y(1,w)\ge Y(0,w)|U=u,X=x]\le P[Y(1,w)\le Y(0,w)|U=u,X=x]$
for all $u$.

\end{asU3}

Assumption U$^{0}$ is weaker than Assumption U, because when $P[Y(1,w)\le Y(0,w)|X=x]=0$,
Assumption U$^{0}$ trivially holds.\footnote{One can come up with assumptions which strength is between Assumption
U$^{0}$ and Assumption U. For example with binary $Y$, one can assume
that, in addition to Assumption U$^{0}$, $P[Y(1,w)=1,Y(0,w)=1|X=x]\ge P[Y(1,w)=0,Y(0,w)=1|X=x]$
and $P[Y(1,w)=0,Y(0,w)=0|X=x]\ge P[Y(1,w)=0,Y(0,w)=1|X=x]$ hold.
We do not explore these assumptions for succinctness.} Compared to Assumption U, Assumption U$^{0}$ allows further treatment
heterogeneity in that positive- and negative-treatment-response types
can both present in the population, but we impose a uniform order
on the probabilities of effect sign across individuals. Researchers
may be more comfortable with this assumption than complete uniformity.
This paper is the first to propose to use this restriction in the
identification of treatment effects. We call assumptions of this type
\textit{rank dominance}. 

We show that the directions in Assumptions U$^{*}$--U$^{0}$ may
be learned from the data. Suppress $X$ for simplicity and let $Z\in\{0,1\}$.
Let $\tau_{LATE}(w,w')\equiv E[Y(1,w)-Y(0,w')|P(0)\le U\le P(1)]$
be the LATE defined using the extended counterfactual outcome $Y(d,w)$
assuming Assumption SEL(a) and $P(1)>P(0)$; the definition of $\tau_{LATE}(w,w')$
under Assumption SEL(b) is analogous.

\begin{lemma}\label{lem:direction}Suppose Assumptions SEL and EX
hold and $Z\in\{0,1\}$. Then, (i) $\tau_{LATE}(w,w')$ is identified
for any $w,w'\in\mathcal{W}$. Based on this, (ii) the directions
in Assumptions U$^{*}$ and U are identified; (iii) additionally,
when $Y\in\{0,1\}$, the direction in Assumption U$^{0}$ is identified.\end{lemma}

The practical implication of this lemma is that the direction of each
assumption will be detected by the feasibility of relevant LP upon
imposing it; the next subsection has more details. The identification
of $\tau_{LATE}(w,w')$ extends the result in \citet{imbens1994identification}.
In particular,
\begin{align*}
\tau_{LATE}(w,w') & =\frac{1}{P(1)-P(0)}\left\{ \left(E[Y|Z=1,W=w]-E[Y|Z=0,W=w]\right.\right.\\
 & \quad+(1-P(0))(E[Y|D=0,Z=0,W=w]-E[Y|D=0,Z=0,W=w'])\\
 & \quad-\left.(1-P(1))(E[Y|D=0,Z=1,W=w]-E[Y|D=0,Z=1,W=w'])\right\} 
\end{align*}
The intuition behind (ii) and (iii) is as follows. The directions
of monotonicity in Assumptions U$^{*}$ and U can be identified by
the signs of relevant LATEs. On the other hand, the direction of inequality
in Assumption U$^{0}$ can only be identified with binary $Y$, because
otherwise the proportions of positive and negative treatment effects
can be offset by the magnitude of the effects. This limited testability
reflects the weak identifying power of Assumption U$^{0}$. The detail
of the intuition can be found in the formal proof of Lemma \ref{lem:direction}
in the Appendix. Previous work has discussed the role of the rank
similarity assumption on determining the sign of the ATE (\citet{bhattacharya2008treatment,SV11,han2019optimal}),
and the result above shows that Assumptions U$^{*}$--U$^{0}$ play
a similar role in the LP approach.

\subsubsection{Direction of Endogeneity}

In some applications, researchers are relatively confident about the
direction of treatment endogeneity. The idea of imposing the direction
of the selection bias as an identifying assumption appears in \citet{MP00},
who introduce monotone treatment selection (MTS), in addition to the
monotone treatment response assumption mentioned above.

\begin{asMTS}For every $w\in\mathcal{W}$ and $x\in\mathcal{X}$,
$E[Y(d,w)|D=1,X=x]\ge E[Y(d,w)|D=0,X=x]$ for $d\in\{0,1\}$.\end{asMTS}

\subsubsection{Shape Restrictions\label{subsec:Shape-Restrictions}}

Monotonicity and concavity are common shape restrictions used in the
context of MTR and MTE framework.

\begin{asM}For $(w,x)\in\mathcal{W}\times\mathcal{X}$, $m_{d}(u,w,x)$
is weakly increasing in $u\in[0,1]$.\end{asM}

\begin{asC}For $(w,x)\in\mathcal{W}\times\mathcal{X}$, $m_{d}(u,w,x)$
is weakly concave in $u\in[0,1]$.\end{asC}

Assumption M assumes monotonicity and appears in \citet{brinch2017beyond}
and \citet{mogstad2018using}; while Assumption C assumes concavity
and appears in \citet{mogstad2018using}. Another shape restriction
introduced in the literature is separability: $m_{d}(u,w,x)=m_{1d}(w,x)+m_{2d}(w,u)$
is weakly concave in $u\in[0,1]$.

\subsection{Incorporating Additional Assumptions in the LP}

In this section, we show how identifying assumptions introduced in
Section \ref{subsec:Additional-Identifying-Assumptio} can be easily
translated into assumptions on the mapping defined in \eqref{eq:map}.
This allows us to incorporate the additional assumptions in the formulation
of the LP, so that one does not need to manually derive analytical
bounds every time she imposes a new assumption (and prove their sharpness).

Before proceeding, we revisit Assumptions SEL, EX, and R in the context
of the LP. First, we formally show that the existence and relevance
of $W$ (as well as $Z$) embodied in Assumptions SEL, EX, and R can
be a useful source in narrowing the bounds.

\begin{theorem}\label{lem:nonredundancy}Under Assumptions SEL, EX,
and R, the variation of $Z$ and $W$ respectively poses non-redundant
constraints on $\theta\in\Theta_{K}$ in \eqref{eq:upper2}--\eqref{eq:constr2}
in that the rows of the constraint matrix in \eqref{eq:constr2} are
linearly independent.

\end{theorem}

Heuristically, the improvement occurs because, with R(i), the constraint
matrix (i.e., the matrix multiplied to the vector $\theta$ in \eqref{eq:constr2})
has greater rank with the variation of $W$ than without. See the
proof of the theorem for a formal argument. Note that non-redundant
constraints on $\theta$ do not always guarantee an improvement of
the bounds in \eqref{eq:upper2}--\eqref{eq:constr2}, because these
constraints may still be non-binding. Nevertheless, non-redundancy
is a necessary condition for the improvement.

We now show how to incorporate Assumptions U$^{*}$, U, U$^{0}$,
MTS, M, and C as additional equality and inequality restrictions in
the LP: Given the LP \eqref{eq:upper}--\eqref{eq:constr}, identifying
assumptions can be imposed by appending
\begin{align}
R_{1}q & =a_{1},\tag{{\text{\ensuremath{\infty}-LP4}}}\label{eq:LP4}\\
R_{2}q & \le a_{2},\tag{{\text{\ensuremath{\infty}-LP5}}}\label{eq:LP5}
\end{align}
where $R_{1}$ and $R_{2}$ are linear operators on $\mathcal{Q}$
that correspond to equality and inequality constraints, respectively,
and $a_{1}$ and $a_{2}$ are some vectors in Euclidean space. Then,
analogous constraints on $\theta$ can be added to the finite-dimensional
LP \eqref{eq:upper2}--\eqref{eq:constr2}. When an assumption violates
the true data-generating process, then the identified set will be
empty. This corresponds to the situation where the LP does not have
a feasible solution. When we reflect sampling errors, this corresponds
to the case where the confidence set is empty.\footnote{In order to verify whether the identified set is empty, we need to
check whether the feasible set of $\theta$ is empty. An efficient
way to do this is to identify vertices of the feasible polytope, if
any. This process is no simpler than the simplex algorithm that we
use to solve the LP. Therefore, we recommend that one first solves
the LP and check if infeasibility is reported.}

Assumptions U and U$^{*}$ are imposed in the LP by ``deactivating''
relevant maps. Motivating by Lemma \ref{lem:direction}, suppose the
researcher knows that $Y(1,w)\ge Y(0,w)$ almost surely for all $w\in\{0,1\}$
under Assumption U. This assumption can be imposed as equality constraints
\eqref{eq:LP4}, namely, in the form of $R_{1}q=a_{1}$: Suppressing
$x$ for simplicity and recalling $\epsilon\equiv(Y(0,0),Y(0,1),Y(1,0),Y(1,1))$,
\begin{align}
q(0,1,0,0|u) & =q(1,1,0,0|u)=q(0,1,1,0|u)=q(1,1,1,0|u)=0,\label{eq:U_restr1}\\
q(1,0,0,0|u) & =q(1,1,0,0|u)=q(1,0,0,1|u)=q(1,1,0,1|u)=0,\label{eq:U_restr2}
\end{align}
respectively, corresponding for $w=1$ and $w=0$ in $Y(1,w)\ge Y(0,w)$.
Therefore, the corresponding $\theta_{k}^{e}=0$. Then, the effective
dimension of $\theta$ will be reduced in \eqref{eq:upper2}--\eqref{eq:constr2}
and thus yields narrower bounds. As another example, suppose the researcher
knows that the following holds almost surely under Assumption U$^{*}$:
$Y(1,1)\ge Y(0,0)$, $Y(1,0)\ge Y(0,1)$, $Y(1,1)\ge Y(0,1)$, and
$Y(1,0)\ge Y(0,0)$. These inequalities respectively imply
\begin{align}
q(1,0,0,0|u) & =q(1,1,0,0|u)=q(1,0,1,0|u)=q(1,1,1,0|u)=0,\label{eq:U_restr3}\\
q(0,0,1,0|u) & =q(1,0,1,0|u)=q(0,0,1,1|u)=q(1,0,1,1|u)=0,\label{eq:U_restr4}
\end{align}
and \eqref{eq:U_restr1}--\eqref{eq:U_restr2}. Recall the discussion
that, in Assumption U (Assumption U$^{*}$), the direction of monotonicity
is allowed to be different for different $w$ ($(w,w')$ pairs). This
direction will be identified from the data (Lemma \ref{lem:direction}).
Specifically, the direction can be automatically determined from the
LP by inspecting whether the LP has a feasible solution; when wrong
maps are removed, there is no feasible solution. This result holds
regardless of the existence of $W$. A similar argument applies to
Assumption U$^{0}$ with binary $Y$. Finally, suppose $P[Y(1,w)\ge Y(0,w)|U=u]\ge P[Y(1,w)\le Y(0,w)|U=u]$
for all $w\in\{0,1\}$ and $u$ under Assumption U$^{0}$. Then,
we can generate the following inequality restrictions:
\begin{align*}
 & q(0,0,1,0|u)+q(0,1,1,0|u)+q(0,0,1,1|u)+q(0,1,1,1|u)\\
 & \ge q(1,0,0,0|u)+q(1,1,0,0|u)+q(1,0,0,1|u)+q(1,1,0,1|u),\\
 & q(0,0,0,1|u)+q(1,0,0,1|u)+q(0,0,1,1|u)+q(1,0,1,1|u)\\
 & \ge q(0,1,0,0|u)+q(1,1,0,0|u)+q(0,1,1,0|u)+q(1,1,1,0|u).
\end{align*}

Next, consider Assumption MTS. This assumption can be imposed in the
form of $R_{2}q\le a_{2}$. To see this, Assumption MTS is equivalent
to
\begin{align*}
\sum_{e:y(d,w)=1}E\left[\left.\int_{P(Z,X)}^{1}q(e|u,X)du-\int_{0}^{P(Z,X)}q(e|u,X)du\right|X=x\right] & \le0
\end{align*}
for all $(d,w,x)$. As is clear from this expression, Assumption MTS
imposes restrictions on the joint distribution of $(\epsilon,U)$.

Finally, consider Assumptions M and C. It is straightforward to incorporate
the shape restrictions on the MTR or MTE function. They can be imposed
via inequality constraints \eqref{eq:LP4}, namely, in the form of
$R_{2}q\le a_{2}$. For implications on the finite-dimensional LP
\eqref{eq:upper2}--\eqref{eq:constr2}, recall that for $q\in\mathcal{Q}_{K}$,
the MTR satisfies
\begin{align*}
m_{d}(u,w,x) & =\sum_{e:y(d,w)=1}q(e|u,x)=\sum_{k\in\mathcal{K}}\sum_{e:y(d,w)=1}\theta_{k}^{e,x}b_{k}(u).
\end{align*}
According to the property of the Bernstein polynomial, Assumption
M implies that $\sum_{e:y(d,w)=1}\theta_{k}^{e,x}$ is weakly increasing
in $k$, i.e.,
\begin{align*}
\sum_{e:y(d,w)=1}\theta_{1}^{e,x} & \le\sum_{e:y(d,w)=1}\theta_{2}^{e,x}\le\cdots\le\sum_{e:y(d,w)=1}\theta_{K}^{e,x}.
\end{align*}
Assumption C implies that
\begin{align*}
\sum_{e:y(d,w)=1}\theta_{k}^{e,x}-\sum_{e:y(d,w)=1}2\theta_{k+1}^{e,x}+\sum_{e:y(d,w)=1}\theta_{k+2}^{e,x} & \le0\qquad\text{for }k=0,...,K-2.
\end{align*}
One can obtain analogous assumptions and their implications in the
presence of $W$.

\begin{remark}\label{rem:U_in_Mogstad}In terms of incorporating
additional identifying assumptions, some implications of Assumptions
U$^{*}$--U$^{0}$ (but not these assumptions directly) may be imposed
via the MTR function of \citet{mogstad2018using}'s framework. Nevertheless,
Assumptions U and U$^{0}$ cannot play distinctive roles in their
framework. To see this, consider $P[Y(1)\ge Y(0)]=1$, which is consistent
with Assumption U (suppressing $(W,X)$). This implies that $m_{1}(u)\ge m_{0}(u)$,
which then can be imposed as a restriction in \citet{mogstad2018using}.
However, $P[Y(1)\ge Y(0)|U=u]\ge P[Y(1)\le Y(0)|U=u]$ $\forall u$,
which is consistent with Assumption U$^{0}$, also implies $m_{1}(u)\ge m_{0}(u)$.\end{remark}

\section{Extension: Continuous $Y$\label{sec:Extension:-Continuous}}

\subsection{Identified Set and Infinite-Dimensional Linear Programming}

The analogous approach of LP can be applied to the case of continuous
outcome variable. We consider the continuous outcome with support
$\mathcal{Y}=[0,1]$ without loss of generality.\footnote{Note that $\mathcal{Y}=\mathbb{R}$ is homeomorphic to the open interval
$(0,1)$. We use the closure of the latter as $\mathcal{Y}$ for notational
convenience later.} As a key component of our LP, we define the following conditional
distribution:
\begin{align*}
\tilde{q}(e|u,x) & \equiv\Pr\left[\epsilon\le e|U=u,X=x\right],
\end{align*}
where $\epsilon\equiv(Y(0,0),Y(0,1),Y(1,0),Y(1,1))$ and $e\equiv(y(0,0),y(0,1),y(1,0),y(1,1))$
as before, and ``$\epsilon\le e$'' is understood as an element-wise
inequality. First, we show how the data distribution imposes restrictions
on $\tilde{q}(e|u,x)$. From the data, we observe
\[
\pi(y,d|z,w,x)\equiv\Pr\left[Y\leq y,D=d|Z=z,W=w,X=x\right]
\]
for all $(y,d,z,w,x)$. Then, for example, consider the case with
$d=1$. The conditional distribution can be written as
\begin{align*}
\pi(y,1|z,w,x) & \equiv\Pr\left[Y\leq y,D=1|Z=z,W=w,X=x\right]\\
 & =\Pr\left[Y(1,w)\leq y,U\leq P(z,x)|X=x\right]\\
 & =\int_{0}^{P(z,x)}\Pr\left[Y(1,w)\leq y|U=u,X=x\right]du\\
 & =\int_{0}^{P(z,x)}\int_{\{e:y(1,w)\le y\}}\tilde{q}(e|u,x)dedu,
\end{align*}
where the second equality follows  by Assumption EX and the inner
integral is the shorthand for a multiple integral with respect to
the vector $e$.\footnote{It is worth noting that even though we use the joint distribution
of $\epsilon\equiv(Y(0,0),Y(0,1),Y(1,0),Y(1,1))$ as the building
block, we do not require joint independence between $\epsilon$ and
$(Z,W)$ for the derivation. This is because the observed data distribution
is a marginal distribution in $Y$, and thus the marginal independence
is enough. The same explanation applies to the analogous derivation
\eqref{eq:data} in Section \ref{sec:Distribution-of-State}. } 

Similarly for the target parameters, the MTR function can be expressed
as follows. For example, for $D=0$,
\begin{align*}
m_{0}(u,w,x) & =E\left[Y(0,w)|U=u,X=x\right]\\
 & =\int_{0}^{1}(1-\Pr\left[Y(0,w)\leq y|U=u,X=x\right])dy\\
 & \equiv1-\int_{0}^{1}\int_{\{e:y(0,w)\le y\}}\tilde{q}(e|u,x)dedy.
\end{align*}
We now define the identified set of the target parameters. Let $\tilde{q}(u)\equiv\{\tilde{q}(e|u,x)\}_{e\in\mathcal{E},x\in\mathcal{X}}$
be the vector of $\tilde{q}(e|u,x)$'s. We introduce the class of
$\tilde{q}(\cdot)$ to be
\begin{align*}
\tilde{\mathcal{Q}}\equiv & \left\{ \tilde{q}(\cdot):0\leq\tilde{q}(\cdot|u,x\text{)}\leq1,\tilde{q}(\cdot|u,x)\text{ is 4-increasing},\footnotemark\right.\\
 & \qquad\quad\left.\tilde{q}(1,1,1,1|u,x)=1,\tilde{q}(e|u,x)=0\text{ for }e\text{ that contains zero }\forall(u,x)\right\} .
\end{align*}
\footnotetext{A function $\tilde{q}:[0,1]^{4}\rightarrow[0,1]$ is 4-increasing if, for any 4-box $[a_{1},b_{1}]\times\cdots\times[a_{4},b_{4}]\subseteq[0,1]^{4}$, it satisfies \[ \Delta_{a_{4}}^{b_{4}}\Delta_{a_{3}}^{b_{3}}\Delta_{a_{2}}^{b_{2}}\Delta_{a_{1}}^{b_{1}}\tilde{q}(e)\geq0, \] where $\Delta_{a_{1}}^{b_{1}}\tilde{q}(e)\equiv\tilde{q}(b_{1},e_{2},e_{3},e_{4})-\tilde{q}(a_{1},e_{2},e_{3},e_{4})$, $\Delta_{a_{2}}^{b_{2}}\tilde{q}(e)\equiv\tilde{q}(e_{1},b_{2},e_{3},e_{4})-\tilde{q}(e_{1},a_{2},e_{3},e_{4})$, and so on.}Define
the vector of CDFs
\begin{align*}
\textbf{\ensuremath{\pi}}(y) & \equiv\left\{ \pi(y,d|z,w,x)\right\} _{(d,z,w,x)\in\{0,1\}\times\mathcal{Z}\times\mathcal{W}\times\mathcal{X}}\\
 & \equiv\left\{ (\pi(y,0|z,w,x),\pi(y,1|z,w,x))'\right\} _{(z,w,x)\in\mathcal{Z}\times\mathcal{W}\times\mathcal{X}}
\end{align*}
and the linear operators $\tilde{R}_{\omega}:\tilde{\mathcal{Q}}\mathcal{\to\mathbb{R}}$
and $\tilde{R}_{0}:\tilde{\mathcal{Q}}\to\mathbb{R}^{d_{\textbf{\ensuremath{\pi}}}}$
(with $d_{\textbf{\ensuremath{\pi}}}$ being the dimension of $\textbf{\ensuremath{\pi}}$)
of $\tilde{q}(\cdot)$ that satisfy:
\begin{align*}
\tilde{R}_{\omega}\tilde{q} & \equiv E\left[\int\left(1-\int_{0}^{1}\int_{\{e:y(0,W)\le y\}}\tilde{q}(e|u,x)dedy\right)\omega_{1}(u,Z,X)du\right.\\
 & \left.-\int\left(1-\int_{0}^{1}\int_{\{e:y(1,W)\le y\}}\tilde{q}(e|u,x)dedy\right)\omega_{0}(u,Z,X)du\right],\\
(\tilde{R}_{0}\tilde{q})(y) & \equiv\left\{ \left(\begin{array}{c}
\int_{\mathcal{U}_{z,x}^{0}}\int_{\{e:y(0,w)\le y\}}\tilde{q}(e|u,x)dedu\\
\int_{\mathcal{U}_{z,x}^{1}}\int_{\{e:y(1,w)\le y\}}\tilde{q}(e|u,x)dedu
\end{array}\right)\right\} _{(z,w,x)\in\mathcal{Z}\times\mathcal{W}\times\mathcal{X}},
\end{align*}
where the expectation is taken over $(W,Z,X)$ and $\text{\ensuremath{\mathcal{U}_{z,x}^{1}=[0,P(z,x)]}}$
and $\mathcal{U}_{z,x}^{0}=(P(z,x),1]$.

\begin{definition}The identified set of $\tau$ is defined as
\[
\mathcal{T}^{*}\equiv\{\tau\in\mathbb{R}:\tau=\tilde{R}_{\omega}\tilde{q}\text{ for some }\tilde{q}\in\tilde{\mathcal{Q}}\text{ such that }(\tilde{R}_{0}\tilde{q})(y)=\textbf{\ensuremath{\pi}}(y)\text{ for all }y\in\mathcal{Y}\}.
\]
\end{definition}

Then the $\infty$-LP is formulated as:
\begin{align}
\overline{\tau} & =\sup_{\tilde{q}\in\tilde{\mathcal{Q}}}\tilde{R}_{\omega}\tilde{q}\label{eq:LP1_y}\\
\underline{\tau} & =\inf_{\tilde{q}\in\tilde{\mathcal{Q}}}\tilde{R}_{\omega}\tilde{q}\label{eq:LP2_y}
\end{align}
subject to 
\begin{equation}
(\tilde{R}_{0}\tilde{q})(y)=\textbf{\ensuremath{\pi}}(y)\qquad\text{for all }y\in\mathcal{Y}.\label{eq:LP3_y}
\end{equation}
Note that the LP is infinite dimensional not only because of $\tilde{q}$
but also \eqref{eq:LP3_y}, which consists of a continuum of constraints.

\subsection{Finite-Dimensional Linear Programming}

Analogous to Section \ref{sec:Sieve-Approximation-and}, we approximate
the unknown function $\tilde{q}(\cdot)$ using multivariate Bernstein
polynomials:
\[
\tilde{q}(e|u,x)\approx\sum_{\boldsymbol{k}=1}^{K}\theta_{\boldsymbol{k}}^{x}b_{\boldsymbol{k}}(e,u),
\]
where $b_{\boldsymbol{k}}(e,u)\equiv b_{\boldsymbol{k},K}(e,u)$ is
a 5-variate Bernstein polynomials with $\boldsymbol{k}\equiv(\boldsymbol{k}_{e},k_{u})$
and $\boldsymbol{k}_{e}\equiv(k_{00},k_{01},k_{10},k_{11})$ and its
coefficient $\theta_{\boldsymbol{k}}^{x}\equiv\theta_{\boldsymbol{k},K}^{x}\equiv\tilde{q}(\boldsymbol{k}_{e}/K|k_{u}/K,x)$.
Note that ``$\sum_{\boldsymbol{k}=1}^{K}$'' stands for ``$\sum_{k_{00},k_{01},k_{10},k_{11},k_{u}=1}^{K}$.''
Then the constraint can be written as a linear combination of the
unknown parameters $\left\{ \theta_{\boldsymbol{k}}^{x}\right\} _{(\boldsymbol{k},x)\in\mathcal{K}^{5}\times\mathcal{X}}$.
For example,
\begin{align}
\pi(y,1|z,w,x) & =\int_{0}^{P(z,x)}\int_{\{e:y(1,w)\le y\}}\tilde{q}(e|u,x)dedu\nonumber \\
 & =\sum_{\boldsymbol{k}=1}^{K}\theta_{\boldsymbol{k}}^{x}\int_{0}^{P(z,x)}\int_{\{e:y(1,w)\le y\}}b_{\boldsymbol{k}}(e,u)dedu\nonumber \\
 & \equiv\sum_{\boldsymbol{k}=1}^{K}\theta_{\boldsymbol{k}}^{x}\sigma_{\boldsymbol{k}}^{1}(y,z,x),\label{eq:deriv_constraint_y}
\end{align}
where $\sigma_{\boldsymbol{k}}^{d}(y,z,x)\equiv\int_{\mathcal{U}_{z,x}^{d}}\int_{\{e:y(1,w)\le y\}}b_{\boldsymbol{k}}(e,u)dedu$.
Similarly, the target parameter can be written as, for example,
\begin{align}
E\left[\tau_{0}(Z,W,X)\right] & =\sum_{(z,w,x)\in\{0,1\}\times\mathcal{W}\times\mathcal{X}}p(z,w,x)\int\left(1-\int_{0}^{1}\int_{\{e:y(0,w)\le y\}}\tilde{q}(e|u,x)dedy\right)\omega_{0}(u,z,x)du\nonumber \\
 & =\sum_{z,w,x}p(z,w,x)\int\omega_{0}(u,z,x)du\nonumber \\
 & \;\;\;\;\;\;-\sum_{w,x}\sum_{\boldsymbol{k}=1}^{K}\theta_{\boldsymbol{k}}^{x}\sum_{z\in\{0,1\}}p(z,w,x)\int\left(\int_{0}^{1}\int_{\{e:y(0,w)\le y\}}b_{\boldsymbol{k}}(e,u)dedy\right)\omega_{0}(u,z,x)du\nonumber \\
 & \equiv c_{0}-\sum_{w,x}\sum_{\boldsymbol{k}=1}^{K}\theta_{\boldsymbol{k}}^{x}\zeta_{\boldsymbol{k}}^{0}(w,x),\label{eq:deriv_target_y}
\end{align}
where $\zeta_{\boldsymbol{k}}^{d}(w,x)\equiv\sum_{z\in\{0,1\}}p(z,w,x)\int\left(\int_{0}^{1}\int_{\{e:y(d,w)\le y\}}b_{\boldsymbol{k}}(e,u)dedy\right)\omega_{d}(u,z,x)du$.

To address the challenge that the constraint \eqref{eq:deriv_constraint_y}
is indexed by continuous $y$, we proceed as follows. Note that, for
any measurable function $h:\mathcal{Y}\rightarrow\mathbb{R}$, $E|h(Y)|=0$
if and only if $h(y)=0$ almost everywhere in $\mathcal{Y}$ (\citet{beresteanu2011sharp}).
Therefore, the constraint (with general $d$) can be replaced by:
\[
E\left|\sum_{\boldsymbol{k}=1}^{K}\theta_{\boldsymbol{k}}^{x}\sigma_{\boldsymbol{k}}^{d}(Y,z,x)-\pi(Y,d|z,w,x)\right|=0,
\]
which is now a single constraint given $(d,z,w,x)$. In estimation,
the population mean can be replaced with the sample mean; see Section
\ref{subsec:Inference}. Now, redefine $\theta\equiv\{\theta_{\boldsymbol{k}}^{x}\}_{(x,\boldsymbol{k})\in\mathcal{X}\times\mathcal{K}^{5}}$
and 
\begin{align*}
\Theta_{K}\equiv & \left\{ \theta:0\leq\theta_{\boldsymbol{k}}^{x}\leq1\text{ }\forall(x,\boldsymbol{k}),\text{\ensuremath{\theta_{\boldsymbol{k}_{e},k_{u}}^{x}}}\text{ is 4-increasing in }\boldsymbol{k}_{e},\right.\\
 & \qquad\left.\theta_{\boldsymbol{K}_{e},k_{u}}^{x}=1,\theta_{\boldsymbol{k}_{e},k_{u}}^{x}=0\text{ for any }\boldsymbol{k}_{e}\text{ that contains }1\text{ }\forall(x,k_{u})\right\} ,
\end{align*}
where $\boldsymbol{K}_{e}\equiv(K,K,K,K)$.Then, the LP can be formulated
as
\begin{align}
\underline{\tau}_{K} & =\min_{\theta\in\Theta_{K}}\sum_{(w,x)\in\mathcal{W}\times\mathcal{X}}\sum_{\boldsymbol{k}=1}^{K}\theta_{\boldsymbol{k}}^{x}\left\{ -\zeta_{\boldsymbol{k}}^{1}(w,x)+\zeta_{\boldsymbol{k}}^{0}(w,x)\right\} \label{eq:LP1_y2}\\
\overline{\tau}_{K} & =\max_{\theta\in\Theta_{K}}\sum_{(w,x)\in\mathcal{W}\times\mathcal{X}}\sum_{\boldsymbol{k}=1}^{K}\theta_{\boldsymbol{k}}^{x}\left\{ -\zeta_{\boldsymbol{k}}^{1}(w,x)+\zeta_{\boldsymbol{k}}^{0}(w,x)\right\} \label{eq:LP2_y2}
\end{align}
subject to 
\begin{equation}
E\left|\sum_{\boldsymbol{k}=1}^{K}\theta_{\boldsymbol{k}}^{x}\sigma_{\boldsymbol{k}}^{d}(Y,z,x)-\pi(Y,d|z,w,x)\right|=0\quad\forall(d,z,w,x)\in\{0,1\}\times\mathcal{Z}\times\mathcal{W}\times\mathcal{X}.\label{eq:LP3_y2}
\end{equation}

\section{Simulation\label{sec:Simulation}}

This section provides numerical results to illustrate our theoretical
framework and to show the role of different identifying assumptions
in improving bounds on the target parameters. For target parameters,
we consider the ATE and the LATEs for always-takers (LATE-AT), never-takers
(LATE-NT), and compliers (LATE-C). We calculate the bounds on them
based only on the information from the data and then show how additional
assumptions (e.g., the existence of additional exogenous variables,
uniformity, and shape restrictions) tighten the bounds.

One important question we want to answer in the exercise is how the
current LP approach compares to that in \citet{mogstad2018using}.
This question is theoretically explored in Section \ref{sec:Identifying-Power-of},
where we showed that our LP approach can capture full independence's
stronger identifying power than mean independence. We show that the
simulation results are consistent with the theoretical finding.

\subsection{Data-Generating Process}

We generate the observables $(Y,D,Z,X,W)$ from the following data-generating
process (DGP). We assume that $W$ is a reverse IV, i.e., we maintain
Assumptions EX and SEL(a). We allow covariate $X$ to be endogenous.
All the variables are set to be binary with $\Pr\left[Z=1\right]=0.5$,
$\Pr\left[X=1\right]=0.6$ and $\Pr\left[W=1\right]=0.4$. The treatment
$D$ is determined by $Z$ and $X$ through the threshold crossing
model specified in Assumption SEL(a), where the propensity scores
$P(z,x)$ are specified as follows: $P(0,0)=0.1$, $P(1,0)=0.4$,
$P(0,1)=0.4$, and $P(1,1)=0.7$. The outcome $Y$ is generated from
$(D,X,W)$ through $Y=DY_{1}+(1-D)Y_{0}$. For the case of binary
$Y$, we generate $Y_{d}$ from
\begin{equation}
Y_{d}=1\left[m_{d}(U,X,W)\geq\epsilon\right]\label{eq:Y_sim}
\end{equation}
with the MTR functions are defined as
\begin{align*}
m_{0}(u,0,0) & =0.01b_{0}^{4}(u)+0.02b_{1}^{4}(u)+0.02b_{2}^{4}(u)+0.02b_{3}^{4}(u)+0.02b_{4}^{4}(u),\\
m_{1}(u,0,0) & =0.03b_{0}^{4}(u)+0.06b_{1}^{4}(u)+0.09b_{2}^{4}(u)+0.12b_{3}^{4}(u)+0.12b_{4}^{4}(u),\\
m_{0}(u,0,1) & =0.05b_{0}^{4}(u)+0.54b_{1}^{4}(u)+0.73b_{2}^{4}(u)+0.84b_{3}^{4}(u)+0.86b_{4}^{4}(u),\\
m_{1}(u,0,1) & =0.94b_{0}^{4}(u)+0.95b_{1}^{4}(u)+0.96b_{2}^{4}(u)+0.96b_{3}^{4}(u)+0.96b_{4}^{4}(u),\\
m_{0}(u,1,0) & =0.01b_{0}^{4}(u)+0.02b_{1}^{4}(u)+0.03b_{2}^{4}(u)+0.04b_{3}^{4}(u)+0.04b_{4}^{4}(u),\\
m_{1}(u,1,0) & =0.01b_{0}^{4}(u)+0.05b_{1}^{4}(u)+0.09b_{2}^{4}(u)+0.13b_{3}^{4}(u)+0.13b_{4}^{4}(u),\\
m_{0}(u,1,1) & =0.26b_{0}^{4}(u)+0.61b_{1}^{4}(u)+0.84b_{2}^{4}(u)+0.93b_{3}^{4}(u)+0.94b_{4}^{4}(u),\\
m_{1}(u,1,1) & =0.95b_{0}^{4}(u)+0.96b_{1}^{4}(u)+0.97b_{2}^{4}(u)+0.98b_{3}^{4}(u)+0.99b_{4}^{4}(u),
\end{align*}
where $b_{k}^{K}$ stands for the $k$-th basis function in the Bernstein
approximation of degree $K$. These MTR functions are consistent with
Assumptions M and C, i.e., to be weakly monotone and concave in $u$
for all $(d,x,w)\in\left\{ 0,1\right\} ^{3}$. Also, the DGP in \eqref{eq:Y_sim}
satisfies Assumption U$^{*}$ because $\epsilon$ does not depend
on $d=0,1$ and the MTR functions satisfy $m_{1}(u,x,w)>m_{0}(u,x,w)$
for all $(d,x,w)\in\left\{ 0,1\right\} ^{3}$. Therefore, the DGP
also satisfies Assumptions U and U$^{0}$. Following the second example
in Section \ref{subsec:Restrictions-on-Treatment Heterogeneity},
the DGP satisfies the following uniform order for the counterfactual
outcomes $Y(d,w)$: $Y(1,1)\ge Y(0,1)\geq Y(1,0)\ge Y(0,0)$ a.s.
The case with non-binary $Y$ has  DGPs with similar structure, which
we omit for succinctness. We generate a sample containing 1,000,000
observations and choose $K=50$. We choose the large sample size to
mimic the population. Our choice of $K$ is discussed below. The number
of unknown parameters $\theta$ in the linear programming is equal
to $\dim(\theta)=\left|\mathcal{E}\right|\times\left|\mathcal{X}\right|\times(K+1)$. 

\subsection{Comparison to \citet{mogstad2018using}\label{subsec:Identifying-Power-of}}

To illustrate the usefulness the current approach in incorporating
full independence of the IV, we make comparisons with \citet{mogstad2018using}.
Motivated by the theoretical results in Section \ref{sec:Identifying-Power-of},
we consider a range of cases for the support $\mathcal{Y}$ of $Y$.
Specifically, $Y$ takes values in $\{0,1\}$, $\{0,0.5,1\}$, $\{0,0.25,0.5,0.75,1\}$,
and $\{0,0.1,0.2,0.3,0.4,0.5,0.6,0.7,0.8,0.9,1\}$. We additionally
consider different supports $\mathcal{Z}$ of $Z$ and allow $Z$
taking values from $\{0,1\}$, $\{0,0.5,1\}$ and $\{0,0.25,0.5,0.75,1\}$.
Note that we intentionally fix the endpoints of $\mathcal{Y}$ and
$\mathcal{Z}$ to remove the effect of increased variation of the
variables. According to Section \ref{sec:Identifying-Power-of}, it
is conceivable that $Y$ departing from binary will deliver narrower
bounds in our approach than \citet{mogstad2018using}'s. The gain
from $Z$ departing from binary can be more subtle as the endpoints
are fixed. For different combinations of $\mathcal{Y}$ and $\mathcal{Z}$,
we derive the bounds on the ATE.
\begin{figure}
\centering{}\includegraphics[scale=0.35]{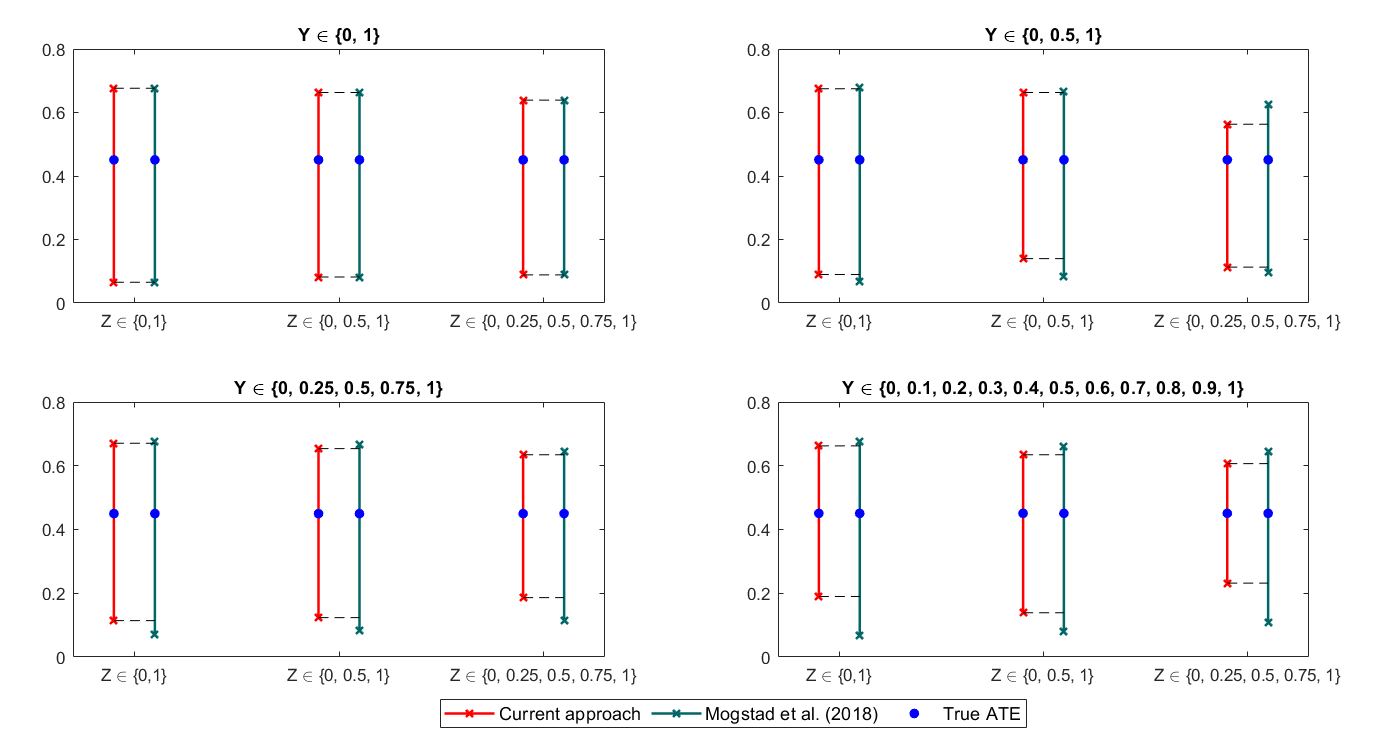}\caption{Full vs. Mean Independence: Bounds on ATE with Discrete $Z$ and $Y$}
\label{fig: MTE diffK-1}\label{fig: compare}
\end{figure}

Figure \ref{fig: compare} compares bounds that are calculated by
using the current approach with those that are replicated by using
\citet{mogstad2018using}'s approach. The first subfigure on the top-left
shows that the bounds are identical between the two methods, because
$Y$ is binary and thus full independence between $Y(d)$ and $Z$
is equivalent to mean independence. This equivalence is maintained
no matter how many values $Z$ takes. In the next three subfigures,
our bounds are tighter than \citet{mogstad2018using}'s as we expect.
They show a pattern that, as $\mathcal{Y}$ contains more values,
the improvement from the current approach is more substantial. This
is because, as $Y$ takes more values, with greater degree, full independence
contains richer structure than mean independence. Again, this pattern
remains to hold regardless of the choice of $\mathcal{Z}$. It would
also be interesting to investigate (i) the case with continuous $Y$
and (ii) the impact of $\mathcal{Z}$ when we move its endpoints further
apart. They are explored in Section \ref{subsec:More-results-on}
of the Appendix.

\subsection{Bounds under Different Assumptions}

\subsubsection{ATE}

Focusing on binary $Y$, Table \ref{tab:ATE} contains the bounds
on the ATE under different assumptions, and these bounds are illustrated
in Figure \ref{fig: ATE-1} and \ref{fig: ATE-2}. The true ATE value
is $0.15$, depicted as the solid red line in the figure. From \ref{fig: ATE-1},
the worst-case bounds on the ATE with no additional assumptions (and
without using variation from $W$) are $[-0.25,0.45]$. Since the
mappings do not involve $W$, we have $|\mathcal{E}|=4$, and the
linear programming is solved with $\dim(\theta)=\left|\mathcal{E}\right|\times\left|\mathcal{X}\right|\times(K+1)=4\times2\times51=408$.

For comparison, we calculate the bounds that incorporate the existence
of $W$. We express the target parameters with mappings involving
$W$ and use data distribution conditional on $W=0$ and $W=1$ as
the constraints. With binary $W$, we have $|\mathcal{E}|=16$, which
gives $\dim(\theta)=\left|\mathcal{E}\right|\times\left|\mathcal{X}\right|\times(K+1)=16\times2\times51=1,632$.
The resulting bounds are depicted in the dotted greenish-blue line.
When the variation from $W$ is used, the bounds on the ATE are $[-0.21,0.42]$,
which is narrower than without using $W$. This result is consistent
with our theoretical finding presented in Theorem \ref{lem:nonredundancy}
that $W$ can help tighten the bounds as long as it is a relevant
variable. Nonetheless, these worst-case bounds are not that informative,
e.g., they do not determine the sign of the ATE. 

Next, we impose Assumption U$^{0}$ without $W$ and with $W$.\footnote{Assumption U and U$^{0}$ give the same bounds in our exercise, therefore,
we use the weaker assumption and present the results.} Under Assumption U$^{0}$, the bounds on the ATE are tightened as
we incorporate extra inequality constraints according to the direction
of monotonicity. As mentioned in Section \ref{subsec:Restrictions-on-Treatment Heterogeneity},
the direction of monotonicity in Assumption U$^{0}$ is determined
by the LPs. We solve the LPs with different directions imposed, then
choose the one with a feasible solution. This means that the corresponding
direction of monotonicity is consistent with the DGP. Under Assumption
U$^{0}$, we obtain a bound $[0.05,0.45]$, which is narrower comparing
with the worst-case bound. With $W$, under Assumption U$^{0}$, the
bounds become $[0.05,0.42]$. In Figure \ref{fig: ATE-1}, these bounds
under Assumptions U$^{0}$ without and with $W$ are depicted as violet
and green dashed lines, respectively. Both sets of bounds identify
the sign of the ATE, consistent with the theoretical discussion. The
improvement is mainly on the lower bounds and the upper bounds coincide
with the corresponding worst-case upper bound without and with $W$.
These improvements come from the ability to identify the sign under
the uniformity assumptions. 

Next, we impose the shape restrictions (Assumptions M and C). As discussed
in Section \ref{subsec:Shape-Restrictions}, these assumptions can
be easily incorporated in the linear programming by directly imposing
inequality constraints on $\theta$. Under these assumptions (and
the existence of $W$), the bounds on the ATE shrink to $[0.12,0.19]$,
which is displayed with the pink line in Figure \ref{fig: ATE-1}.
We find that shape restrictions are powerful assumptions and yield
narrower bounds compared to those with uniformity assumptions. They
function differently in the linear programming: unlike the uniformity
assumption, which maintains the ranking of individuals across counterfactual
groups, shape restrictions directly control the MTR functions.

\begin{table}[H]
\caption{\label{tab:ATE}Bounds on ATE under Various Assumptions}

\bigskip{}

\begin{centering}
{\small{}}%
\begin{tabular}{c|cc|c}
\hline 
Assumptions & Lower Bound & Upper Bound & True Value\tabularnewline
\hline 
\hline 
{\small{}No Assumption} & -0.25 & {\small{}0.45} & {\small{}0.15}\tabularnewline
\hline 
Incorporating $W$ & -0.21 & {\small{}0.42} & {\small{}0.15}\tabularnewline
\hline 
Assumption U$^{0}$ & {\small{}0.05} & {\small{}0.45} & {\small{}0.15}\tabularnewline
\hline 
Incorporating $W$+Assumption U$^{0}$ & {\small{}0.05} & {\small{}0.42} & {\small{}0.15}\tabularnewline
\hline 
Incorporating $W$+Monotonicity+Concavity & {\small{}0.12} & {\small{}0.19} & {\small{}0.15}\tabularnewline
\hline 
Incorporating $W$+Assumption U$^{*}$ & {\small{}0.05} & {\small{}0.38} & {\small{}0.15}\tabularnewline
\hline 
\end{tabular}{\small\par}
\par\end{centering}
\bigskip{}
\end{table}

\begin{figure}
\centering{}\hspace{-0.7cm}\includegraphics[scale=0.5]{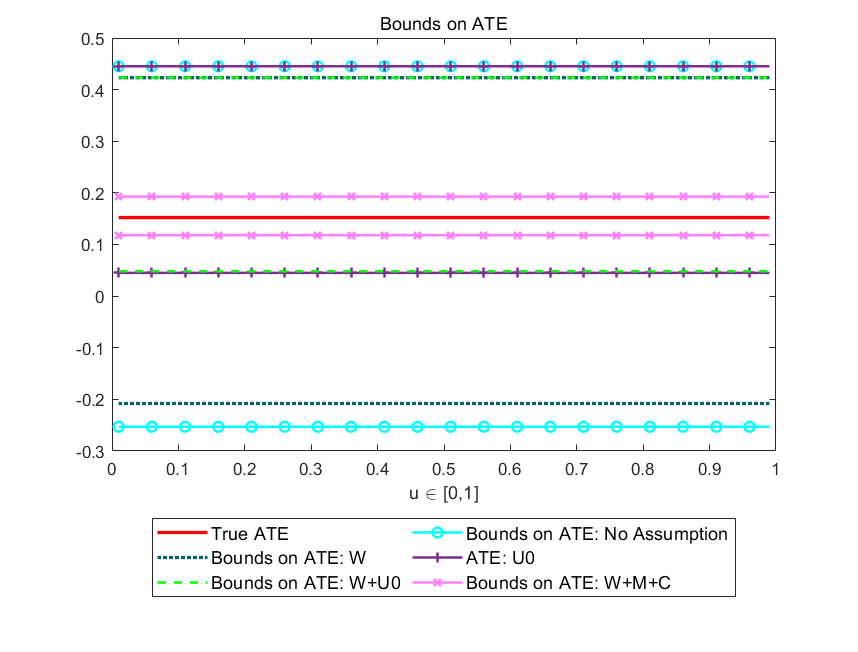}\caption{Bounds on the ATE under Different Assumptions}
\label{fig: ATE-1}
\end{figure}

Figure \ref{fig: ATE-2} presents the results under Assumption U$^{0}$,
versus under Assumption U$^{*}$ with existence of $W$. Under Assumption
U$^{*}$, the bounds become $[0.05,0.38]$. While their lower bounds
coincide, Assumption U$^{*}$ yields a lower upper bound compared
to Assumption U$^{0}$.

\begin{figure}
\centering{}\hspace{-0.7cm}\includegraphics[scale=0.5]{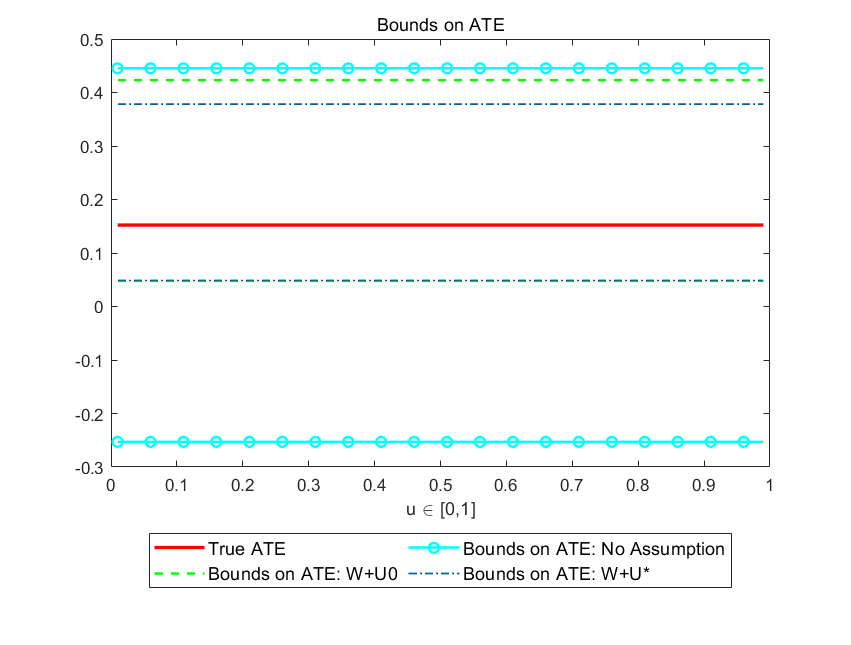}\caption{Bounds on the ATE under Different Assumptions}
\label{fig: ATE-2}
\end{figure}

\subsubsection{Generalized LATEs}

Next, we construct bounds on the generalized LATEs. Again, we focus
on binary $Y$. The original definition of the LATE is the ATE for
compliers (C). Researchers may also have interests in other local
treatment effects. We consider two other parameters---LATEs for always-takers
(AT) and never-takers (NT). Figure \ref{fig: LATEs-1} and \ref{fig: LATEs-2}
display the bounds on the LATE-AT, LATE-C, and LATE-NT under different
assumptions. This analysis is analogous to that with the ATE. Since
the covariate $X$ affects the decision of compliance, to avoid confusion
in the definition of the compliance groups, we instead establish bounds
on the LATEs conditional on $X$. We draw the conditional MTE functions
with solid red lines in both panels as a reference. 

The DGP implies constant MTE function, therefore, the LATE-AT, LATE-C
and LATE-NT are all equivalent to true ATE, equaling to $0.33,0.23,0.13$
and $0.20,0.11,0.08$, conditional on $X=0$ and $X=1$ respectively.
The feature that there exists no defiers in the DGP is known. When
there is no defier, the LATE-C is point identified, which has an analytical
expression of the two-stage least squares estimand. Therefore, even
when we add the tuning parameters\footnote{The tuning parameters are used to prevent infeasibility in the LP
due to the sampling error. The details are introduced in Section \ref{subsec:Inference}.}, the estimates remain very close to the true values throughout. And
when we do not need tuning parameters to adjust the numerical errors
or when the tuning parameters are very small, the linear programming
yields point estimates as shown in Figure \ref{fig: LATEs-1}. 
\begin{figure}
\centering{}\includegraphics[scale=0.3]{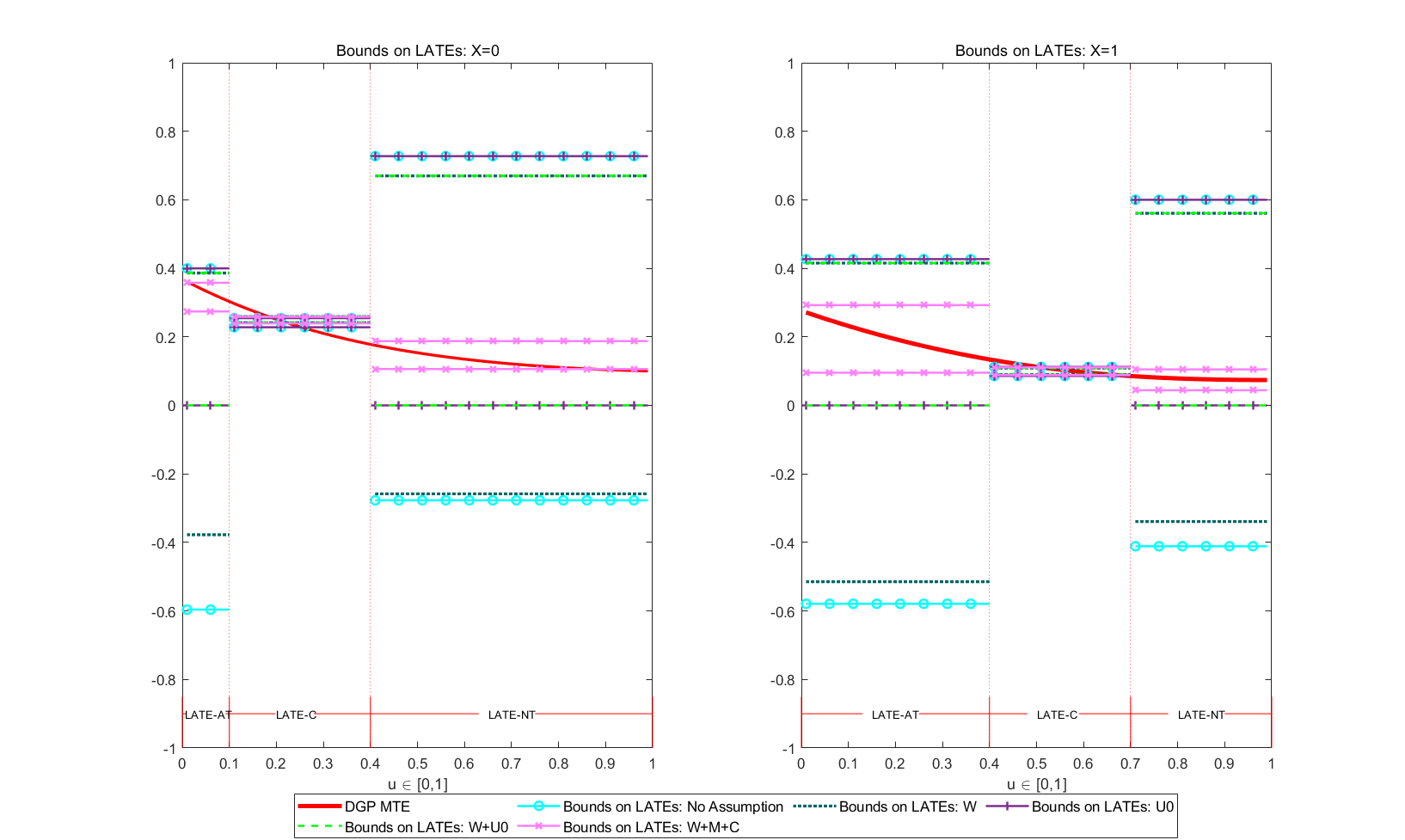}\caption{Bounds on the LATEs under Different Assumptions}
\label{fig: LATEs-1}
\end{figure}

For the LATE-AT and LATE-NT, as before, we first consider the worst-case
bounds where the existence of $W$ is ignored versus where $W$ is
taken into account. Without $W$, we get the bounds $[-0.60,0.40]$
and $[-0.28,0.73]$ on the LATE-AT and the LATE-NT conditional on
$X=0$, and $[-0.58,0.43]$ and $[-0.41,0.60]$ conditional on $X=1$;
with $W$, we get the bounds $[-0.38,0.39]$ and $[-0.26,0.67]$ on
the LATE-AT and the LATE-NT conditional on $X=0$, and $[-0.51,0.42]$
and $[-0.34,0.56]$ conditional on $X=1$.\emph{ }Incorporating information
from $W$ helps improve both the upper and lower bounds. Imposing
Assumption U$^{0}$ without $W$ helps to identify the sign by raising
up the lower bound to $0$, for the LATEs; when considering the case
with $W$, the pattern remains the same, but with a lower upper bound
and a slightly improved lower bound above $0$. We then apply M and
C with $W$ taking into account. The bounds on the LATE-AT and the
LATE-NT turn to $[0.27,0.36]$ and $[0.11,0.19]$ conditional on $X=0$,
and $[0.10,0.29]$ and $[0.04,0.11]$ conditional on $X=1$.

From Figure \ref{fig: LATEs-2}, under the Assumption U$^{*}$, the
bounds shrink to $[0,0.37]$ and $[0,0.53]$ conditional on $X=0$,
and $[0,0.40]$ and $[0,0.55]$ conditional on $X=1$, comparing with
the bounds under Assumption U$^{0}$. The improvement is from complete
order of 16 mapping types we have in this environment and is most
significant for the never-taker LATE upper bound.

\begin{figure}
\centering{}\includegraphics[scale=0.38]{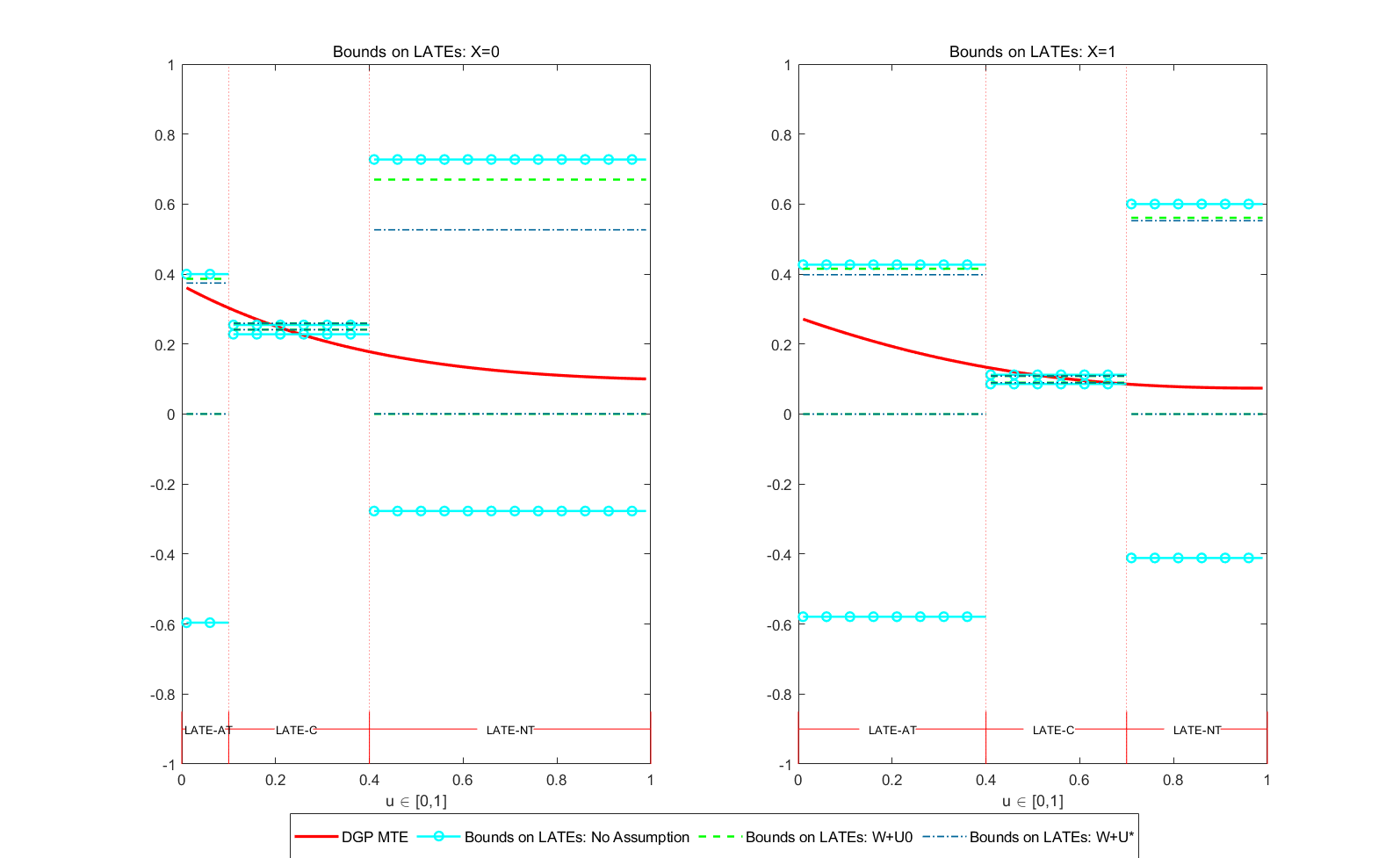}\caption{Bounds on the LATEs under Different Assumptions}
\label{fig: LATEs-2}
\end{figure}

\subsection{The Choice of $K$}

As a tuning parameter in the LP, we need to choose the order of Bernstein
polynomials, $K$. In general, $K$ should be chosen based on the
sample size and the smoothness of the function to be approximated,
in our case, $q(\cdot)$. The choice of the sieve dimension or more
generally, regularization parameters, is a difficult question (\citet{chen2007large})
and developing data-driven procedure is a subject of on-going research
in various nonparametric contexts of point identification; see, e.g.,
\citet{chen2018optimal} and \citet{han2020nonparametric}. In this
partial identification setup, we propose the following heuristic and
conservative approach, which is in spirit consistent with the very
motivation of partial identification.

First, we do not want to claim any prior knowledge about the smoothness
of $q(\cdot)$ because it is the distribution of a latent variable.
Because $K$ determines the dimension of unknown parameter $\theta$
in the linear programming, the width of the bounds tends to increase
with $K$. At the same time, the computational burden increases with
$K$. One interesting numerical finding is that, when $K$ is sufficiently
large, the increase of the width slows down and the bounds become
stable. This suggests that we may be able to conservatively choose
$K$ that acknowledges our lack of knowledge of the smoothness but,
at the same time, produces a reasonable computational task for the
linear programming. 
\begin{figure}
\centering{}\includegraphics[scale=0.26]{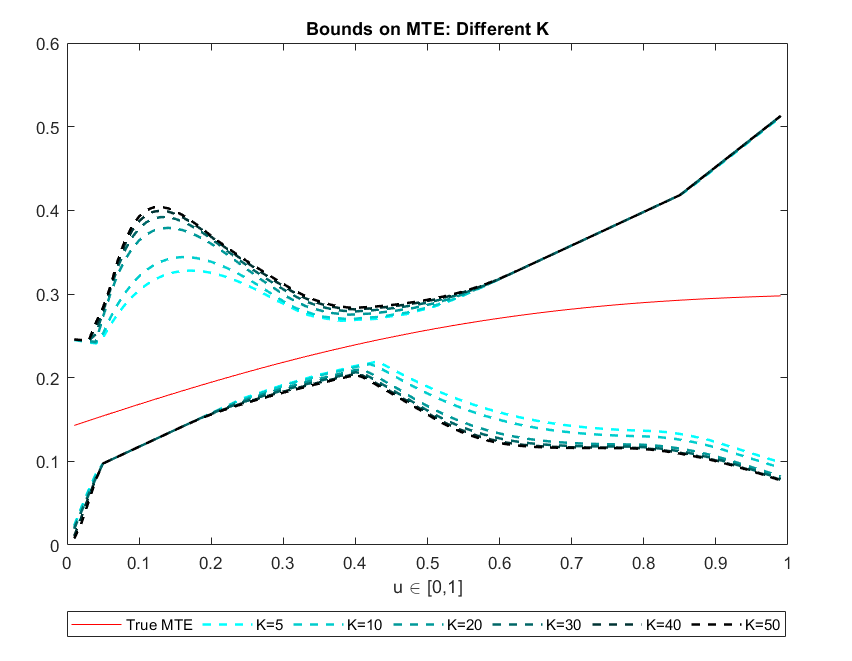}\includegraphics[scale=0.3]{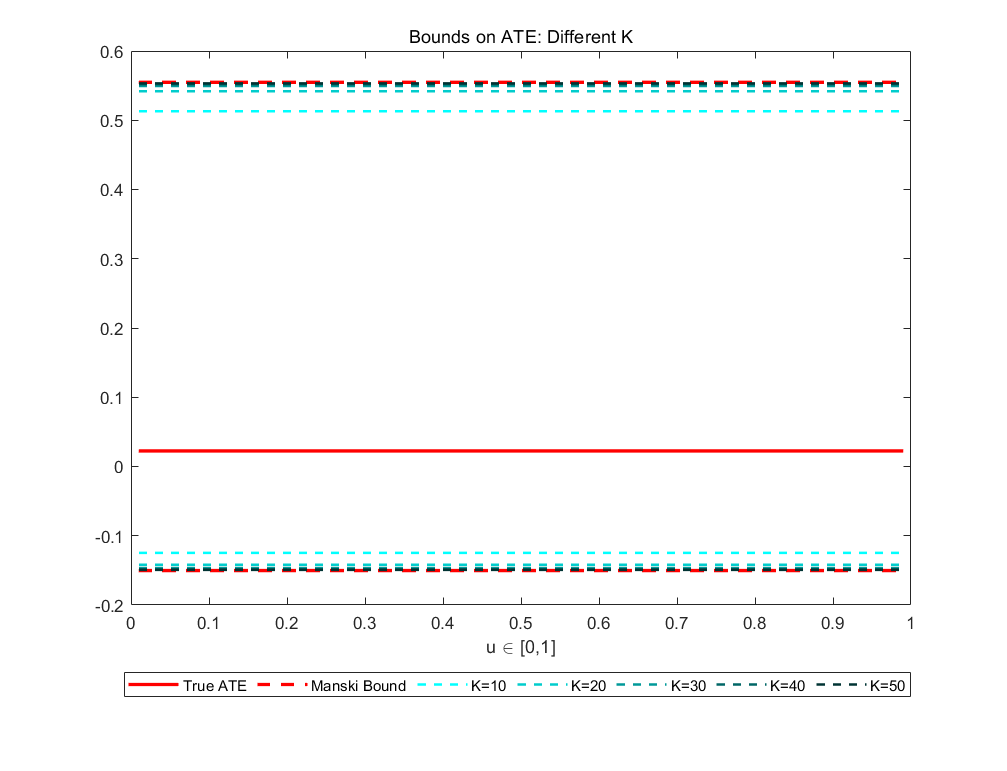}\caption{Bounds on MTE and ATE with Different $K$}
\label{fig: MTE diffK}
\end{figure}

To illustrate this point, we first consider the conditional MTE as
the target parameter and show how its bounds change as $K$ increases.
We consider the MTE because it is a fundamental parameter that generates
other target parameters, and hence, it is important to understand
the sensitivity of its bounds to $K$. The left panel of Figure \ref{fig: MTE diffK}
shows the evolution of the bounds on the MTE as $K$ grows. We use
a DGP similar to the one described earlier. When $K=5$, the bounds
are narrow. Although it may be tempting to choose this value of $K$,
this attempt should be avoided as it may be subject to the misspecification
of the true smoothness. When $K$ increases beyond $30$, the bounds
start to converge and become stable. We choose $K=50$, and this is
the choice we made in our previous numerical exercises. 

To compare this converging pattern with a known benchmark, in the
right panel of Figure \ref{fig: MTE diffK}, we depict the identified
set for the ATE relative to Manski's analytical bounds (\citet{manski1990nonparametric}).
We observe the identified set approaches to Manski's bound as $K$
increases and it almost overlaps with Manski's bounds when $K$ is
around 50.\footnote{Note that with large $K$, some LP solvers would ignore coefficients
with negligible (e.g., $10^{-13}$) values that cause a large range
of magnitude in the coefficient matrix. It may be recommended to simultaneously
rescale a column and a row to achieve a smaller range in the coefficients;
see Section \ref{subsec:Rescaling-of-Linear} for details. We found
that when $K=50$, the bounds from the rescaled LP and original LP
are very close to each other (e.g., for the ATE bounds without extra
assumptions, the difference is up to 0.01).}

As discussed in Section \ref{subsec:Point-wise-and-Uniform} in the
Appendix, it is worth mentioning that the bounds on the MTE are pointwise
sharp but \textit{not} uniformly sharp. The graph for the MTE bounds
are drawn by calculating the pointwise sharp bounds on MTE at each
point of $u$ (after properly discretizing it) and then connecting
them. Therefore, these bounds should \textit{not} be viewed as uniformly
sharp bounds. Nonetheless, this graph is still useful for the purpose
of our illustration. Given the current DGP, we find that there are
no uniformly sharp bounds for the MTE.

\section{Empirical Application\label{sec:Empirical-Application}}

It is widely recognized in the empirical literature that health insurance
coverage can be an essential factor for the utilization of medical
services (\citet{hurd1997medical}; \citet{dunlop2002gender}; \citet{finkelstein2012oregon,taubman2014medicaid}).
Prior studies on this topic typically make use of parametric econometric
models for the analysis. In their application, \citet{han2019estimation}
relax this common approach by introducing a semiparametric bivariate
probit model to measure the average effect of insurance coverage on
patients' medical visits. By applying our theoretical framework of
partial identification, we further relax the parametric and semiparametric
structures used in these studies. More importantly, we try to understand
how much we can learn about the effect of insurance that is utilized
through various counterfactual policies by learning the effect of
different compliance groups.

We use the 2010 wave of the Medical Expenditure Panel Survey (MEPS)
and focus on all the medical visits in January 2010. The sample is
restricted to contain individuals aged between 25 and 64 and exclude
those who had any kind of federal or state insurance in 2010. The
outcome $Y$ is a binary variable indicating whether or not an individual
has visited a doctor's office; the treatment $D$ is whether an individual
has private insurance. We choose whether a firm has multiple locations
as the binary instrument $Z$. This IV reflects the size of the firm,
and larger firms are more likely to provide fringe benefits, including
health insurance. On the other hand, the number of branches of a firm
does not directly affect employee decisions about medical visits.
To justify the IV, self-employed individuals are excluded. For potentially
endogenous covariates $X$, we include the age being 45 and older,
gender, income above median. Lastly, for an exogenous covariate $W$,
we use the percentage of workers who are provided with paid sick leave
benefits within each industry. Following \citet{han2019estimation},
we assume $W$ satisfies Assumptions SEL$_{W}$(b) and EX$_{W}$(b),
as $X$ is controlled. The rationale is the following: First, we assume
$W$ is exogenous (conditional on covariates) arguably because it
is determined by the employer or is in accordance with the local legislation
and thus is not correlated with individual employee's preferences.
However, due to its nature, it can influence the employee's health-related
decisions, such as enrolling in an insurance program ($D$) or utilizing
medical services ($Y$).\footnote{Since the relevance of $W$ to $D$ is slightly less plausible than
to $Y$, we test whether the propensity score is a not a function
of $W$ but cannot reject the null. Therefore, we decided to use $W$
as a common exogenous variable than a reverse IV.} We construct a categorical variable such that $W=0$ for less than
median value of the pay sick leave provision, $W=1$ for above the
median.

\begin{table}[H]
\caption{\label{tab:Summary}Summary Statistics}

\bigskip{}

\begin{centering}
{\small{}}%
\begin{tabular}{cccccc}
\hline 
 & {\small{}Variables} & {\small{}Mean} & {\small{}S.D} & {\small{}Min} & {\small{}Max}\tabularnewline
\hline 
{\small{}$Y$} & {\small{}Whether or not visit doctors} & {\small{}0.18} & {\small{}0.39} & {\small{}0} & {\small{}1}\tabularnewline
\hline 
{\small{}$D$} & {\small{}Whether or not have insurance} & {\small{}0.66} & {\small{}0.47} & {\small{}0} & {\small{}1}\tabularnewline
\hline 
{\small{}$Z$} & {\small{}Firm has multiple locations} & {\small{}0.68} & {\small{}0.47} & {\small{}0} & {\small{}1}\tabularnewline
\hline 
\multirow{3}{*}{{\small{}$X$}} & {\small{}Age above 45} & {\small{}0.41} & {\small{}0.49} & {\small{}0} & {\small{}1}\tabularnewline
 & {\small{}Gender} & {\small{}0.50} & {\small{}0.50} & {\small{}0} & {\small{}1}\tabularnewline
 & {\small{}Income above median} & {\small{}0.50} & {\small{}0.50} & {\small{}0} & {\small{}1}\tabularnewline
\hline 
$W$ & Pay sick leave provision & 0.49 & 0.50 & 0 & 1\tabularnewline
\hline 
\multicolumn{6}{c}{{\small{}Number of observations = 7,555}}\tabularnewline
\hline 
\end{tabular}{\small\par}
\par\end{centering}
\bigskip{}
\end{table}

First, as a benchmark, we report that the LATE-C estimate calculated
via our linear programming approach is equal to a singleton of 0.05,
which is in fact identical to the 2SLS estimate we separately calculate.
In what follows, we extrapolate this LATE beyond the complier group
to the ATE. The presence of covariates reduces the effective sample
size and thus leads to larger sampling errors in estimating the $p$
of the $\infty$-LP \eqref{eq:upper}--\eqref{eq:constr}. This may
create inconsistencies in the set of equality constraints \eqref{eq:constr},
resulting in no feasible solution. This is in fact what happens in
this application. To resolve this estimation problem, we introduce
a slackness parameter $\kappa$ and modify \eqref{eq:constr} so that,
with some slackness, it satisfies
\begin{align}
||\hat{R}_{0}q-\hat{p}|| & \leq\inf_{q\in\mathcal{Q}}||\hat{R}_{0}q-\hat{p}||+\kappa,\label{eq:LP3_eta}
\end{align}
where $\hat{R}_{0}$ and $\hat{p}$ are estimates of $R_{0}$ and
$p$. A similarly modified constraint can then be followed in the
finite-dimensional LP after approximation, as well as by combining
\eqref{eq:LP4}--\eqref{eq:LP5}. The appropriate value of $\kappa$
should depend on the sample size, the dimension of covariates, and
the dimension of the unknown parameter $\theta$. To explain the latter,
as $K$ increases, the dimension of $\theta$ (i.e., unknowns) increases,
while the number of constraints (i.e., simultaneous equations for
the unknowns) is fixed. Therefore, as $K$ increases, the chance that
the LP does not have a feasible solution would decrease. Based on
the method discussed in the previous section, we set $K=50$ in this
application.

We calculate worst-case bounds on the ATE, as well as bounds after
imposing Assumptions U$^{0}$ and M and after using $W$. Under Assumption
U$^{0}$, the data rules out the possibility that $Y(0)>Y(1)$, indicating
that individuals with private insurance are more likely to visit a
doctor. Assumption M imposes that the MTR function is weakly increasing
in $U=u$. Usually, $U$ is interpreted as the latent cost of obtaining
treatment. \citet{kowalski2021reconciling} interpreted $U$ as eligibility
in a similar setup for Medicaid insurance. The eligibility for Medicaid
is related to income level and age. In our setup, because the treatment
is having the private insurance, we interpret the eligibility as the
health status, which is reflected in the premium. Interpreting $U$
as a latent cost (e.g., premium) of getting private insurance, Assumption
M states that the chance of making a medical visit (with or without
insurance) increases for those with higher cost. This is a reasonable
assumption given that sicker individuals typically face higher insurance
costs and also visit doctors more often. We choose the slackness parameter
$\kappa$ to be consistently 0.01 under all assumptions for a comparable
comparison.
\begin{figure}
\centering{}\includegraphics[scale=0.7]{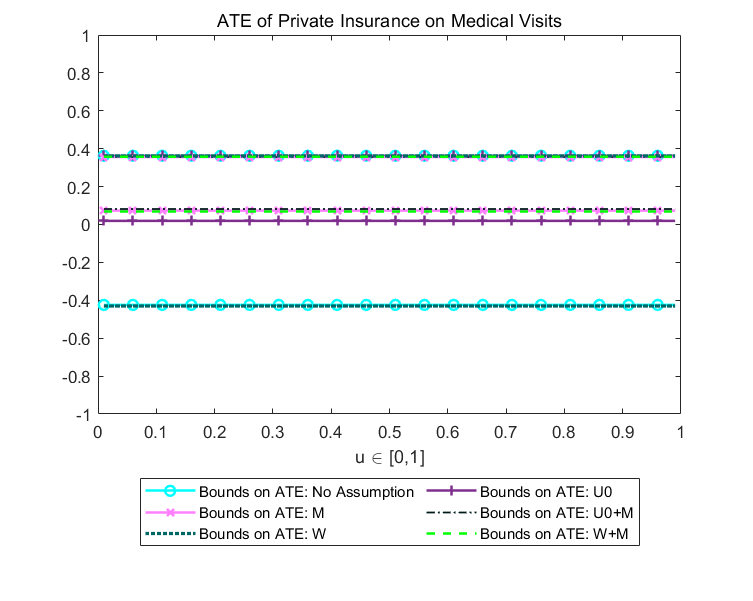}\caption{Bounds on the ATE of Private Insurance on Medical Visits}
\label{fig: application ATE}
\end{figure}

The bounds on the ATE are shown in Figure \ref{fig: application ATE}.
The worst-case bound on the ATE equals $[-0.42,0.36]$. The bounds
become $[0.02,0.36]$ under Assumption U$^{0}$ and $[0.07,0.36]$
under Assumption M. It is interesting to note that the identifying
power of the uniformity and the shape restriction is similar in this
example. When both Assumption U$^{0}$ and Assumption M are imposed,
the bounds are further tightened to $[0.08,0.36]$, although not substantially,
indicating that the two assumptions are complementary. However, we
do not see gain from incorporating $W$ in this case.

\begin{figure}
\centering{}\includegraphics[scale=0.7]{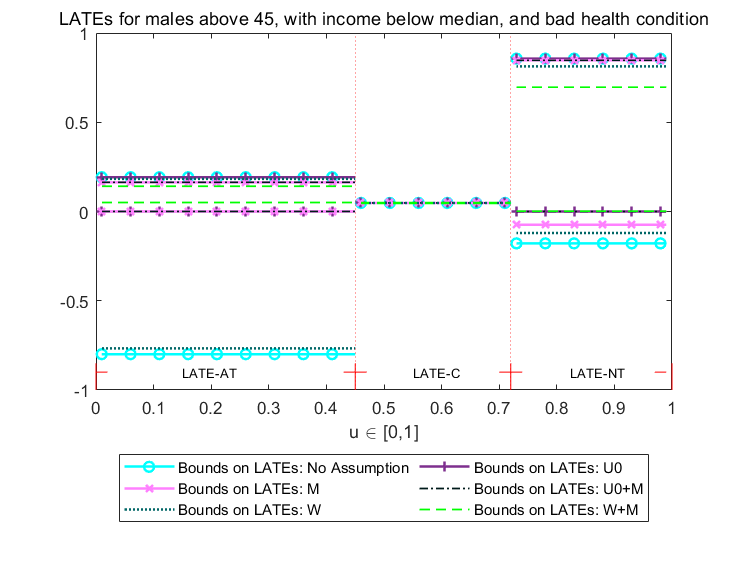}\caption{Bounds on the generalized LATEs of Private Insurance on Medical Visits
for Male Above 45, with Income Below Median}
\label{fig: application LATEs}
\end{figure}
\begin{table}
\caption{\label{tab:pe_ates-3234}Estimated Bounds on generalized LATEs for
Males Above 45, with Income Below Median}

\bigskip{}

\begin{centering}
\begin{tabular}{ccccccc}
\hline 
 & {\footnotesize{}No Assumption} & {\footnotesize{}Assumption U$^{0}$} & {\footnotesize{}M} & {\footnotesize{}Assumption U$^{0}$+M} & {\footnotesize{}W} & {\footnotesize{}M+W}\tabularnewline
\hline 
{\footnotesize{}LATE-AT} & {\footnotesize{}{[}-0.80,0.19{]}} & {\footnotesize{}{[}0,0.19{]}} & {\footnotesize{}{[}0,0.16{]}} & {\footnotesize{}{[}0,0.16{]}} & {\footnotesize{}{[}-0.77,0.18{]}} & {\footnotesize{}{[}0.05,0.14{]}}\tabularnewline
{\footnotesize{}LATE-C} & {\footnotesize{}0.05} & {\footnotesize{}0.05} & {\footnotesize{}0.05} & {\footnotesize{}0.05} & {\footnotesize{}0.05} & {\footnotesize{}0.05}\tabularnewline
{\footnotesize{}LATE-NT} & {\footnotesize{}{[}-0.18,0.86{]}} & {\footnotesize{}{[}0,0.86{]}} & {\footnotesize{}{[}-0.07,0.85{]}} & {\footnotesize{}{[}0,0.85{]}} & {\footnotesize{}{[}-0.12,0.81{]}} & {\footnotesize{}{[}0,0.70{]}}\tabularnewline
\hline 
{\footnotesize{}Slackness $\kappa$} & {\footnotesize{}0.01} & {\footnotesize{}0.01} & {\footnotesize{}0.01} & {\footnotesize{}0.01} & {\footnotesize{}0.01} & {\footnotesize{}0.01}\tabularnewline
\hline 
\multicolumn{7}{c}{{\footnotesize{}Number of observations = 7,555}}\tabularnewline
\hline 
\end{tabular}
\par\end{centering}
\bigskip{}
\end{table}

Next, we consider the always-taker, complier, and never-taker LATEs.
We consider these generalized LATEs conditional on $X=x$. Specifically,
we focus on the treatment effects for males above age 45, with income
below the median. The results are shown in Table \ref{tab:pe_ates-3234}
and depicted in Figure \ref{fig: application LATEs}. The LATE-C is
analytically calculated via TSLS.\footnote{When the alternative constraint \eqref{eq:LP3_eta} is used with the
slackness parameter, the LATE-C is no longer a singleton.} For the LATE-AT and LATE-NT, Assumption U$^{0}$ identifies the sign
of the effects, and Assumption M nearly identifies it. Using the variation
in $W$ mostly improves the bounds compared to the ones without it.\footnote{Most of the extra assumptions we impose help to determine the direction
of treatment effect, i.e., to raise the lower bound if the treatment
effect is positive. Therefore, improvements on LATE-NT are smaller
than LATE-AT after imposing extra assumptions, since the evidence
of positive treatment effect is relatively strong even with the worst-case
bounds of LATE-NT. } From the results we can conclude that, for a range of identifying
assumptions, the private insurance tends to have a large effect on
medical visits for never-takers, that is, people who face higher insurance
cost. For example, for all the cases, the upper bound on the effect
(i.e., the most optimistic scenario) is much larger for the never-takers
than the always-takers. Under Assumption W or no assumption, the lower
bounds (i.e., the most pessimistic scenario) shows a similar pattern.
This suggests a policy implication that lowering the cost of private
insurance may be important, because high costs may hinder those with
the most need from receiving enough medical services.

\begin{appendix}

\section{Examples of the Target Parameters\label{sec:Examples-of-the}}

Table \ref{tab:GLATE} contains the list of target parameters.
\begin{table}
\hspace{-1cm}%
\begin{tabular}{cccc}
\hline 
\multirow{2}{*}{{\scriptsize{}Target Parameters}} & \multirow{2}{*}{{\scriptsize{}Expressions}} & \multirow{2}{*}{{\scriptsize{}Ranges of $u$}} & {\scriptsize{}Weights}\tabularnewline
\cline{4-4} 
 &  &  & {\scriptsize{}$w_{d}(u,z,x)$}\tabularnewline
\hline 
{\scriptsize{}Average Treatment Effect} & \multirow{2}{*}{{\scriptsize{}$E[Y(1)-Y(0)]$}} & \multirow{2}{*}{{\scriptsize{}$[0,1]$}} & \multirow{2}{*}{{\scriptsize{}$1$}}\tabularnewline
{\scriptsize{}(ATE)} &  &  & \tabularnewline
{\scriptsize{}LATE for Compliers} & \multirow{2}{*}{{\scriptsize{}$E\left\{ Y(1)-Y(0)|u\in\left[P(z_{0},x),P(z_{1},x)\right]\right\} $}} & \multirow{2}{*}{{\scriptsize{}$[P(z_{0},x),P(z_{1},x)]$}} & \multirow{2}{*}{{\scriptsize{}$\frac{1\left(u\in[P(z_{0},x),P(z_{1},x)]\right)}{P(z_{1},x)-P(z_{0},x)}$}}\tabularnewline
{\scriptsize{}(LATE-C) given $x\in\mathcal{X}$} &  &  & \tabularnewline
{\scriptsize{}LATE for Always-Takers} & \multirow{2}{*}{{\scriptsize{}$E\left\{ Y(1)-Y(0)|u\in\left[0,P(z_{0},x)\right]\right\} $}} & \multirow{2}{*}{{\scriptsize{}$[0,P(z_{0},x)]$}} & \multirow{2}{*}{{\scriptsize{}$\frac{1\left(u\in[0,P(z_{0},x)]\right)}{P(z_{0},x)}$}}\tabularnewline
{\scriptsize{}(LATE-AT) given $x\in\mathcal{X}$} &  &  & \tabularnewline
\multicolumn{1}{c}{{\scriptsize{}LATE for Never-Takers}} & \multirow{2}{*}{{\scriptsize{}$E\left\{ Y(1)-Y(0)|u\in\left[P(z_{1},x),1\right]\right\} $}} & \multirow{2}{*}{{\scriptsize{}$[P(z_{1},x),1]$}} & \multirow{2}{*}{{\scriptsize{}$\frac{1\left(u\in[P(z_{1},x),1]\right)}{1-P(z_{1},x)}$}}\tabularnewline
{\scriptsize{}(LATE-NT) given $x\in\mathcal{X}$} &  &  & \tabularnewline
{\scriptsize{}LATE for $[\underline{u},\overline{u}]$} & {\scriptsize{}$E[Y(1)-Y(0)|u\in[\underline{u},\overline{u}]]$} & {\scriptsize{}$\left[P(z_{0},x),P(z_{1},x)\right]$} & {\scriptsize{}$\frac{1\left(u\in[\underline{u},\overline{u}]\right)}{\overline{u}-\underline{u}}$}\tabularnewline
{\scriptsize{}Marginal Treatment Effect} & \multirow{2}{*}{{\scriptsize{}$E[Y(1)-Y(0)|u']$}} & \multirow{2}{*}{{\scriptsize{}$u'$}} & \multirow{2}{*}{{\scriptsize{}$1(u=u')$}}\tabularnewline
{\scriptsize{}(MTE)$^{*}$} &  &  & \tabularnewline
{\scriptsize{}Policy Relevant Treatment Effect} & \multirow{2}{*}{{\scriptsize{}$\frac{E(Y')-E(Y)}{E(D')-E(D)}$}} & \multirow{2}{*}{{\scriptsize{}$[0,1]$}} & \multirow{2}{*}{{\scriptsize{}$\frac{\Pr\left[u\leq P'(z')\right]-\Pr\left[u\leq P'(z)\right]}{E\left[P(Z')\right]-E\left[P(Z)\right]}$}}\tabularnewline
{\scriptsize{}(PRTE) for a new policy $(P',Z')$} &  &  & \tabularnewline
\hline 
\end{tabular}

\smallskip{}

{\footnotesize{}{*} The MTE uses the Dirac measure at $u'$, while
the other target parameters use the Lebesgue measure on $[0,1]$.}{\footnotesize\par}

\caption{Examples of the Target Parameters}

\label{tab:GLATE}
\end{table}

\section{Further Discussions\label{sec:Discussions}}

\subsection{Rescaling of Linear Programs\label{subsec:Rescaling-of-Linear}}

Let $B\theta=p$ represents the constraints \eqref{eq:constr2} in
the LP \eqref{eq:upper2}--\eqref{eq:constr2}. In practice, the
matrix $B$ has the number of columns that grows with $K$. An important
consequence is that, when $K$ is large, the entries of $B$ (i.e.,
constraint coefficients) take values of very different orders of magnitude;
some coefficients are too small and some are too large. In this case,
many optimization algorithms do not work properly because, to address
the issue, they arbitrarily drop coefficients with small values (e.g.,
GUROBI drops coefficients that are less than $10^{-13}$). This may
arbitrarily change the bounds we obtain. In this section, we propose
a rescaling method to address this problem.

To better understand the rescaling strategy, we first express the
original LP \eqref{eq:upper2}--\eqref{eq:constr2} in terms of matrices:
\[
\max_{\theta\in\Theta_{K}}A\theta
\]
subject to 
\[
B\theta=p.
\]
Here, $\theta$ is defined as a vector of unknown parameters $\{\theta_{k}^{e,x}\}_{k,e,x}$
and $\Theta_{K}$ is redefined as 
\[
\Theta_{K}\equiv\left\{ \theta:M\theta=\boldsymbol{1},\theta\geq\boldsymbol{0}\right\} ,
\]
where $M$ is a weight matrix corresponding to $\sum_{e\in\mathcal{E}}\theta_{k}^{e,x}=1$
$\forall(k,x)$, $\boldsymbol{1}$ is a column vector of ones, and
$\boldsymbol{0}$ is a zero vector.

Because the Bernstein polynomials are only used in generating the
coefficients in the equality restrictions from the data, we focus
on rescaling of this constraint. Suppose the dimension of $B$ is
$m\times n$ with $m<n$.\footnote{$m$ is determined by the dimension of $p$, which is determined by
the cardinality of $\mathcal{Y}$, $\mathcal{Z}$ and $\mathcal{X}$,
and $n$ is determined by the order of polynomials, $K$, we choose
in sieve approximation. Usually $K$ (and thus $n$) is set to be
a large number to guarantee the accuracy of sieve approximation, and
it usually is larger than $m$. When $m=n$, theoretically, we achieve
a unique solution, but in practice, the numerical error may cause
infeasibility.} First, we show that $B$ has full rank of $m$ in our setting. To
prove this, we need to understand the structure of $B$. The number
of columns of $B$ is determined by the size of $\mathcal{E}$ and
the order of polynomials $K$. The number of rows of $B$ is determined
by the dimension of $p$. We consider an example with binary $(Y,Z,W)$
for illustration. Since $(D,W)$ are binary, $\left|\mathcal{E}\right|=16$
and $B$ takes the form of the following:
\[
\begin{array}{ccccccccccccccccc}
e= & 1 & 2 & 3 & 4 & 5 & 6 & 7 & 8 & 9 & 10 & 11 & 12 & 13 & 14 & 15 & 16\\
Z=0,D=0,W=0 &  & \text{\textifsymbol[ifgeo]{96}} &  & \text{\textifsymbol[ifgeo]{96}} &  & \text{\textifsymbol[ifgeo]{96}} &  & \text{\textifsymbol[ifgeo]{96}} &  & \text{\textifsymbol[ifgeo]{96}} &  & \text{\textifsymbol[ifgeo]{96}} &  & \text{\textifsymbol[ifgeo]{96}} &  & \text{\textifsymbol[ifgeo]{96}}\\
Z=1,D=0,W=0 &  & \text{\textifsymbol[ifgeo]{96}} &  & \text{\textifsymbol[ifgeo]{96}} &  & \text{\textifsymbol[ifgeo]{96}} &  & \text{\textifsymbol[ifgeo]{96}} &  & \text{\textifsymbol[ifgeo]{96}} &  & \text{\textifsymbol[ifgeo]{96}} &  & \text{\textifsymbol[ifgeo]{96}} &  & \text{\textifsymbol[ifgeo]{96}}\\
Z=0,D=1,W=0 &  &  &  &  & \text{\textifsymbol[ifgeo]{96}} & \text{\textifsymbol[ifgeo]{96}} & \text{\textifsymbol[ifgeo]{96}} & \text{\textifsymbol[ifgeo]{96}} &  &  &  &  & \text{\textifsymbol[ifgeo]{96}} & \text{\textifsymbol[ifgeo]{96}} & \text{\textifsymbol[ifgeo]{96}} & \text{\textifsymbol[ifgeo]{96}}\\
Z=1,D=1,W=0 &  &  &  &  & \text{\textifsymbol[ifgeo]{96}} & \text{\textifsymbol[ifgeo]{96}} & \text{\textifsymbol[ifgeo]{96}} & \text{\textifsymbol[ifgeo]{96}} &  &  &  &  & \text{\textifsymbol[ifgeo]{96}} & \text{\textifsymbol[ifgeo]{96}} & \text{\textifsymbol[ifgeo]{96}} & \text{\textifsymbol[ifgeo]{96}}\\
Z=0,D=0,W=1 &  &  & \text{\textifsymbol[ifgeo]{96}} & \text{\textifsymbol[ifgeo]{96}} &  &  & \text{\textifsymbol[ifgeo]{96}} & \text{\textifsymbol[ifgeo]{96}} &  &  & \text{\textifsymbol[ifgeo]{96}} & \text{\textifsymbol[ifgeo]{96}} &  &  & \text{\textifsymbol[ifgeo]{96}} & \text{\textifsymbol[ifgeo]{96}}\\
Z=1,D=0,W=1 &  &  & \text{\textifsymbol[ifgeo]{96}} & \text{\textifsymbol[ifgeo]{96}} &  &  & \text{\textifsymbol[ifgeo]{96}} & \text{\textifsymbol[ifgeo]{96}} &  &  & \text{\textifsymbol[ifgeo]{96}} & \text{\textifsymbol[ifgeo]{96}} &  &  & \text{\textifsymbol[ifgeo]{96}} & \text{\textifsymbol[ifgeo]{96}}\\
Z=0,D=1,W=1 &  &  &  &  &  &  &  &  & \text{\textifsymbol[ifgeo]{96}} & \text{\textifsymbol[ifgeo]{96}} & \text{\textifsymbol[ifgeo]{96}} & \text{\textifsymbol[ifgeo]{96}} & \text{\textifsymbol[ifgeo]{96}} & \text{\textifsymbol[ifgeo]{96}} & \text{\textifsymbol[ifgeo]{96}} & \text{\textifsymbol[ifgeo]{96}}\\
Z=1,D=1,W=1 &  &  &  &  &  &  &  &  & \text{\textifsymbol[ifgeo]{96}} & \text{\textifsymbol[ifgeo]{96}} & \text{\textifsymbol[ifgeo]{96}} & \text{\textifsymbol[ifgeo]{96}} & \text{\textifsymbol[ifgeo]{96}} & \text{\textifsymbol[ifgeo]{96}} & \text{\textifsymbol[ifgeo]{96}} & \text{\textifsymbol[ifgeo]{96}}
\end{array}
\]
The square represents a vector of coefficients corresponding to $\theta$'s
used in approximating the mapping types, and the blank represents
a zero vector. By construction, each entry in matrix $B$ is equivalent
to $\int_{\mathcal{U}_{z,x}^{d}}b_{k,K}(u)du$ such that the product
of $B$ and $\theta$ is equal to the data distribution. From the
matrix form above, we can guarantee that $B$ has full row rank if,
for given $(d,w)$, the row representing $Z=0$ cannot be a constant
multiplication of the row representing $Z=1$. 

\begin{lemma}\label{lem:full_rank_B}Suppose $Z\in\{z_{1},z_{2}\}$
is a binary IV. Assume that $P(z_{1}),P(z_{2})\in(0,1)$ and $P(z_{1})\ne P(z_{2})$.
For $k=0,1,...,K$, define $f(k)=\frac{\int_{0}^{P(z_{1})}b_{k,K}(u)du}{\int_{0}^{P(z_{2})}b_{k,K}(u)du}$.
Then, $f(k)$ is not a constant function, and thus $B$ has full row
rank.\end{lemma}

The proof of this lemma appears below. Our goal is to rescale the
coefficient matrix $B$ into a new matrix $\tilde{B}$ such that its
entries have balanced orders of magnitude. The most intuitive choice
of $\tilde{B}$ is the fully reduced form of $B$. Note $B$ usually
has more rows than columns (otherwise, we achieve point identification),
therefore, the fully reduced form of $B$ would take the form of 
\[
\tilde{B}=\left[I,\boldsymbol{0}\right],
\]
where $I$ is the identity matrix of rank $m$, and $\boldsymbol{0}$
is a zero matrix with dimension $m\times(n-m)$. The next step is
to find a transformation matrix $X$ such that $BX=\tilde{B}$, so
that we can rewrite the optimization problem as
\[
\max_{\tilde{\theta}\in\tilde{\Theta}_{K}}\tilde{A}\ensuremath{\tilde{\theta}}
\]
subject to
\[
\tilde{B}\tilde{\theta}=p,
\]
where $\tilde{\theta}=X^{-1}\theta$, $\tilde{A}=AX$, and $\tilde{\Theta}_{K}\equiv\left\{ \tilde{\theta}:MX\tilde{\theta}=\boldsymbol{1},X\tilde{\theta}\geq\boldsymbol{0}\right\} $.

We propose a simple algorithm to find a full rank $X$. Since $B$
is the column-reduced form of $B$, $X$ can be viewed as the matrix
of elementary operation used to reach the reduced form. To construct
$X$, we first use the transposed matrix $B'$ and apply to it Gauss-Jordan
elimination with partial pivoting to achieve a row-reduced form. We
apply the exactly same procedure to an identity matrix $I$ with dimension
$n\times n$. Then, the transpose of the resulting matrix becomes
$X$. Because simple row operations preserve the rank, $X$ is guaranteed
to have full rank. Note that there may exist multiple solutions of
$X$, which essentially makes this procedure computationally easier
than solving an LP.

\subsection{Pointwise and Uniform Sharp Bounds on MTE\label{subsec:Point-wise-and-Uniform}}

In Section \ref{sec:Preliminaries:-Observables,-Assu}, we provided
some examples of target parameters. The building block for these parameters
is the MTE, $m_{1}(u)-m_{0}(u)$ (suppressing $w$ and $x$). \citet{heckman2005structural}
show why this fundamental parameter can be of independent interest.
Unlike other target parameters proposed here, we may want to recover
the MTE as a function of $u$ (besides evaluating it at fixed $u$).
In this section, we discuss the subtle issue of pointwise and uniform
sharp bounds on $\tau_{MTE}(u)\equiv m_{1}(u)-m_{0}(u)$ as a function
of $u$.

For simplicity, suppress $W$ and $X$ and redefine $\epsilon\equiv(Y(0),Y(1))$
and $e\equiv(y(0),y(1))$. Recall $q(u)\equiv\{q(e|u)\}_{e\in\mathcal{\mathcal{E}}}$
and $\mathcal{Q}\equiv\{q(\cdot):\sum_{e}q(e|u)=1\,\forall u\text{ and }q(e|u)\ge0\,\forall(e,u)\}$.
Let $\mathcal{M}$ be the set of MTE functions, i.e.,
\[
\mathcal{M}\equiv\Big\{ m_{1}(\cdot)-m_{0}(\cdot):m_{d}(\cdot)=\sum_{e\in\mathcal{E}:y(d)=1}q(e|\cdot)\text{ }\forall d\in\{0,1\}\text{ for }q(\cdot)\in\mathcal{Q}\Big\}.
\]
The bounds on $\tau_{MTE}\in\mathcal{M}$ in the $\infty$-LP are
given by using a Dirac delta function as a weight. Therefore, given
evaluation point $u\in[0,1]$, \eqref{eq:upper}--\eqref{eq:constr}
can be simplified as follows, defining the upper and lower bounds
$\overline{\tau}(u)$ and $\underline{\tau}(u)$ (that are explicit
about the evaluation point) on $\tau_{MTE}(u)$:
\begin{align}
\overline{\tau}(u) & =\sup_{q\in\mathcal{Q}}\sum_{e\in\mathcal{E}:y(1)=1}q(e|u)-\sum_{e\in\mathcal{E}:y(0)=1}q(e|u)\label{eq:upper4}\\
\underline{\tau}(u) & =\inf_{q\in\mathcal{Q}}\sum_{e\in\mathcal{E}:y(1)=1}q(e|u)-\sum_{e\in\mathcal{E}:y(0)=1}q(e|u)\label{eq:lower4}
\end{align}
subject to
\begin{align}
\sum_{e:y(d)=1}\int_{\mathcal{U}_{z}^{d}}q(e|\tilde{u})d\tilde{u} & =p(1,d|z)\qquad\forall(d,z)\in\{0,1\}\times\mathcal{Z}.\label{eq:constr4}
\end{align}
Then, for any fixed $u\in[0,1]$,
\begin{align*}
\underline{\tau}(u) & \le\tau_{MTE}(u)\le\overline{\tau}(u).
\end{align*}
We argue that these bounds are pointwise sharp but not necessarily
uniformly sharp for $\tau_{MTE}(\cdot)$.\footnote{See \citet{firpo2019partial} for related definitions of pointwise
and uniform sharpness.}

\begin{definition}[Pointwise Sharpness]$\overline{\tau}(\cdot)$
and $\underline{\tau}(\cdot)$ are pointwise sharp if, for any $\bar{u}\in[0,1]$,
there exist $\overline{\tau}_{MTE,\bar{u}},\underline{\tau}_{MTE,\bar{u}}\in\mathcal{M}$
such that $\overline{\tau}(\bar{u})=\overline{\tau}_{MTE,\bar{u}}(\bar{u})$
and $\underline{\tau}(\bar{u})=\underline{\tau}_{MTE,\bar{u}}(\bar{u})$.\end{definition}

\begin{theorem}\label{thm:ptws_sharp}$\overline{\tau}(\cdot)$ and
$\underline{\tau}(\cdot)$ are pointwise sharp bounds on $\tau_{MTE}(\cdot)$.\end{theorem}

The proofs of this and other theorems appear later. Note that pointwise
bounds will maintain some properties of an MTE function, but not all.
For uniform sharpness, $\overline{\tau}(\cdot)$ and $\underline{\tau}(\cdot)$
themselves have to be MTE functions on $[0,1]$, i.e., $\overline{\tau}(\cdot)$
and $\underline{\tau}(\cdot)$ should be elements in $\mathcal{M}$.

\begin{definition}[Uniform Sharpness]$\overline{\tau}(\cdot)$ and
$\underline{\tau}(\cdot)$ are uniformly sharp if $\overline{\tau}(\cdot),\underline{\tau}(\cdot)\in\mathcal{M}$.\end{definition}

The following theorem is almost immediate.

\begin{theorem}\label{thm:unif_sharp}$\overline{\tau}(\cdot)$ is
uniformly sharp if and only if there exists $q^{*}(\cdot)\in\mathcal{Q}$
such that $q^{*}(\cdot)$ is in the feasible set and $\overline{\tau}(u)=\sum_{e\in\mathcal{E}:y(1)=1}q^{*}(e|u)-\sum_{e\in\mathcal{E}:y(0)=1}q^{*}(e|u)$
for all $u\in[0,1]$. Similarly, $\underline{\tau}(\cdot)$ is uniformly
sharp if and only if there exists $q^{\dagger}(\cdot)\in\mathcal{Q}$
such that $q^{\dagger}(\cdot)$ is in the feasible set and $\underline{\tau}(u)=\sum_{e\in\mathcal{E}:y(1)=1}q^{\dagger}(e|u)-\sum_{e\in\mathcal{E}:y(0)=1}q^{\dagger}(e|u)$
for all $u\in[0,1]$.\end{theorem}

The following is a more useful result that relates pointwise bounds
with uniform bounds. For each $\bar{u}$, let $q_{\bar{u}}^{*}(\cdot)$
and $q_{\bar{u}}^{\dagger}(\cdot)$ be the pointwise maximizer and
minimizer of \eqref{eq:upper4}--\eqref{eq:constr4}, respectively.

\begin{corollary}\label{cor:unif_sharp}$\overline{\tau}(\cdot)$
is uniformly sharp if and only if there exists $q^{*}(\cdot)\in\mathcal{Q}$
such that $q^{*}(\cdot)$ is in the feasible set and $q_{\bar{u}}^{*}(\bar{u})=q^{*}(\bar{u})$
for all $\bar{u}\in[0,1]$. Also, $\underline{\tau}(u)$ is uniformly
sharp if and only if there exists $q^{\dagger}(\cdot)\in\mathcal{Q}$
such that $q^{\dagger}(\cdot)$ is in the feasible set and $q_{\bar{u}}^{\dagger}(\bar{u})=q^{\dagger}(\bar{u})$
for all $\bar{u}\in[0,1]$.\end{corollary}

Based on the Bernstein approximation we introduce, this corollary
implies that for a uniform upper bound to exist, there should exist
a common maximizer $\theta^{*}$ such that $\theta^{*}$ is in the
feasible set of the LP and $\overline{\tau}(u)=\sum_{k\in\mathcal{K}}\Big\{\sum_{e\in\mathcal{E}:y(1)=1}\theta_{k}^{e*}b_{k}(u)-\sum_{e\in\mathcal{E}:y(0)=1}\theta_{k}^{e*}b_{k}(u)\Big\}$
\textit{for all} $u$. In other words, if $\theta_{\bar{u}}^{*}$
is the maximizer of the LP for given $\bar{u}$, then there should
exist $\theta^{*}$ in the feasible set such that $\theta_{\bar{u}}^{*}=\theta^{*}$
for all $\bar{u}\in[0,1]$. Since this condition will not generally
hold, uniformly sharp bounds on the MTE may not exist. The condition
can be verified in practice by implementing the LP in a finite grid
of $u$ in $[0,1]$ and checking whether $\theta_{u}^{*}$ is constant
for all values in the grid.

The discussion of this section extends to the case with continuous
$Y$ by analogously defining the set of MTE functions generated by
$\tilde{\mathcal{Q}}$. One implication of uniform sharpness is that
$\overline{\tau}(\cdot)$ and $\underline{\tau}(\cdot)$ should be
consistent with the CDF properties of $\tilde{q}(e|u)$ embedded in
$\tilde{\mathcal{Q}}$.

\subsection{Linear Programming with Continuous $X$\label{subsec:Linear-Programming-with}}

Suppose $X$ is a vector of continuously distributed covariates and
assume $\mathcal{X}=[0,1]^{d_{X}}$ without loss of generality. Let
$q(u,x)\equiv\{q(e|u,x)\}_{e\in\mathcal{\mathcal{E}}}$ and $p(x)\equiv\{p(1,d|z,w,x)\}_{d,z,w}$.
Recall that $R_{\omega}:\mathcal{Q}\rightarrow\mathbb{R}$ and $R:\mathcal{Q}\rightarrow\mathbb{R}^{d_{p}}$
are the linear operators of $q(\cdot)$ where $d_{p}$ is the dimension
of $p$. Consider the following LP:
\begin{align}
\overline{\tau}= & \sup_{q\in\mathcal{Q}}R_{\omega}q,\label{eq:LP1_x}\\
\underline{\tau}= & \inf_{q\in\mathcal{Q}}R_{\omega}q,\label{eq:LP2_x}\\
 & s.t.\quad(Rq)(x)=p(x)\qquad\text{for all }x\in\mathcal{X},\label{eq:LP3_x}
\end{align}
where $(Rq)(x)=p(x)$ emphasizes the dependence on $x$, and thus
represents infinitely many constraints. Therefore, this LP is infinite
dimensional because of both the decision variable and the constraints.

Now, for the sieve space of $\mathcal{Q}$, we consider
\begin{align}
\bar{\mathcal{Q}}_{K} & \equiv\left\{ \Big\{\sum_{k=1}^{K}\theta_{k}^{e}b_{k}(u,x)\Big\}_{e\in\mathcal{E}}:\sum_{e\in\mathcal{E}}\theta_{k}^{e}=1\text{ and }\theta_{k}^{e}\ge0\text{ }\forall(e,k)\right\} \subseteq\mathcal{Q},\label{eq:sieve-conti}
\end{align}
where $b_{k}(u,x)$ is a bivariate Bernstein polynomial and $\mathcal{K}\equiv\{1,...,K\}$.
Then, for $q\in\bar{\mathcal{Q}}_{K}$
\begin{align}
E[\tau_{d}(Z,W,X)] & =E\left[\sum_{e:y(d,W)=1}\sum_{k\in\mathcal{K}}\theta_{k}^{e}\int b_{k}(u,X)\omega_{d}(u,Z,X)du\right]\nonumber \\
 & \equiv\sum_{w\in\mathcal{W}}\sum_{e:y(d,w)=1}\sum_{k\in\mathcal{K}}\theta_{k}^{e}\tilde{\gamma}_{k}^{d}(w),\label{eq:derive_target-1}
\end{align}
where $\tilde{\gamma}_{k}^{d}(w)\equiv E\left[\sum_{z\in\{0,1\}}p(z,w|X)\int b_{k}(u,X)\omega_{d}(u,z,X)du\right]$
with $p(z,w|x)\equiv\Pr[Z=z,W=w|X=x]$. Also,
\begin{align}
p(1,d|z,w,x) & =\sum_{e:y(d,w)=1}\sum_{k\in\mathcal{K}}\theta_{k}^{e}\int_{\mathcal{U}_{z,x}^{d}}b_{k}(u,x)du\nonumber \\
 & \equiv\sum_{e:y(d,w)=1}\sum_{k\in\mathcal{K}}\theta_{k}^{e}\tilde{\delta}_{k}^{d}(z,x),\label{eq:derive_constr-1}
\end{align}
where $\tilde{\delta}_{k}^{d}(z,x)\equiv\int_{\mathcal{U}_{z,x}^{d}}b_{k}(u,x)du$.
To deal with this infinite dimensional constraints (with respect to
$x$), we proceed as follows. For any measurable function $h:\mathcal{X}\rightarrow\mathbb{R}$,
$E\left|h(X)\right|=0$ if and only if $h(x)=0$ almost everywhere
in $\mathcal{X}$. Therefore, the equality restriction \eqref{eq:derive_constr-1}
can be replaced by
\begin{align*}
E\left|\sum_{e:y(d,w)=1}\sum_{k\in\mathcal{K}}\theta_{k}^{e}\tilde{\delta}_{k}^{d}(z,X)-p(1,d|z,w,X)\right| & =0
\end{align*}
for all $(d,z,w)\in\{0,1\}\times\mathcal{Z}\times\mathcal{W}$. Let
$\tilde{\theta}\equiv\{\theta_{k}^{e}\}_{(e,k)\in\mathcal{E}\times\mathcal{K}}$
and let
\begin{align*}
\tilde{\Theta}_{K} & \equiv\left\{ \tilde{\theta}:\sum_{e\in\mathcal{E}}\theta_{k}^{e}=1\text{ and }\theta_{k}^{e}\ge0\text{ }\forall(e,k)\in\mathcal{E}\times\mathcal{K}\right\} .
\end{align*}
Then, we can formulate the following finite-dimensional LP:
\begin{align}
\overline{\tau}_{K} & =\max_{\theta\in\Theta_{K}}\sum_{(k,w)\in\mathcal{K}\times\mathcal{W}}\Big\{\sum_{e:y(1,w)=1}\theta_{k}^{e}\tilde{\gamma}_{k}^{1}(w)-\sum_{e:y(0,w)=1}\theta_{k}^{e}\tilde{\gamma}_{k}^{0}(w)\Big\}\label{eq:LP1_x2}\\
\underline{\tau}_{K} & =\min_{\theta\in\Theta_{K}}\sum_{(k,w)\in\mathcal{K}\times\mathcal{W}}\Big\{\sum_{e:y(1,w)=1}\theta_{k}^{e}\tilde{\gamma}_{k}^{1}(w)-\sum_{e:y(0,w)=1}\theta_{k}^{e}\tilde{\gamma}_{k}^{0}(w)\Big\}\label{eq:LP2_x2}
\end{align}
subject to
\begin{align}
E\left|\sum_{e:y(d,w)=1}\sum_{k\in\mathcal{K}}\theta_{k}^{e}\tilde{\delta}_{k}^{d}(z,X)-p(1,d|z,w,X)\right| & =0\qquad\forall(d,z,w)\in\{0,1\}\times\mathcal{Z}\times\mathcal{W}.\label{eq:LP3_x2}
\end{align}

Later, we want to introduce additional constraints from some identifying
assumptions:
\begin{align}
R_{1}q & =a_{1},\label{eq:LP4-conti}\\
R_{2}q & \le a_{2}.\label{eq:LP5-conti}
\end{align}
For the equality restrictions, we can use the same approach that transforms
\eqref{eq:LP3_x}. For the inequality restrictions \eqref{eq:LP5-conti},
we can allow any identifying assumptions for which $R_{2}$ is a matrix
rather than an operator:

\begin{asMAT}$R_{2}$ is a $\dim(a_{2})\times\dim(q)$ matrix.\end{asMAT}

Assumptions M and C and the unconditional version of Assumption MTS
satisfy this condition. 

\subsection{Estimation and Inference\label{subsec:Inference}}

Although the paper's main focus is identification, we briefly discuss
estimation and inference. Assume a random sample of $\{Y_{i},D_{i},Z_{i},W_{i},X_{i}\}_{i=1}^{N}$.
The estimation of the bounds characterized by the LP \eqref{eq:upper2}--\eqref{eq:constr2}
is straightforward by replacing the population objects $(\gamma_{k}^{d},\delta_{k}^{d},p)$
with their sample counterparts $(\hat{\gamma}_{k}^{d},\hat{\delta}_{k}^{d},\hat{p})$.
To account for the statistical error arising from estimation based
on the finite sample, we replace \ref{eq:constr} with a tuned constraint:
\[
||\hat{R}_{0}q-\hat{p}||\leq\inf_{q\in\mathcal{Q}}||\hat{R}_{0}q-\hat{p}||+\kappa.
\]
Here $\kappa$ is a tuning parameter selected by the researchers.
In the finite-dimensional LP, to be more specific, we apply the Euclidean
norm and replace \eqref{eq:constr2} by
\begin{align*}
 & \sqrt{\sum_{z,w,x}\left(\sum_{e:y(d,w)=1}\sum_{k\in\mathcal{K}}\theta_{k}^{e,x}\hat{\delta}_{k}^{d}(z,x)-\hat{p}(1,d|z,w,x)\right)^{2}}\leq\\
 & \;\;\;\;\;\;\;\;\;\;\;\;\;\;\;\;\;\;\;\;\;\;\;\;\;\;\;\;\;\;\;\;\;\;\;\;\;\;\;\inf_{\theta\in\Theta_{k}}\sqrt{\sum_{z,w,x}\left(\sum_{e:y(d,w)=1}\sum_{k\in\mathcal{K}}\theta_{k}^{e,x}\hat{\delta}_{k}^{d}(z,x)-\hat{p}(1,d|z,w,x)\right)^{2}}+\kappa.
\end{align*}

With continuous $Y$ in \eqref{eq:LP1_y2}--\eqref{eq:LP3_y2}, we
replace \eqref{eq:LP3_y2} with (a slack version of) its sample counterparts:
\begin{align*}
\frac{1}{N}\sum_{i=1}^{N}\left|\sum_{\boldsymbol{k}\in\mathcal{K}^{5}}\theta_{\boldsymbol{k}}^{x}\hat{\sigma}_{\boldsymbol{k}}^{d}(Y_{i},Z_{i},X_{i})-\hat{\pi}(Y_{i},d|Z_{i},W_{i},X_{i})\right| & \le\eta,
\end{align*}
where $\hat{\pi}(y,d|z,w,x)$ is some preliminary estimate of $\pi(y,d|z,w,x)$
and $\eta$ is another slackness parameter. Since $\eta$ takes into
account the error from the finite sample, there is no need to introduce
an additional tuning parameter to account for the estimation error.
A similar idea applies to the case with continuous $X$ in \eqref{eq:LP1_x2}--\eqref{eq:LP3_x2}.

It is important to construct a confidence set for our target parameter
or its bounds in order to account for the sampling variation in measuring
treatment effectiveness. It will also be interesting to develop a
procedure to conduct a specification test for the identifying assumptions
discussed in Section \ref{sec:Additional-Assumptions}. The problem
of statistical inference when the identified set is constructed via
linear programming has been studied in, e.g., \citet{deb2023revealed},
\citet{mogstad2018using}, \citet{hsieh2022inference}, \citet{torgovitsky2019nonparametric},
and \citet{fang2023inference}. Among these papers, \citet{mogstad2018using}'s
setting is closest to our setting with discrete variables, and their
inference procedure can be directly adapted to our problem. Instead
of repeating their result here, we only briefly discuss the procedure. 

Recall $q(u)\equiv\{q(e|u,x)\}_{e\in\mathcal{\mathcal{E}},x\in\mathcal{X}}$
is the latent distribution and $p\equiv\{p(1,d|z,x)\}_{d,z,x}$ is
the distribution of the data, and $R_{\omega}$, $R_{0}$, $R_{1}$,
and $R_{2}$ denote the linear operators of $q(\cdot)$ that correspond
to the target and constraints. Consider the following hypotheses:
\begin{align*}
H_{0}: & p\in\mathcal{P}_{0},\qquad H_{1}:p\in\mathcal{P}\backslash\mathcal{P}_{0},
\end{align*}
where
\begin{align*}
\mathcal{P}_{0} & \equiv\{p\in\mathcal{P}:Rq=a\text{ for some }q\in\mathcal{Q}\}
\end{align*}
and
\begin{align*}
R & \equiv(R_{\omega}',R_{0}',R_{1}',R_{2}')',\\
a & \equiv(\tau,p',a_{1}',a_{2}')'.
\end{align*}
Suppose $\hat{R}$ and $\hat{a}$ are sample counterparts of $R$
and $a$. Then, a minimum distance test statistic can be constructed
as
\begin{align*}
T_{N}(\tau) & \equiv\inf_{q\in\mathcal{Q}_{K}}\sqrt{N}\left\Vert \hat{R}q-\hat{a}\right\Vert .
\end{align*}
Similar to \citet{mogstad2017using}, $T_{N}(\tau)$ is the solution
to a convex optimization problem that can be reformulated as an LP
using duality. A $(1-\alpha)$-confidence set for the target parameter
$\tau$ can be constructed by inverting the test:
\begin{align*}
CS_{1-\alpha} & \equiv\{\tau:T_{N}(\tau)\le\hat{c}_{1-\alpha}\}
\end{align*}
where $\hat{c}_{1-\alpha}$ is the critical value for the test. The
resulting object is of independent interest, and it can further be
used to conduct specification tests. The large sample theory for $T_{N}(\tau)$,
as well as a bootstrap procedure to calculate $\hat{c}_{1-\alpha}$,
will directly follow according to \citet{mogstad2017using}, which
is omitted for succinctness.

When $Y$ or $X$ is continuously distributed, then the resulting
LP is semi-infinite dimensional. In this case, the inference procedure
by \citet{chernozhukov2013intersection} may be applied. In this case,
the estimation of the bounds can be conducted within the framework.

\section{Proofs\label{sec:Proofs}}

\subsection{Proof of Lemma \ref{lem:redun}}

Fix $(d,z,w,x)$. By $\sum_{e\in\mathcal{E}}q(e|u,x)=1$ for $q\in\mathcal{Q}$,
we have
\begin{align*}
1 & =\sum_{e\in\mathcal{E}}q(e|u,x)=\sum_{e:y(d,w)=1}q(e|u,x)+\sum_{e:y(d,w)=0}q(e|u,x).
\end{align*}
Then, in \eqref{eq:constr}, the constraint with $p(0,d|z,w,x)$ can
be written as 
\begin{align*}
p(0,d|z,w,x) & =\int_{\mathcal{U}_{z,x}^{d}}\sum_{e:y(d,w)=0}q(e|u,x)du=\int_{\mathcal{U}_{z,x}^{d}}\Big\{1-\sum_{e:y(d,w)=1}q(e|u,x)\Big\} du\\
 & =\Pr[D=d|Z=z,W=w,X=x]-\int_{\mathcal{U}_{z,x}^{d}}\sum_{e:y(d,w)=1}q(e|u,x)du.
\end{align*}
Then by rearranging terms, this constraint becomes
\begin{align*}
p(1,d|z,w,x) & =\int_{\mathcal{U}_{z,x}^{d}}\sum_{e:y(d,w)=1}q(e|u,x)du,
\end{align*}
since $\Pr[D=d|Z=z,W=w,X=x]-p(0,d|z,w,x)=p(1,d|z,w,x)$. Therefore,
the constraint with $p(0,d|z,w,x)$ does not contribute to the restrictions
imposed by \eqref{eq:constr} and $q\in\mathcal{Q}$. $\square$

\subsection{Proof of Theorem \ref{thm:equivalence_with_mogstad}}

As $Y\in\{0,1\}$, $\mathcal{M}_{f}$ can be rewritten as:
\[
\mathcal{M}_{f}=\Big\{ m=(m_{0},m_{1}):m_{d}(u,x)=\sum_{e:y(d)=y}q(e|u,x),d=\{0,1\},q\in\mathcal{Q}_{f}\Big\}.
\]
From (\ref{eq:constr}), we can write $E\text{\ensuremath{\left[Y|D=0,z,x\right]}}$
in terms of $q(e|u,x)$ as below:
\begin{align}
E\text{\ensuremath{\left[Y|D=0,z,x\right]}=} & \Pr\left[Y=1|D=0,z,x\right]=\frac{\Pr\left[Y=1,D=0|z,x\right]}{\Pr\left[D=0|z,x\right]}\nonumber \\
= & \frac{1}{1-P(z,x)}\sum_{e:y(0)=1}\int_{P(z,x)}^{1}q(e|u,x)du\nonumber \\
= & \frac{1}{1-P(z,x)}\int_{P(z,x)}^{1}\sum_{e:y(0)=1}q(e|u,x)du\label{eq:proof1}
\end{align}
for $\forall(z,x)\in\mathcal{Z\times\mathcal{X}}$. Then, for $(m_{0},m_{1})\in\mathcal{M}_{f}$
\[
E[Y|D=0,Z,X]=\frac{1}{1-P(Z,X)}\int_{P(Z,X)}^{1}m_{0}(u,X)du
\]
almost surely. Symmetrically, 
\[
E[Y|D=1,Z,X]=\frac{1}{P(Z,X)}\int_{0}^{P(Z,X)}m_{1}(u,X)du
\]
almost surely. Therefore, $\mathcal{M}_{f}\subseteq\mathcal{\mathcal{M}}_{id}$. 

Now suppose $m\in\mathcal{\mathcal{M}}_{id}$. By (\ref{eq: infoset0})
and (\ref{eq:proof1}), for $\forall z,x$ in a set of positive measure
(and symmetrically for the other equation), 
\begin{align}
\frac{1}{1-P(z,x)}\int_{P(z,x)}^{1}m_{0}(u,x)du & =\frac{1}{1-P(z,x)}\sum_{e:y(0)=1}\int_{P(z,x)}^{1}q(e|u,x)du\label{eq:proof2_0}\\
\frac{1}{P(z,x)}\int_{P(z,x)}^{1}m_{1}(u,x)du & =\frac{1}{P(z,x)}\sum_{e:y(1)=1}\int_{0}^{P(z,x)}q(e|u,x)du\label{eq:proof2}
\end{align}
Therefore, we need to find some $q(e|u,x)\in\mathcal{Q}_{f}$ such
that (\ref{eq:proof2_0}) and (\ref{eq:proof2}) are satisfied. Recall
$e\equiv(y(0),y(1))$. We construct $q(e|u,x)$ as
\begin{equation}
q(e|u,x)=\begin{cases}
1-max\{m_{0}(u,x),m_{1}(u,x)\}, & e=(0,0)\\
max\{m_{1}(u,x)-m_{0}(u,x),0\}, & e=(0,1)\\
max\{m_{0}(u,x)-m_{1}(u,x),0\}, & e=(1,0)\\
min\{m_{0}(u,x),m_{1}(u,x)\}, & e=(1,1)
\end{cases}\label{eq:construction}
\end{equation}
By the property of binary $Y$, it can be proved that the constructed
$q(e|u,x)$ in (\ref{eq:construction}) is a function in $Q_{f}$
and also satisfies (\ref{eq:proof2_0}) and (\ref{eq:proof2}) following
the fact that $\forall(u,x)\in\mathcal{U}\times\mathcal{X}$
\begin{align*}
m_{0}(u,x) & =\sum_{e:y(0)=1}q(e|u,x),\\
m_{1}(u,x) & =\sum_{e:y(1)=1}q(e|u,x).
\end{align*}
Therefore, $\mathcal{M}_{id}\subseteq\mathcal{\mathcal{M}}_{f}$.
$\square$

\subsection{Proof of Theorem \ref{thm:better_than_mogstad}}

For this proof, we introduce a proposition adapting a result of \citet{hafsa2003interchange}.
Let $q:\mathcal{U}\rightarrow\mathcal{V}$ be a measurable function
with $\mathcal{U}$ compact in $\mathbb{R}^{L_{\mathcal{U}}}$ and
$\mathcal{V}\subset\mathbb{R}^{L_{\mathcal{V}}}$ and let $h:\mathcal{U}\times\mathcal{V}\rightarrow\mathbb{R}\cup\{-\infty\}\cup\{+\infty\}$
be a continuous function. Let $C_{c}(\mathcal{U})$ be the space of
all continuous functions from $\mathcal{U}$ whose support is compact.
Finally, let $\mathcal{D}(h)$ be the set of measurable function $q:\mathcal{U}\rightarrow\mathcal{V}$
such that $h(\cdot,q(\cdot))\in L^{1}(\mathcal{U},\mu)$.

\begin{proposition}\label{prop:interchange}Let $\mathcal{A}\subset L^{2}(\mathcal{U},\mu)$.
If $\mathcal{A}\subset C_{c}(\mathcal{U})$, $\mathcal{A}\cap\mathcal{D}(h)\neq\emptyset$,
and $h$ is convex in its second variable, then
\begin{align*}
\inf_{q\in\mathcal{A}}\int_{\mathcal{U}}h(u,q(u))d\mu(u) & =\int_{\mathcal{U}}\inf_{\xi\in\Gamma(u)}h(u,\xi)d\mu(u)
\end{align*}
with $\Gamma(u)\equiv\text{cl}\{q(u):q\in\mathcal{A\cap\mathcal{D}}(h)\}$
$\mu$-a.e.\end{proposition}

This proposition is immediate from Corollary 5.1 in \citet{hafsa2003interchange}
because $C_{c}(\mathcal{U})$ is normally decomposable and the $\mu$-essential
supremum $\Gamma$ of a subset $\mathcal{H}\subset L^{2}(\mathcal{U},\mu)$
satisfies $\Gamma(u)=\text{cl}\{q(u):q\in\mathcal{A\cap\mathcal{D}}(h)\}$
$\mu$-a.e.\footnote{For the definition of a normally decomposable set and $\mu$-essential
supremum, we refer the reader to \citet{hafsa2003interchange}.}

\bigskip{}

We prove Theorem \ref{thm:better_than_mogstad} by showing (i) $\mathcal{M}_{f}\subseteq\mathcal{M}_{id}$
and (ii) $\mathcal{M}_{id}\not\subseteq\mathcal{M}_{f}$. The proof
of (i) is straightforward. For any $(m_{0},m_{1})\in\mathcal{M}_{f}$,
\begin{align*}
m_{0}(u,x) & =\sum_{y\in\mathcal{Y}}y\sum_{e:y(0)=y}q(e|u,x),\\
m_{1}(u,x) & =\sum_{y\in\mathcal{Y}}y\sum_{e:y(1)=y}q(e|u,x),\forall(u,x)\in\mathcal{U}\times\mathcal{X}
\end{align*}
for some $q(e|u,x)\in\mathcal{Q}_{f}$. Then, for any $(z,x)\in\mathcal{Z}\times\mathcal{X}$,
\begin{align*}
\frac{1}{1-P(z,x)}\int_{P(z,x)}^{1}m_{0}(u,x)du & =\frac{1}{1-P(z,x)}\sum_{y\in\mathcal{Y}}y\sum_{e:y(0)=y}\int_{P(z,x)}^{1}q(e|u,x)du,\\
\frac{1}{P(z,x)}\int_{0}^{P(z,x)}m_{1}(u,x)du & =\frac{1}{P(z,x)}\sum_{y\in\mathcal{Y}}y\sum_{e:y(1)=y}\int_{0}^{P(z,x)}q(e|u,x)du.
\end{align*}
By definition of $\mathcal{Q}_{f}$, it holds that
\begin{align*}
\sum_{e:y(0)=y}\int_{P(z,x)}^{1}q(e|u,x)du & =p(y,0|z,x)=\Pr[Y=y|D=0,z,x]\left(1-P(z,x)\right),\\
\sum_{e:y(1)=y}\int_{P(z,x)}^{1}q(e|u,x)du & =p(y,1|z,x)=\Pr[Y=y|D=1,z,x]P(z,x),
\end{align*}
which implies
\begin{align*}
\frac{1}{1-P(Z,X)}\int_{P(Z,X)}^{1}m_{0}(u,X)du & =E\left[Y|D=0,Z,X\right],\\
\frac{1}{P(Z,X)}\int_{0}^{P(Z,X)}m_{1}(u,X)du & =E\left[Y|D=1,Z,X\right],
\end{align*}
almost surely. Therefore, we conclude that $\mathcal{M}_{f}\subseteq\mathcal{M}_{id}$.

Next, we prove (ii) by showing that there exists some $m_{d}\in\mathcal{M}_{id}$,
$m_{d}\notin\mathcal{M}_{f}$. Without loss of generality, we pick
$d=0$ and construct a $m_{0}(u,x)$ function such that $m_{0}\in\mathcal{M}_{id}$
and $m_{0}\notin\mathcal{M}_{f}$, or equivalently, $\forall q\in\mathcal{Q}_{f}$,
$\exists x\in\mathcal{X}$ such that 
\[
m_{0}(u,x)\ne\sum_{y\in\mathcal{Y}}y\sum_{e:y(0)=y}q(e|u,x)
\]
for $u$ lying in some set $A\subseteq[0,1]$ with non-zero measure.

Given $\mathcal{Z}$ is discrete and $\mathcal{X}$ is compact, let
$(z^{*},x^{*})\in\mathcal{Z}\times\mathcal{X}$ be such that $P(z^{*},x^{*})\ge P(z,x)$
$\forall(z,x)\in\mathcal{Z}\times\mathcal{X}$. For example, let $(z^{*},x^{*})=\arg\max P(z,x)$
by Assumption EC(i). Define $\mathcal{\hat{\mathcal{Q}}}$ as
\[
\hat{\mathcal{Q}}=\left\{ \text{argsup}_{q\in\mathcal{Q}_{f}}\sum_{y\in\mathcal{Y}}y\sum_{e:y(0)=y}\int_{P(z^{*},x^{*})}^{\frac{1+P(z^{*},x^{*})}{2}}q(e|u,x^{*})du\right\} .
\]
By Assumption EC(ii) and $q$ being a function on a compact set,
$\mathcal{Q}_{f}$ is uniformly equicontinuous. Also, $\{q(u):q\in\mathcal{Q}_{f}\}$
is relatively compact in $[0,1]$ for all $u$.\footnote{A subset of a topological space is relatively compact if its closure
is compact.} Therefore, $\mathcal{Q}_{f}$ is relatively compact (\citet{simon1986compact}).
Moreover, $\mathcal{Q}_{f}$ is closed by Assumption EC(ii), Arzel\`a-Ascoli
theorem, and Lebesgue's dominated convergence theorem.\footnote{This is because of the following: $\mathcal{Q}_{f}\subset L^{2}([0,1])$
is equicontinuous and uniformly bounded and thus, by Arzel\`a-Ascoli
theorem, for any $\{q_{n}\}$ such that $\left\Vert q_{n}-q\right\Vert _{2}\rightarrow0$
there exists a subsequence $\{q_{n_{k}}\}$ such that $\left\Vert q_{n_{k}}-q\right\Vert _{\infty}\rightarrow0$.
Then, $q_{n_{k}}(u)\rightarrow q(u)$ $\forall u$. Note that $0\le q(e|\cdot,x)\le1$
and $\sum_{e}q(e|\cdot,x)=1$ trivially hold as $q_{n_{k}}\in\mathcal{Q}_{f}\subset\mathcal{Q}$.
Also $q$ is continuous as $q_{n_{k}}$ is continuous and uniformly
converges to $q$. Finally, since $q_{n_{k}}\le1$, by the dominated
convergence theorem, $R_{0}q=R_{0}\lim_{k}q_{n_{k}}=\lim_{k}R_{0}q_{n_{k}}=p$.
Therefore, $q\in\mathcal{Q}_{f}$.} Therefore $\mathcal{Q}_{f}$ is compact. Finally, the integral operator
shown in $\hat{\mathcal{Q}}$ is continuous in $q$. Therefore, there
exists $\hat{q}\in\mathcal{Q}_{f}$ that achieves the supremum, or
equivalently $\hat{q}\in\hat{\mathcal{Q}}$.

Next, since $\mathcal{Y}=\{y_{1},...,y_{L}\}$ with $L\ge3$, there
exist at least two non-zero values in $\mathcal{Y}$, which we denote
as $y^{*}$ and $y^{\dagger}$. We can then construct
\[
\psi(e,u)=\begin{cases}
\frac{sgn(y^{\dagger})\cdot y^{\dagger}}{(1-P(z^{*},x^{*}))\cdot L}, & e\in\left\{ e\in\mathcal{E:}y(0)=y^{*}\right\} ,u\in[P(z^{*},x^{*}),\frac{1+P(z^{*},x^{*})}{2}]\\
-\frac{sgn(y^{\dagger})\cdot y^{*}}{(1-P(z^{*},x^{*}))\cdot L}, & e\in\left\{ e\in\mathcal{E:}y(0)=y^{\dagger}\right\} ,u\in[\frac{1+P(z^{*},x^{*})}{2},1]\\
0, & \text{otherwise},
\end{cases}
\]
where $sgn(y)=1$ if $y>0$ and $-1$ if $y<0$. We construct $m_{0}(u,x)$
as:
\[
m_{0}(u,x)=\sum_{y\in\mathcal{Y}}y\sum_{e:y(0)=y}\hat{q}(e|u,x)+\psi(e,u).
\]
Then we can show $m_{0}(u,x)\in\mathcal{M}_{id}$ as follows: $\forall(z,x)\in\mathcal{Z}\times\mathcal{X}$,
\begin{align*}
\int_{P(z,x)}^{1}m_{0}(u,x)du & =\sum_{y\in\mathcal{Y}}y\sum_{e:y(0)=y}\int_{P(z,x)}^{1}\{\hat{q}(e|u,x)+\psi(e,u)\}du\\
 & =E[Y|D=0,z,x]\left(1-P(z,x)\right)+y^{*}\cdot\sum_{e:y(0)=y^{*}}sgn(y^{\dagger})\cdot y^{\dagger}/L\\
 & \;\;\;\;\;\;\;\;\;\;\;\;\;\;\;\;\;\;\;\;\;\;\;\;\;\;\;\;\;\;\;\;\;\;\;\;\;\;\;\;\;\;\;\;\;\;\;\;-y^{\dagger}\cdot\sum_{e:y(0)=y^{\dagger}}sgn(y^{\dagger})\cdot y^{*}/L\\
 & =E[Y|D=0,z,x]\left(1-P(z,x)\right)+sgn(y^{\dagger})\cdot y^{\dagger}\cdot y^{*}-sgn(y^{\dagger})\cdot y^{\dagger}\cdot y^{*}\\
 & =E[Y|D=0,z,x]\left(1-P(z,x)\right),
\end{align*}
where the second equality holds as $\hat{q}\in\mathcal{Q}_{f}(x^{*})$,
the third equality is by the definition of $\psi$ and $P(z,x)\le P(z^{*},x^{*})$
$\forall(z,x)$ and the forth equality follows by the fact that for
any $y\in\mathcal{Y}$, the cardinality of set $\left\{ e\in\mathcal{E}:g_{e}(0)=y\right\} $
is the same as $|\mathcal{Y}|=L$. Therefore, $m_{0}(u,x)$ satisfies
\eqref{eq: infoset0} and thus $m_{0}(u,x)\in\mathcal{M}_{id}$.

Next we show $m_{0}(u,x)\notin\mathcal{M}_{f}$. Note that for $x=x^{*}$
\begin{align*}
\int_{P(z^{*},x^{*})}^{\frac{1+P(z^{*},x^{*})}{2}}m_{0}(u,x^{*})du & =\sum_{y\in\mathcal{Y}}y\sum_{e:y(0)=y}\int_{P(z^{*},x^{*})}^{\frac{1+P(z^{*},x^{*})}{2}}\hat{q}(e|u,x^{*})du+\sum_{y\in\mathcal{Y}}y\sum_{e:y(0)=y}\int_{P(z^{*},x^{*})}^{\frac{1+P(z^{*},x^{*})}{2}}\psi(e,u)du\\
 & =\sup_{q\in\mathcal{Q}_{f}}\sum_{y\in\mathcal{Y}}y\sum_{e:y(0)=y}\int_{P(z^{*},x^{*})}^{\frac{1+P(z^{*},x^{*})}{2}}q(e|u,x^{*})du+\frac{1}{2}sgn(y^{\dagger})y^{\dagger}\\
 & >\sup_{q\in\mathcal{Q}_{f}}\sum_{y\in\mathcal{Y}}y\sum_{e:y(0)=y}\int_{P(z^{*},x^{*})}^{\frac{1+P(z^{*},x^{*})}{2}}q(e|u,x^{*})du.
\end{align*}
Then, we want to show that, for any $q\in\mathcal{Q}_{f}$
\begin{align}
\sup_{\bar{q}\in\mathcal{Q}_{f}}\sum_{y\in\mathcal{Y}}y\sum_{e:y(0)=y}\int_{P(z^{*},x^{*})}^{\frac{1+P(z^{*},x^{*})}{2}}\bar{q}(e|u,x^{*})du & \ge\int_{P(z^{*},x^{*})}^{\frac{1+P(z^{*},x^{*})}{2}}\sum_{y\in\mathcal{Y}}y\sum_{e:y(0)=y}q(e|u,x^{*})du.\label{eq:interchange}
\end{align}
If this inequality is true, then we can conclude that, for any $q\in\mathcal{Q}_{f}$
\[
\int_{P(z^{*},x^{*})}^{\frac{1+P(z^{*},x^{*})}{2}}\left[m_{0}(u,x^{*})-\sum_{y\in\mathcal{Y}}y\sum_{e:y(0)=y}q(e|u,x^{*})\right]du>0.
\]
In other words, there exists some $A\subseteq[P(z^{*},x^{*}),\frac{1+P(z^{*},x^{*})}{2}]$
with non-zero measure such that for $u\in A$ and for any $q\in\mathcal{Q}_{f}$
\[
m_{0}(u,x^{*})>\sum_{y\in\mathcal{Y}}y\sum_{e:y(0)=y}q(e|u,x^{*}),
\]
which implies that $m_{0}\notin\mathcal{M}_{f}$. Therefore, we can
conclude $\mathcal{M}_{id}\not\subseteq\mathcal{M}_{f}$.

It remains to show \eqref{eq:interchange}. It is convenient to define
$q(u;x)\equiv\{q(e|u,x)\}_{e\in\mathcal{E}}$ and
\begin{align*}
\mathcal{Q}(x) & \equiv\{q(\cdot;x):\sum_{e\in\mathcal{E}}q(e|u,x)=1\text{ and }q(e|u,x)\ge0\text{ }\forall(e,u)\},\\
\mathcal{Q}_{f}(x) & \equiv\Big\{ q(\cdot;x)\in L^{2}:q(\cdot;x)\in\mathcal{Q}(x)\text{ and satisfies equation }(\ref{eq:constr})\text{ for given }x\Big\},
\end{align*}
and simplify notation as $q^{*}(u)\equiv q(u;x^{*})\equiv\{q(e|u,x^{*})\}_{e\in\mathcal{E}}$
with $q_{e}^{*}(u)\equiv q(e|u,x^{*})$ and $\mathcal{Q}_{f}^{*}\equiv\mathcal{Q}_{f}(x^{*})$.
Then,
\begin{align*}
\sup_{q\in\mathcal{Q}_{f}}\sum_{y\in\mathcal{Y}}y\sum_{e:y(0)=y}\int_{P(z^{*},x^{*})}^{\frac{1+P(z^{*},x^{*})}{2}}q(e|u,x^{*})du & \ge\sup_{q^{*}\in\mathcal{Q}_{f}^{*}}\sum_{y\in\mathcal{Y}}y\sum_{e:y(0)=y}\int_{P(z^{*},x^{*})}^{\frac{1+P(z^{*},x^{*})}{2}}q(e|u,x^{*})du.
\end{align*}
For the right-hand side term in the inequality, we want to use Proposition
\ref{prop:interchange} to interchange the supremum and the integral.
In the proposition, let $q^{*}(\cdot):\mathcal{U}\rightarrow\mathcal{V}$
with $\mathcal{U}=[P(z^{*},x^{*}),(1+P(z^{*},x^{*}))/2]$ and $\mathcal{V}=[0,1]^{L^{4}}$
and $h(u,q^{*}(u))=-\sum_{y\in\mathcal{Y}}y\sum_{e:y(0)=y}q_{e}^{*}(u)$.
First, since $q^{*}$ is assumed to be continuous (Assumption EC),
$\mathcal{Q}_{f}^{*}\subset C_{c}([0,1])$. Next, note that $-\sum_{y\in\mathcal{Y}}y\sum_{e:y(0)=y}q_{e}^{*}(\cdot)\in L^{1}(\mathcal{U},F)$
where $F$ is the CDF of $U$ (i.e., an identity function) because
$q_{e}^{*}(\cdot)\in L^{1}(\mathcal{U},F)$. Therefore, $\mathcal{Q}_{f}^{*}\cap\mathcal{D}(h)=\mathcal{Q}_{f}^{*}$.
Finally, $h$ is linear and thus trivially convex. Therefore, by Proposition
\ref{prop:interchange}, $\Gamma(u)=\text{cl}\{q^{*}(u):q^{*}\in\mathcal{Q}_{f}^{*}\}$
and
\begin{align*}
\sup_{q^{*}\in\mathcal{Q}_{f}^{*}}\int_{P(z^{*},x^{*})}^{\frac{1+P(z^{*},x^{*})}{2}}\sum_{y\in\mathcal{Y}}y\sum_{e:y(0)=y}q_{e}^{*}(u)du & =\int_{P(z^{*},x^{*})}^{\frac{1+P(z^{*},x^{*})}{2}}\sum_{y\in\mathcal{Y}}y\sum_{e:y(0)=y}\sup\{q^{*}(u):q^{*}\in\mathcal{Q}_{f}^{*}\}du.
\end{align*}
But because, for any given $u$, $\sup\{q^{*}(u):q^{*}\in\mathcal{Q}_{f}^{*}\}\ge q^{*}(u)$
for any $q^{*}\in\mathcal{Q}_{f}^{*}$, we have
\begin{align*}
\int_{P(z^{*},x^{*})}^{\frac{1+P(z^{*},x^{*})}{2}}\sum_{y\in\mathcal{Y}}y\sum_{e:y(0)=y}\sup\{q^{*}(u):q^{*}\in\mathcal{Q}_{f}^{*}\}du & \ge\int_{P(z^{*},x^{*})}^{\frac{1+P(z^{*},x^{*})}{2}}\sum_{y\in\mathcal{Y}}y\sum_{e:y(0)=y}q^{*}(u)du
\end{align*}
for any $q^{*}\in\mathcal{Q}_{f}^{*}$. This proves the desired inequality
\eqref{eq:interchange}. $\square$

\subsection{Proof of Lemma \ref{lem:direction}}

Let $\Delta_{w,w'}\equiv Y(1,w)-Y(0,w')$. We first prove (i), namely
the identifiability of $\tau_{LATE}(w,w')$. Observe that
\begin{align*}
 & E[\Delta_{w,w'}|P(0)\le U\le P(1)]\\
 & =E[\Delta_{w,w}|P(0)\le U\le P(1)]+E[Y(0,w)-Y(0,w')|P(0)\le U\le P(1)]
\end{align*}
The first term can be identified as
\begin{align*}
E[\Delta_{w,w}|P(0)\le U\le P(1)] & =\frac{E[Y|Z=1,W=w]-E[Y|Z=0,W=w]}{P(1)-P(0)}
\end{align*}
under Assumptions EX and SEL(a) and by a simple extension of the proof
in \citet{imbens1994identification}. For the second term, it suffices
to identify $E[Y(0,w)|P(0)\le U\le P(1)]$ for any given $w$. Note
that
\begin{align*}
(1-P(z))E[Y|D=0,Z=z,W=w] & =(1-P(z))E[Y(0,w)|U>P(z)]\\
 & =\int_{P(z)}^{1}E[Y(0,w)|U=u]du
\end{align*}
by Assumptions EX and SEL(a). Therefore,
\begin{align*}
 & E[Y(0,w)|P(0)\le U\le P(1)]\\
 & =\frac{1}{P(1)-P(0)}\int_{P(0)}^{P(1)}E[Y(0,w)|U=u]du\\
 & =\frac{(1-P(0))E[Y|D=0,Z=0,W=w]-(1-P(1))E[Y|D=0,Z=1,W=w]}{P(1)-P(0)}.
\end{align*}
We prove (ii) for Assumption U. Let $\Delta_{w}\equiv\Delta_{w,w}\equiv Y(1,w)-Y(0,w)$.
Suppose $\tau_{LATE}(w,w)>0$. Then, because $P(1)-P(0)>0$,
\begin{align}
0 & <\int_{P(0)}^{P(1)}E[\Delta_{w}|U=u]du\nonumber \\
 & =\int_{P(0)}^{P(1)}\{E[\Delta_{w}|\Delta_{w}>0,U=u]\Pr[\Delta_{w}>0|U=u]\nonumber \\
 & \qquad\qquad+E[\Delta_{w}|\Delta_{w}\le0,U=u]\Pr[\Delta_{w}\le0|U=u]\}du.\label{eq:late}
\end{align}
Then it cannot be that $\Pr[\Delta_{w}\le0|U=u]=1$ $\forall u$ because
$E[\Delta_{w}|\Delta_{w}\le0,U=u]\le0$ $\forall u$, which contradicts
$\tau_{LATE}(w,w)>0$. A symmetric argument proves (ii) for Assumption
U$^{*}$ by using $(w,w')$ instead of $(w,w)$. To prove (iii), note
that \eqref{eq:late} yields 
\[
\int_{P(0)}^{P(1)}\left\{ \Pr[\Delta_{w}=1|U=u]-\Pr[\Delta_{w}=-1|U=u]\right\} du
\]
because $Y(d,w)\in\{0,1\}$. Then it cannot be that $\Pr[\Delta_{w}\ge0|U=u]\le\Pr[\Delta_{w}\le0|U=u]$
$\forall u$ because this inequality is equivalent to $\Pr[\Delta_{w}=1|U=u]\le\Pr[\Delta_{w}=-1|U=u]$
as $\Pr[\Delta_{w}=0|U=u]$ cancels out, which then contradicts $\tau_{LATE}(w,w)>0$.
The proof with Assumptions SEL(b) and EX(b) is analogous, and thus
omitted. $\square$

\subsection{Proof of Theorem \ref{lem:nonredundancy}}

Suppress $X$ for simplicity. In proving the claim of the theorem
for $W$, we fix $Z=z$. We first prove with Case (a). To simplify
notation, define $\tilde{e}$ be the decimal transform of $e\equiv(y(0,0),y(0,1),y(1,0),y(1,1))$
defined in the text, where its value corresponds to ``\#'' in Table
\ref{tab:16maps}. Define the r.v. $\tilde{\epsilon}$ accordingly.
Also, let $q(\tilde{e}_{1},...,\tilde{e}_{J}|u)\equiv\Pr[\tilde{\epsilon}\in\{\tilde{e}_{1},...,\tilde{e}_{J}\}|u]=\sum_{j=1}^{J}q(\tilde{e}_{j}|u)$.
Based on Table \eqref{tab:16maps}, we can easily derive
\begin{align*}
p(1,1|z,1) & =\int_{0}^{P(z)}\sum_{\tilde{e}:y(1,1)=1}q(\tilde{e}|u)du=\int_{0}^{P(z)}q(9,...,16|u)du,\\
p(1,1|z,0) & =\int_{0}^{P(z)}\sum_{\tilde{e}:y(1,0)=1}q(\tilde{e}|u)du=\int_{0}^{P(z)}q(5,...,8,13,...,16|u)du,\\
p(1,0|z,1) & =\int_{P(z)}^{1}\sum_{\tilde{e}:y(0,1)=1}q(\tilde{e}|u)du=\int_{P(z)}^{1}q(3,4,7,8,11,12,15,16|u)du,\\
p(1,0|z,0) & =\int_{P(z)}^{1}\sum_{\tilde{e}:y(0,0)=1}q(\tilde{e}|u)du=\int_{P(z)}^{1}q(2,4,6,8,10,12,14,16|u)du.
\end{align*}
Define the operator
\begin{align*}
T_{z}^{d}q^{\tilde{e}} & \equiv\int_{\mathcal{U}_{z}^{d}}q(\tilde{e}|u)du.
\end{align*}
Then, for the r.h.s. $(p_{11|z1},p_{11|z0},p_{10|z1},p_{10|z0})'$
of the constraints in \eqref{eq:constr2} that correspond to $Z=z$,
the corresponding l.h.s. is
\begin{align*}
 & \left(\begin{array}{c}
\int_{0}^{P(z)}q(9,...,16|u)du\\
\int_{0}^{P(z)}q(5,...,8,13,...,16|u)du\\
\int_{P(z)}^{1}q(3,4,7,8,11,12,15,16|u)du\\
\int_{P(z)}^{1}q(2,4,6,8,10,12,14,16|u)du
\end{array}\right)\\
 & =\left(\begin{array}{cccccccccccccccc}
0 & 0 & 0 & 0 & 0 & 0 & 0 & 0 & T_{z}^{1} & T_{z}^{1} & T_{z}^{1} & T_{z}^{1} & T_{z}^{1} & T_{z}^{1} & T_{z}^{1} & T_{z}^{1}\\
0 & 0 & 0 & 0 & T_{z}^{1} & T_{z}^{1} & T_{z}^{1} & T_{z}^{1} & 0 & 0 & 0 & 0 & T_{z}^{1} & T_{z}^{1} & T_{z}^{1} & T_{z}^{1}\\
0 & 0 & T_{z}^{0} & T_{z}^{0} & 0 & 0 & T_{z}^{0} & T_{z}^{0} & 0 & 0 & T_{z}^{0} & T_{z}^{0} & 0 & 0 & T_{z}^{0} & T_{z}^{0}\\
0 & T_{z}^{0} & 0 & T_{z}^{0} & 0 & T_{z}^{0} & 0 & T_{z}^{0} & 0 & T_{z}^{0} & 0 & T_{z}^{0} & 0 & T_{z}^{0} & 0 & T_{z}^{0}
\end{array}\right)q\\
 & \equiv Tq,
\end{align*}
where $T$ is a matrix of operators implicitly defined and $q(u)\equiv(q(1|u),....,q(16|u))$.
Now for $q\in\mathcal{Q}_{K}$, define a $16K$-vector
\begin{align*}
\theta & \equiv\left(\begin{array}{c}
\theta^{1}\\
\vdots\\
\theta^{16}
\end{array}\right)
\end{align*}
where, for each $\tilde{e}\in\{1,...,16\}$, $\theta^{\tilde{e}}\equiv(\theta_{1}^{\tilde{e}},...,\theta_{K}^{\tilde{e}})'$.
Similarly, let $b(u)\equiv(b_{1}(u),...,b_{K}(u))'$. Then, we have
$q(\tilde{e}|u)=b(u)'\theta^{\tilde{e}}$. Let $H$ be a $16\times16$
diagonal matrix of 1's and 0's that imposes additional identifying
assumptions on the outcome data-generating process. In this proof,
$H$ is used to incorporate Assumption R(i). Given $H$, the constraints
in \eqref{eq:constr2} (that correspond to $Z=z$) can be written
as
\begin{align*}
THq & =\left\{ TH\otimes b'\right\} \theta=(p_{11|z1},p_{11|z0},p_{10|z1},p_{10|z0})'
\end{align*}
for $q\in\mathcal{Q}_{K}$ and $\theta\in\Theta_{K}$.

Now, we prove the claim of the theorem. Suppose the claim is not true,
i.e., the even rows are linearly dependent to odd rows in $TH$. Given
the form of $T$, which has full rank under Assumption R(ii)(a), this
linear dependence only occurs when $H$ is such that $H_{jj}=1$ for
$j\in\{1,4,13,16\}$ and $0$ otherwise. But, according to Table \ref{tab:16maps},
this implies that $\Pr[Y(d,w)\neq Y(d,w')]=0$ for all $d$ and $w\neq w'$,
which contradicts Assumption R(i). This proves the theorem for Case
(a).

Now we move to prove the theorem for Case (b), analogous to the previous
case. For every $z$, we can derive
\begin{align*}
p(1,1|z,1) & =\int_{0}^{P(z,1)}\sum_{\tilde{e}:y(1,1)=1}q(\tilde{e}|u)du=\int_{0}^{P(z,1)}q(9,...,16|u)du,\\
p(1,1|z,0) & =\int_{0}^{P(z,0)}\sum_{\tilde{e}:y(1,0)=1}q(\tilde{e}|u)du=\int_{0}^{P(z,0)}q(5,...,8,13,...,16|u)du,\\
p(1,0|z,1) & =\int_{P(z,1)}^{1}\sum_{\tilde{e}:y(0,1)=1}q(\tilde{e}|u)du=\int_{P(z,1)}^{1}q(3,4,7,8,11,12,15,16|u)du,\\
p(1,0|z,0) & =\int_{P(z,0)}^{1}\sum_{\tilde{e}:y(0,0)=1}q(\tilde{e}|u)du=\int_{P(z,0)}^{1}q(2,4,6,8,10,12,14,16|u)du.
\end{align*}
Define
\begin{align*}
T_{z,w}^{d}q^{\tilde{e}} & \equiv\int_{\mathcal{U}_{z,w}^{d}}q(\tilde{e}|u)du
\end{align*}
where $\mathcal{U}_{z,w}^{d}$ can be analogously defined. Then,
\begin{align*}
 & \left(\begin{array}{c}
\int_{0}^{P(z,w)}q(9,...,16|u)du\\
\int_{0}^{P(z,w')}q(5,...,8,13,...,16|u)du\\
\int_{P(z,w)}^{1}q(3,4,7,8,11,12,15,16|u)du\\
\int_{P(z,w')}^{1}q(2,4,6,8,10,12,14,16|u)du
\end{array}\right)\\
 & =\tiny\left(\begin{array}{cccccccccccccccc}
0 & 0 & 0 & 0 & 0 & 0 & 0 & 0 & T_{z,w}^{1} & T_{z,w}^{1} & T_{z,w}^{1} & T_{z,w}^{1} & T_{z,w}^{1} & T_{z,w}^{1} & T_{z,w}^{1} & T_{z,w}^{1}\\
0 & 0 & 0 & 0 & T_{z,w'}^{1} & T_{z,w'}^{1} & T_{z,w'}^{1} & T_{z,w'}^{1} & 0 & 0 & 0 & 0 & T_{z,w'}^{1} & T_{z,w'}^{1} & T_{z,w'}^{1} & T_{z,w'}^{1}\\
0 & 0 & T_{z,w}^{0} & T_{z,w}^{0} & 0 & 0 & T_{z,w}^{0} & T_{z,w}^{0} & 0 & 0 & T_{z,w}^{0} & T_{z,w}^{0} & 0 & 0 & T_{z,w}^{0} & T_{z,w}^{0}\\
0 & T_{z,w'}^{0} & 0 & T_{z,w'}^{0} & 0 & T_{z,w'}^{0} & 0 & T_{z,w'}^{0} & 0 & T_{z,w'}^{0} & 0 & T_{z,w'}^{0} & 0 & T_{z,w'}^{0} & 0 & T_{z,w'}^{0}
\end{array}\right)q\\
 & \equiv\tilde{T}q,
\end{align*}
where $\tilde{T}$ is a matrix of operators implicitly defined. Then,
inserting $H$, the constraint becomes
\begin{align*}
\tilde{T}Hq & =\left\{ \tilde{T}H\otimes b'\right\} \theta=(p_{11|z1},p_{11|z0},p_{10|z1},p_{10|z0})'
\end{align*}
for $q\in\mathcal{Q}_{K}$ and $\theta\in\Theta_{K}$. Then the remaining
argument is similar to the previous case by adopting the idea of Lemma
\ref{lem:full_rank_B} (and taking $B=\left\{ \tilde{T}H\otimes b'\right\} $)
to account for the change in the range of integral due to $P(z,w)$
being a function of $w$. This completes the proof for $W$. The proof
for $Z$ can be analogously done and shown in Lemma \ref{lem:full_rank_B}.
$\square$

\subsection{Proof of Lemma \ref{lem:full_rank_B}}

Suppose $f(k)$ is a constant function: $f(k)=c$ for all $k$ for
some $c$. By properties of Bernstein polynomials, it satisfies that
\begin{align*}
\int_{0}^{P(z_{1})}b_{k,K}(u)du=\sum_{i=k+1}^{K+1}\frac{b_{i,K+1}\left(P(z_{1})\right)}{K+1} & =\sum_{i=k+1}^{K+1}\left(\begin{array}{c}
K+1\\
i
\end{array}\right)\frac{P(z_{1})^{i}\left(1-P(z_{1})\right)^{K+1-i}}{K+1},\\
\int_{0}^{P(z_{2})}b_{k,K}(u)du=\sum_{i=k+1}^{K+1}\frac{b_{i,K+1}\left(P(z_{2})\right)}{K+1} & =\sum_{i=k+1}^{K+1}\left(\begin{array}{c}
K+1\\
i
\end{array}\right)\frac{P(z_{2})^{i}\left(1-P(z_{2})\right)^{K+1-i}}{K+1}.
\end{align*}
Let $k=m$ for some $m\in\{1,...,K\}$. Then, $f(m)=c$ is equivalent
to
\begin{align}
\sum_{i=m+1}^{K+1}\left(\begin{array}{c}
K+1\\
i
\end{array}\right)P(z_{1})^{i}\left(1-P(z_{1})\right)^{K+1-i} & =c\sum_{i=m+1}^{K+1}\left(\begin{array}{c}
K+1\\
i
\end{array}\right)P(z_{2})^{i}\left(1-P(z_{2})\right)^{K+1-i}.\label{eq:k_m}
\end{align}
Similarly let $k=m-1$, then $f(m-1)=c$ is equivalent to
\begin{align}
\sum_{i=m}^{K+1}\left(\begin{array}{c}
K+1\\
i
\end{array}\right)P(z_{1})^{i}\left(1-P(z_{1})\right)^{K+1-i} & =c\sum_{i=m}^{K+1}\left(\begin{array}{c}
K+1\\
i
\end{array}\right)P(z_{2})^{i}\left(1-P(z_{2})\right)^{K+1-i}.\label{eq:k_m-1}
\end{align}
By subtracting \eqref{eq:k_m-1} from \eqref{eq:k_m}, we have
\begin{align*}
\left(\begin{array}{c}
K+1\\
m
\end{array}\right)P(z_{1})^{m}\left(1-P(z_{1})\right)^{K+1-m} & =c\left(\begin{array}{c}
K+1\\
m
\end{array}\right)P(z_{2})^{m}\left(1-P(z_{2})\right)^{K+1-m}
\end{align*}
or equivalently,
\begin{align*}
c & =\frac{P(z_{1})^{m}\left(1-P(z_{1})\right)^{K+1-m}}{P(z_{2})^{m}\left(1-P(z_{2})\right)^{K+1-m}}.
\end{align*}
Because this equation holds for any $m\in\{1,...,K\}$, take $m=1$,
then
\begin{equation}
c=\frac{P(z_{1})\left(1-P(z_{1})\right)^{K}}{P(z_{2})\left(1-P(z_{2})\right)^{K}},\label{eq:k_m2}
\end{equation}
and take $m=K$, then
\begin{equation}
c=\frac{P(z_{1})^{K}\left(1-P(z_{1})\right)}{P(z_{2})^{K}(1-P(z_{2}))}.\label{eq:k_m-12}
\end{equation}
But \eqref{eq:k_m2} and \eqref{eq:k_m-12} hold if and only if $P(z_{1})=P(z_{2})$,
which is a contradiction. Analogously, we can show that $\tilde{f}(k)=\frac{\int_{P(z_{1})}^{1}b_{k,K}(u)du}{\int_{P(z_{2})}^{1}b_{k,K}(u)du}$
is not a constant function. Therefore, the coefficient matrix has
full rank.

\subsection{Proof of Theorem \ref{thm:ptws_sharp}}

For any given $\bar{u}\in[0,1]$, $\overline{\tau}(\bar{u})=\sum_{e\in\mathcal{E}:y(1)=1}q_{\bar{u}}^{*}(e|\bar{u})-\sum_{e\in\mathcal{E}:y(0)=1}q_{\bar{u}}^{*}(e|\bar{u})$
for some $q_{\bar{u}}^{*}(\cdot)\equiv\{q_{\bar{u}}^{*}(e|\cdot)\}_{e\in\mathcal{\mathcal{E}}}$
in the feasible set of the LP, \eqref{eq:upper4} and \eqref{eq:constr4}.
Therefore, $\overline{\tau}(\bar{u})=\overline{\tau}_{MTE,\bar{u}}(\bar{u})$
for $\overline{\tau}_{MTE,\bar{u}}(\bar{u})=\sum_{e\in\mathcal{E}:y(1)=1}q_{\bar{u}}^{*}(e|\bar{u})-\sum_{e\in\mathcal{E}:y(0)=1}q_{\bar{u}}^{*}(e|\bar{u})$,
which is in $\mathcal{M}$ by definition. We can have a symmetric
proof for $\underline{\tau}(\cdot)$. $\square$

\subsection{Proof of Theorem \ref{thm:unif_sharp}}

Again, by the fact that $\tau_{MTE}(\cdot)=\sum_{e\in\mathcal{E}:y(1)=1}q(e|\cdot)-\sum_{e\in\mathcal{E}:y(0)=1}q(e|\cdot)$
in general, $\overline{\tau}(u)=\sum_{e\in\mathcal{E}:y(1)=1}q^{*}(e|u)-\sum_{e\in\mathcal{E}:y(0)=1}q^{*}(e|u)$
for all $u\in[0,1]$ is equivalent to $\overline{\tau}(\cdot)$ being
contained in $\mathcal{M}$, and similarly for $\underline{\tau}(\cdot)$.
$\square$

\section{\label{sec:Additional-Numerical-Exercise}Additional Numerical Exercises}

We present several additional numerical simulation results in this
section.

\subsection{\label{subsec:More-results-on}More results on the varying cardinality
of $\mathcal{Z}$ and $\mathcal{Y}$}

As showed in Section \ref{subsec:Identifying-Power-of}, the identifying
power of statistical independence of IVs becomes increasingly salient
as the cardinality of $\mathcal{Y}$ increases. An extreme case is
when $Y$ is continuous. In this scenario, we present additional simulation
results that compare our bounds with \citet{mogstad2018using}'s.
To reduce the computational burden arising from the dimension of $Y$,
we fix the order of Bernstein approximation at 5 for both our and
\citet{mogstad2018using}'s approach. Both when the endpoints of $\mathcal{Z}$
stay fixed or vary, our bounds are significantly narrower, which is
consistent with the results in the discrete scenario; see Figure \ref{fig:conti_Y}.

\begin{figure}
\begin{centering}
\includegraphics[scale=0.32]{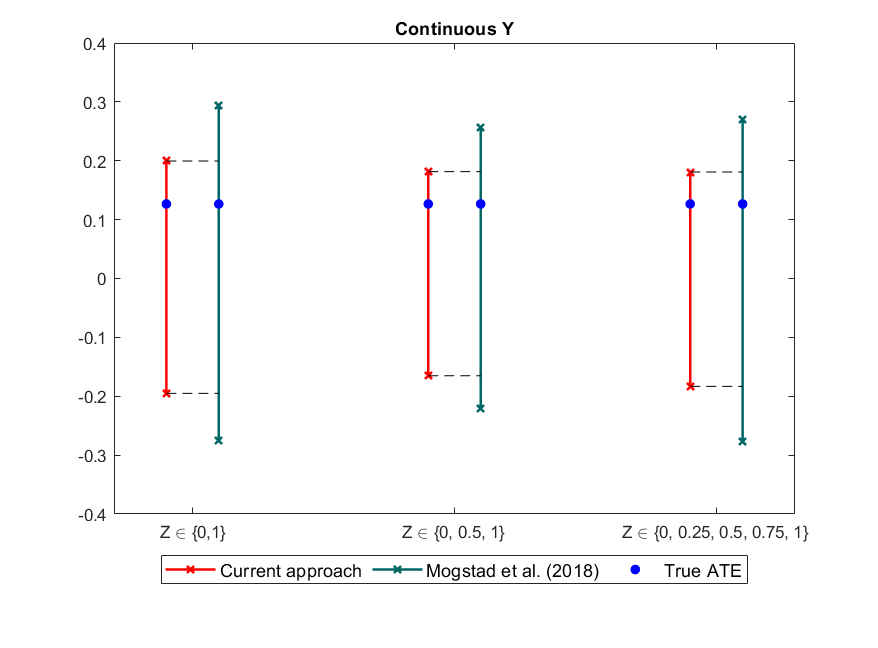}\includegraphics[scale=0.31]{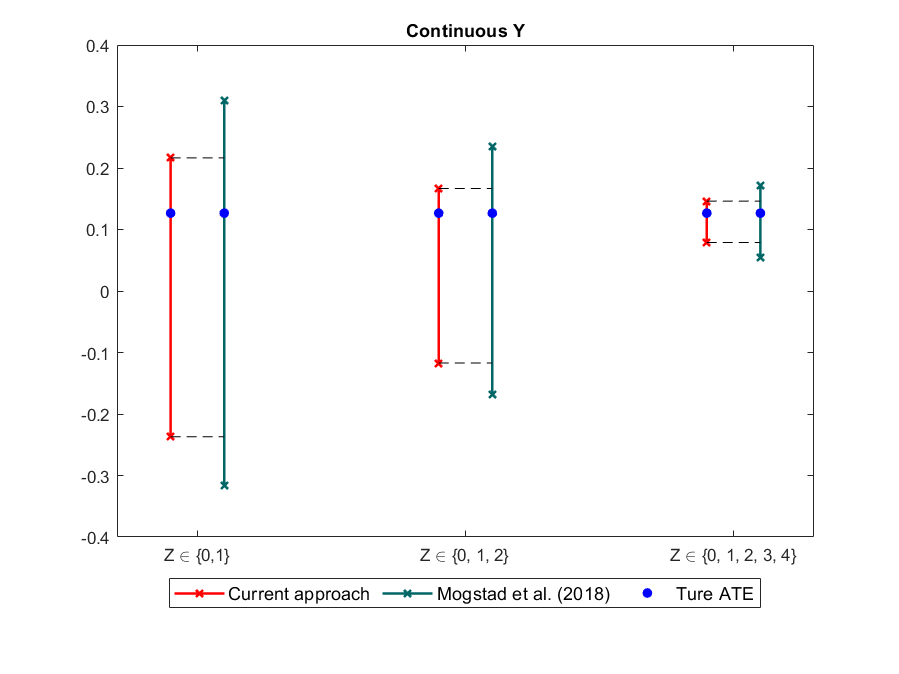}
\par\end{centering}
\caption{Full vs. Mean Independence: Bounds on ATE with Continuous $Y$}
\label{fig:conti_Y}
\end{figure}

Next, the manner in which the bounds are affected by the cardinality
of $\mathcal{Z}$ is unclear in Section \ref{subsec:Identifying-Power-of}.
This is because we fix the endpoints of $Z$. Intuitively, if the
endpoints of $\mathcal{Z}$ move further away, the bound would shrink.
Figure \ref{fig:Identification-power-from} shows that this is in
fact the case.

\begin{figure}
\begin{centering}
\includegraphics[scale=0.33]{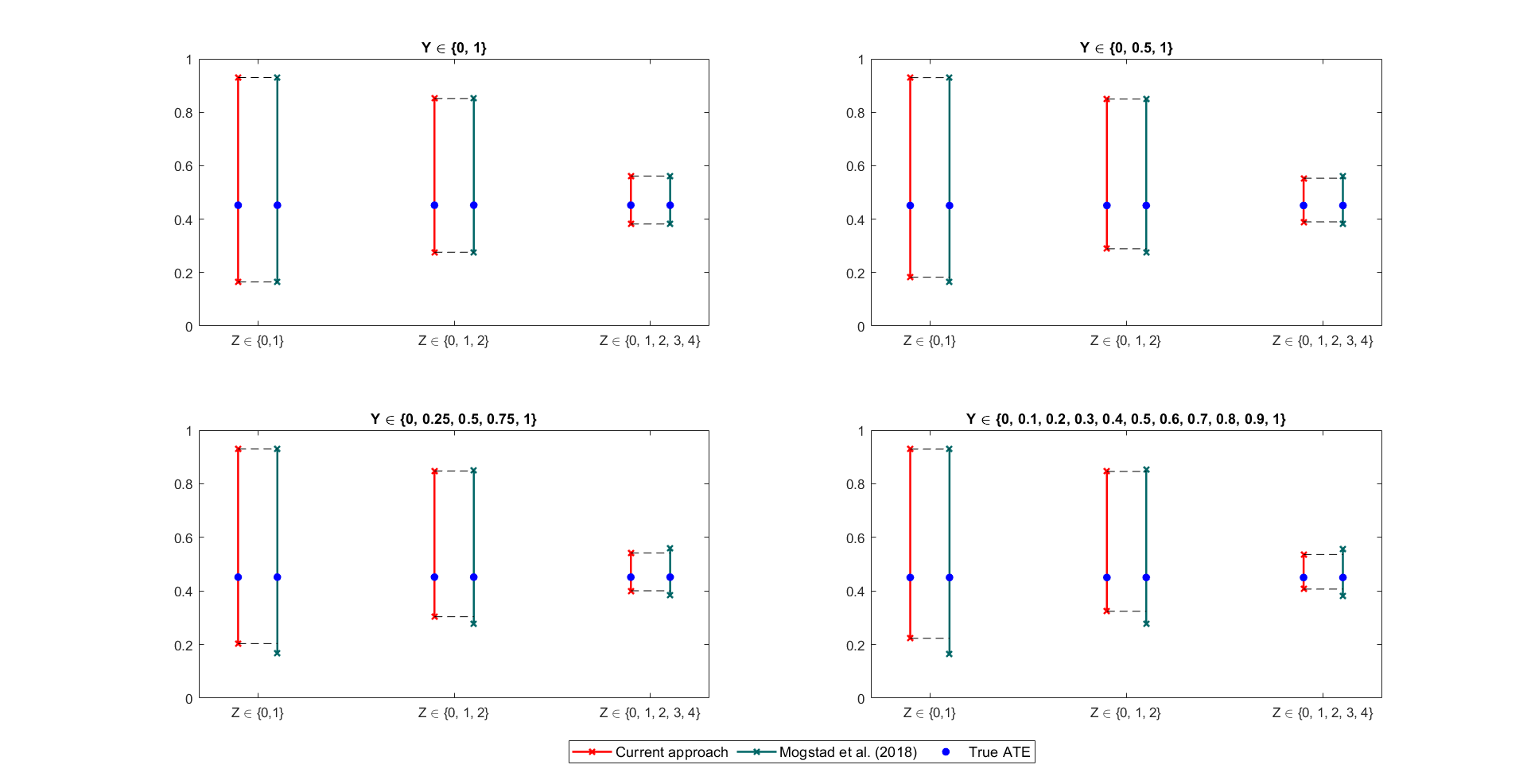}
\par\end{centering}
\caption{\label{fig:Identification-power-from}Bounds on ATE for Different
Support of $\mathcal{Z}$}
\end{figure}

\subsection{\label{subsec:Common-Exogeneous-Variable}Common Exogenous Variable}

In Section \ref{sec:Simulation}, we set $W$ to be a reverse IV in
our main DGP. It is important to note that what matters for identifying
power is the exogeneity of $W$ and \emph{not} the reverse exclusion
of $W$. To clarify this point, we compare the results derived from
a reverse IV with those from a common exogenous variable. For the
latter, we use the main DGP in Section \ref{sec:Simulation}.

Figure \ref{fig:Comparison-between-results} exhibits very similar
results between the reverse IV and the common exogenous variable.
This evidence suggests that the crucial role in improving the bounds
lies in the exogeneity of $W$, rather than its reverse exclusion. 

\begin{figure}
\begin{centering}
\includegraphics[scale=0.4]{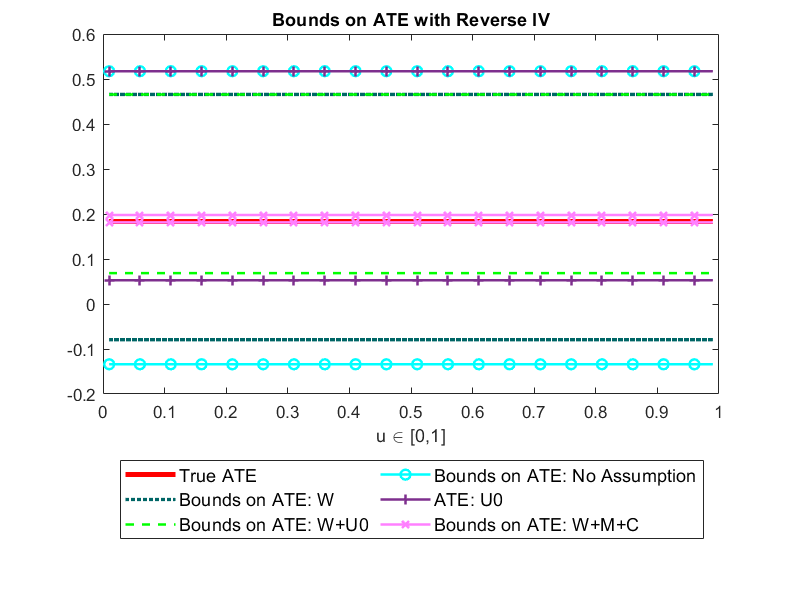}\includegraphics[scale=0.4]{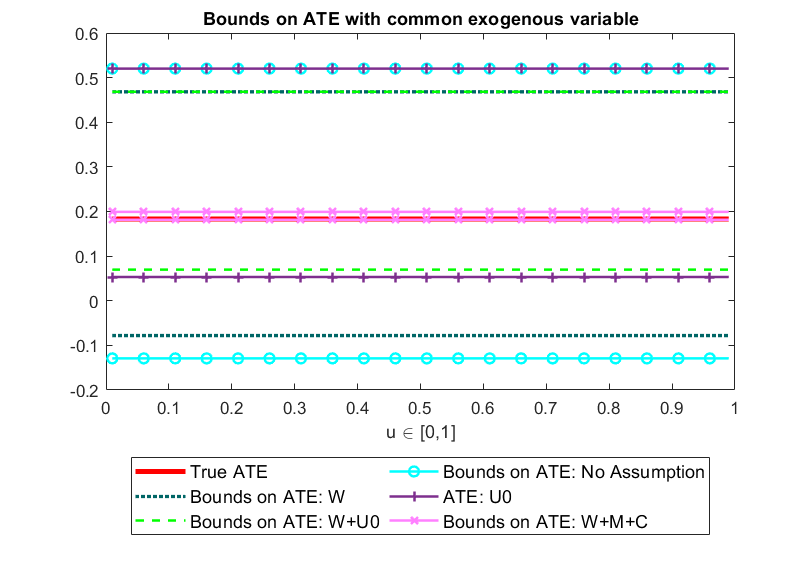}
\par\end{centering}
\caption{\label{fig:Comparison-between-results}Comparison with Common Exogenous
Variable vs. Reverse IV}
\end{figure}

\subsection{\label{subsec:Discussion-on-misspecification}Discussions on Misspecification}

For sieve approximations, implicit smoothness assumptions are inherent.
It would be crucial to evaluate the extent of misspecification using
a DGP that exhibits the lack of smoothness. To demonstrate the possibility
of misspecification within our approach and further illustrate the
choice of $K$, we select a DGP where the MTE function is defined
as $m_{1}(u)-m_{0}(u)=0.7\cdot|sin(2\pi u)|$. Note that the sine
function is difficult to approximate in general. Moreover, the MTE
has kink points, which the Bernstein approximation would fail to capture.
Under this DGP, the pointwise bound for the MTE function with varying
polynomial orders $K$ is depicted in Figure \ref{fig:Bounds-on-MTE}.
When $K$ is set as low as 5, the sieve approximation incurs severe
misspecification in that the bound only covers the true MTE over a
brief interval among the compliers. From $K=10$ onwards, the bounds
cover most of the true MTE except the middle part. While all polynomials
orders fail to capture the kink point, the magnitude of misspecification
diminishes when $K$ is large.

\begin{figure}
\begin{centering}
\includegraphics[scale=0.4]{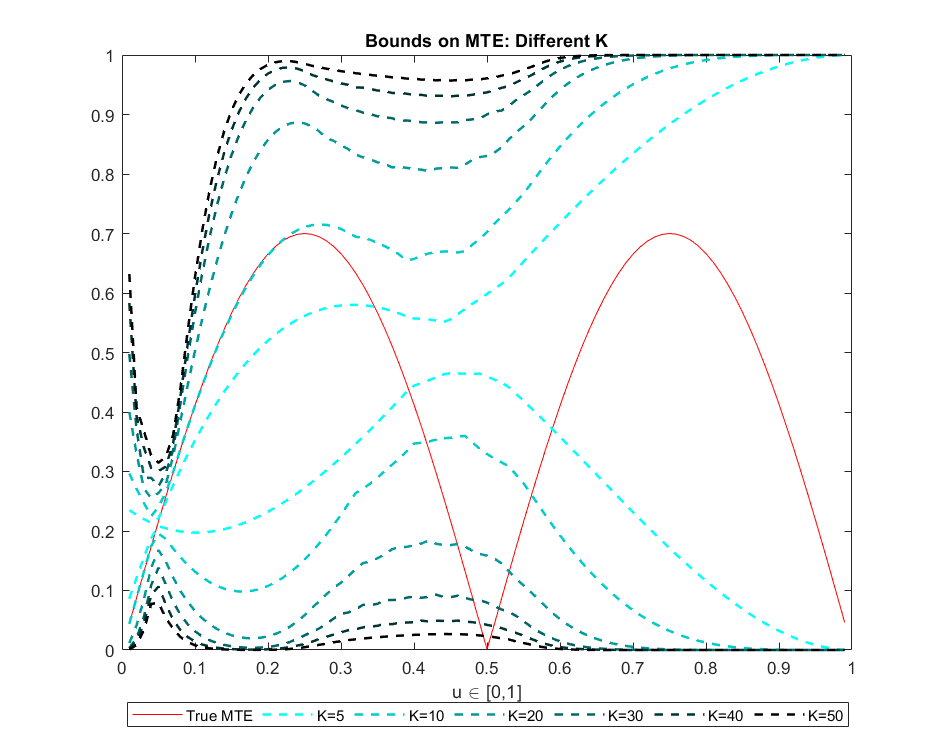}
\par\end{centering}
\caption{\label{fig:Bounds-on-MTE}Bounds on Non-Smooth MTE with Different
$K$}
\end{figure}

To more clearly compare the magnitude of misspecification across the
choices of $K$, we select 100 evenly-spaced points on $\mathcal{U}=[0,1]$,
and compute the Hausdorff distance for the set of $u$ where the MTE
falls outside the identified set. Figure \ref{fig:Misspecification-with-Different}
show a significant reduction in the magnitude of misspecification
when $K$ takes larger values.

\begin{figure}
\begin{centering}
\includegraphics[scale=0.55]{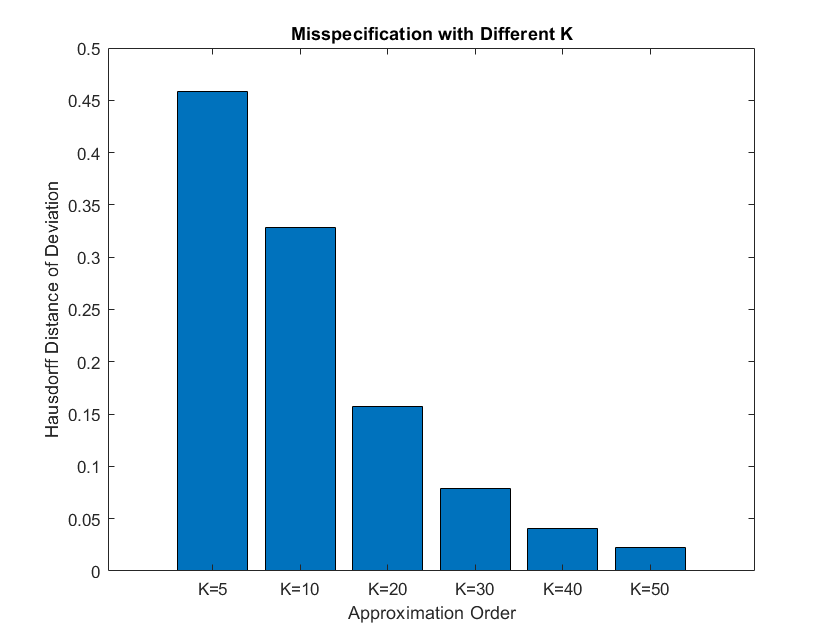}
\par\end{centering}
\caption{\label{fig:Misspecification-with-Different}Misspecification with
Different $K$}
\end{figure}

\end{appendix}

\bibliographystyle{ecta}
\bibliography{genLATE_080323}

\end{document}